\documentclass[opre,nonblindrev]{informs4Preprint} 
\RequirePackage{tgtermes}
\RequirePackage{newtxtext}
\RequirePackage{newtxmath}
\RequirePackage{endnotes}

\OneAndAHalfSpacedXII 

\usepackage{natbib}
 \bibpunct[, ]{(}{)}{,}{a}{}{,}\def\bibfont{\small}

\usepackage{fix-cm}

\usepackage{natbib}
 \bibpunct[, ]{(}{)}{,}{a}{}{,}\def\bibfont{\small}

\usepackage{bbm}
\usepackage{amssymb}
\usepackage{amsmath}
\usepackage{tcolorbox}
\usepackage{booktabs}
\DeclareRobustCommand{\mybox}[2][gray!15]{
\begin{tcolorbox}[  
        left=0pt,
        right=0pt,
        top=0pt,
        bottom=0pt,
        colback=#1,
        colframe=#1,
        enlarge left by=0mm,
        boxsep=10pt,
        arc=2pt,outer arc=2pt,
        ]
        #2
\end{tcolorbox}
}

\usepackage{thm-restate}

\usepackage{natbib}
\bibpunct[, ]{(}{)}{,}{a}{}{,}
 \def\bibfont{\small}

\usepackage[utf8]{inputenc} \usepackage[T1]{fontenc}    

\usepackage[colorlinks=true, linkcolor=blue, citecolor=blue!70!black,urlcolor=black,breaklinks=true]{hyperref}

\usepackage{comment}

\usepackage{url}            \usepackage{booktabs}       \usepackage{amsfonts}       \usepackage{nicefrac}       \usepackage{microtype}      \usepackage{mathtools, bm}

\usepackage[ruled,vlined,linesnumbered]{algorithm2e}
\usepackage{setspace}

\SetCommentSty{mycommfont}

\usepackage{multirow}
\usepackage{float}
\usepackage{xspace}

\usepackage{enumitem}
\usepackage{xfrac}
\usepackage{cleveref}

\newenvironment{numberedlemma}[1]{\renewcommand{\thetheorem}{#1}\begin{lemma}}{\end{lemma}\addtocounter{theorem}{-1}}

\newenvironment{myprocedure}[1][htb]
{  \begin{algorithm}[#1]}{\end{algorithm}}

\newcommand{\xhdr}[1]{\vspace{5pt}\noindent{\bf {#1.}}}
\newcommand{\xhdrQ}[1]{\vspace{5pt}\noindent{\bf {#1?}}}

\newcounter{relctr} \everydisplay\expandafter{\the\everydisplay\setcounter{relctr}{0}}  
\AtBeginDocument{}

\newlength\myindent
\newcommand{\signalSpace}{\Sigma}
\newcommand{\stateSpace}{\Theta}
\newcommand{\state}{\theta}

\newcommand{\signal}{\sigma}

\newcommand{\RUtility}{\rho}

\newcommand{\SUtility}{\xi}
\newcommand{\SExpUtility}{U}

\newcommand{\posteriorState}{\mu}

\newcommand{\costState}{\omega}

\newcommand{\cc}[1]{\ensuremath{\mathsf{#1}}}

\newcommand{\primed}{^{\dagger}}

\newcommand{\True}{\cc{True}}
\newcommand{\False}{\cc{False}}

\newcommand{\poly}{\cc{poly}}

\newcommand{\MemberLP}{\cc{MembershipLP}}
\newcommand{\CheckPersu}{\texttt{CheckIC}}

\newcommand{\Unif}{\cc{Unif}}

\newcommand{\Sender}{platform}
\newcommand{\Receiver}{user}
\newcommand{\Receivers}{users}
\newcommand{\action}{a}
\newcommand{\totalstate}{m}
\newcommand{\scheme}{\pi}
\newcommand{\RUtilityi}{\RUtility(\state, \action)}
\newcommand{\totaltime}{T}
\newcommand{\period}{t}
\newcommand{\signalSpacei}{\signalSpace_{\period}}
\newcommand{\schemei}{\scheme_{\period}}
\newcommand{\stochastic}{\Delta}
\newcommand{\statei}{\state_{\period}}
\newcommand{\signali}{\signal_{\period}}
\newcommand{\posteriori}{\posteriorState_{\period}}

\newcommand{\realizedstatei}{i}
\newcommand{\realizedstatej}{j}
\newcommand{\prior}{\lambda}

\newcommand{\actioni}{\action_{\period}}
\newcommand{\reals}{\R}

\newcommand{\optoff}{\scheme^*}
\newcommand{\ALG}{\texttt{ALG}}

\newcommand{\Reg}[2][]{\text{\bf REG}\ifthenelse{\not\equal{}{#1}}{_{#1}}{}\!\left[{\def\givenn{\middle|}#2}\right]}

\newcommand{\polymregret}{O(\poly(\totalstate\log\totaltime))}

\newcommand{\schemeanchorpair}{\scheme^{I}}

\newcommand{\convexcombin}{f}
\newcommand{\convexcombinbf}{\mathbf{\convexcombin}}
\newcommand{\mean}{\lambda}
\newcommand{\support}{\mathcal{S}}
\newcommand{\supportlareq}{\support_{+}}
\newcommand{\supportlar}{\support_{+}}
\newcommand{\supportsma}{\support_{-}}

\newcommand{\xlareq}{x_{+}}
\newcommand{\xlar}{x_{+}}
\newcommand{\xsma}{x_{-}}
\newcommand{\xlari}{\xlar^{(i)}}
\newcommand{\xsmai}{\xsma^{(j)}}
\newcommand{\convexcombini}{\convexcombin_{ij}}

\newcommand{\problar}{q_{+}}
\newcommand{\probsma}{q_{-}}
\newcommand{\problari}{\problar^{(i)}}
\newcommand{\probsmai}{\probsma^{(j)}}

\newcommand{\realizedstateanchori}{\realizedstatei\doubleprimed}
\newcommand{\realizedstateanchorj}{\realizedstatej\doubleprimed}

\newcommand{\schemestartingpoint}{\scheme^{(0)}}
\newcommand{\schemestartingpointTilde}{\tilde{\scheme}^{(0)}}
\newcommand{\anchorset}{\mathcal{S}}

\newcommand{\meaningfulset}{\tilde\stateSpace}
\newcommand{\meaningfulsetopt}{\scheme\doubleprimed}

\newcommand{\anchorsetHat}{\hat{\mathcal{S}}}

\newcommand{\schemeanchorpairk}{\scheme^{\texttt{II}}}

\newcommand{\RUtilityDiff}{\delta}
\newcommand{\costStatei}{\RUtilityDiff(\realizedstatei)\prior(\realizedstatei)}
\newcommand{\costStatej}{\RUtilityDiff(\realizedstatej)\prior(\realizedstatej)}
\newcommand{\optoffi}{\optoff(\realizedstatei)}

\newcommand{\CRP}{\texttt{ConRP}}
\newcommand{\LPRP}{\texttt{LP-RP}}

\newcommand{\history}{h}
\newcommand{\historyaction}{h_{a}}
\newcommand{\histories}{\mathcal{H}}
\newcommand{\historiesaction}{\mathcal{H}_a}
\newcommand{\schemes}{\Pi}
\newcommand{\optval}{\val^*}

\newcommand{\val}{v}

\newcommand{\posterior}{\posteriorState}

\newcommand{\convexcombinhat}{\hat\convexcombin}
\newcommand{\convexcombinhati}{\convexcombinhat_{ij}}
\newcommand{\natures}{\mathbb{N}}

\newcommand{\xlarn}{\xlar^{(n_1)}}
\newcommand{\xsman}{\xsma^{(n_2)}}
\newcommand{\problarn}{\problar^{(n_1)}}
\newcommand{\probsman}{\probsma^{(n_2)}}

\newcommand{\problartilden}{\tilde{q}_{+}^{(n_1)}}

\newcommand{\IC}{\cc{IC}}
\newcommand{\netSUtility}{\SUtility^{\cc{net}}}
\newcommand{\newnetSUtility}{\hat\SUtility^{\cc{net}}}

\newcommand{\context}{\boldsymbol{x}}
\newcommand{\contextSpace}{\mathcal{X}}
\newcommand{\unknownVec}{\boldsymbol{\omega}}
\newcommand{\videoNum}{n}

\newcommand{\squishlist}{
   \begin{list}{{{\small{$\vullet$}}}}
    { \setlength{\itemsep}{3pt}      \setlength{\parsep}{1pt}
      \setlength{\topsep}{1pt}       \setlength{\partopsep}{0pt}
     \setlength{\leftmargin}{1em} \setlength{\labelwidth}{1em}
      \setlength{\labelsep}{0.5em} } }
\newcommand{\squishend}{  \end{list}  }

\newcommand{\N}{\mathbb{N}}
\newcommand{\R}{\mathbb{R}}

\newcommand{\cA}{\mathcal{A}}

\newcommand{\supp}{\texttt{supp}}

\newcommand{\prob}[2][]{\text{Pr}\ifthenelse{\not\equal{}{#1}}{_{#1}}{}\!\left[{\def\givenn{\middle|}#2}\right]}
\newcommand{\expect}[2][]{\mathbb{E}\ifthenelse{\not\equal{}{#1}}{_{#1}}{}\!\left[{\def\givenn{\middle|}#2}\right]}

\newcommand{\indicator}[2][]{\mathbbm{1}\ifthenelse{\not\equal{}{#1}}{_{#1}}{}\!\left[{\def\givenn{\middle|}#2}\right]}

\newcommand{\condition}{\,\mid\,}

\newcommand{\instance}{I}
\newcommand{\optoffPrimed}{\scheme^{*\dagger}}
\newcommand{\eps}{\epsilon}
\newcommand{\doubleprimed}{^\ddagger}
\newcommand{\price}{p}
\newcommand{\pricei}{\price_t}
\newcommand{\ked}{^{(k)}}

\newif\ifarxiv
\newif\ifec
\newif\ifinforms
\arxivtrue
\ecfalse
\informstrue

\usepackage{titlesec}
\titleformat{\subsubsection}
  {\normalfont\normalsize\bfseries}{\thesubsubsection}{1em}{}
\titlespacing*{\subsubsection}
  {0pt}{6pt plus 3pt minus 2pt}{1.5ex plus .2ex}

\makeatletter
\DeclareRobustCommand{\qed}{\ifmmode \else \leavevmode\unskip\penalty9999 \hbox{}\nobreak\hfill
  \fi
  \quad\hbox{\qedsymbol}}
\newcommand{\qedsymbol}{\openbox}
\renewenvironment{proof}[1][\proofname]{\par
  \normalfont
  \topsep6\p@\@plus6\p@ \trivlist
  \item[\hskip\labelsep\itshape
    #1.]\ignorespaces
}{\qed\endtrivlist
}
\newcommand{\proofname}{Proof}
\makeatother

\EquationsNumberedThrough    

\TheoremsNumberedThrough     \ECRepeatTheorems  

\MANUSCRIPTNO{}

\begin{document}

\RUNAUTHOR{Feng, Tang and Xu}

\RUNTITLE{No-Regret Bayesian Recommendation to Homogeneous Users}

\TITLE{No-regret Bayesian Recommendation to Homogeneous Users}

\ARTICLEAUTHORS{\AUTHOR{Yiding Feng}
\AFF{Hong Kong University of Science and Technology, \EMAIL{ydfeng@ust.hk}}
\AUTHOR{Wei Tang}
\AFF{Chinese University of Hong Kong, \EMAIL{weitang@cuhk.edu.hk}}
\AUTHOR{Haifeng Xu}
\AFF{University of Chicago, \EMAIL{haifengxu@uchicago.edu}}
}

\ABSTRACT{

We introduce and study the online Bayesian recommendation problem for a recommender system platform. The platform has the privilege to privately observe a utility-relevant \emph{state} of a product at each round and uses this information to make online recommendations to a stream of myopic users. This paradigm is common in a wide range of scenarios in the current Internet economy. The platform commits to an online recommendation policy that utilizes her information advantage on the product state to persuade self-interested \Receivers\ to follow the recommendation. Since the platform does not know users' preferences or beliefs in advance, we study the platform's online learning problem of designing an adaptive recommendation policy to persuade users while gradually learning users' preferences and beliefs en route.

Specifically, we aim to design online learning policies with no \emph{Stackelberg regret} for the platform, i.e., against the optimal benchmark policy in hindsight under the assumption that users will correspondingly adapt their responses to the benchmark policy. Our first result is an online policy that achieves double logarithmic regret dependence on the number of rounds. We also present an information-theoretic lower bound showing that no adaptive online policy can achieve regret with better dependency on the number of rounds. Finally, by formulating the platform's problem as optimizing a linear program with membership oracle access, we present our second online recommendation policy that achieves regret with polynomial dependence on the number of states but logarithmic dependence on the number of rounds.

 }


\maketitle

\section{Introduction}
\label{sec:intro}

Thanks to the rapid growth of modern technology, online platforms have become a major component of today's economy. By the end of 2021, at least 30 social platforms had over 100 million active users, and seven of them surpassed 1 billion users. According to a recent report by \citet{RAM}, the global networking platforms market was valued at 192 billion U.S.\ dollars in 2019 and is projected to reach 940 billion U.S.\ dollars by 2026.
This burgeoning economic sector raises numerous complex algorithmic challenges.

Within these platforms, recommended products or services are often tagged with labels signaling varying degrees of their relevance or quality to users. For example, music streaming platforms like Spotify use curated playlists with titles such as ``Discover Weekly'' and ``Release Radar'' to highlight the novelty or relevance of the recommended tracks. E-commerce platforms like Amazon use labels such as ``Top Picks'', ``Best Sellers'' and ``Customers Also Bought'' to categorize their product recommendations, each indicating the popularity level of the items. From the platforms' perspective, these labels act as informational signals that guide user behavior and decision-making. 
However, in practice, the platform may not fully know users' preferences about the relevance or their prior beliefs about the quality of the recommended products. Below, we provide two concrete motivating applications.

\begin{example}[Recommendation in real-time short-video platforms]
\label{ex:short video}
A prominent example of recommendations on social platforms, which serves as a major motivating application for this work, is real-time short-video recommendation systems, such as TikTok, Instagram Reels, and YouTube Shorts, as well as e-commerce platforms like Amazon Live and Taobao Live.
Taking TikTok as an example, its signature feature is the ``For You'' feed, a stream of videos recommended to users in real-time. In this application, a user repeatedly opens the app throughout the day and spends a significant amount of time on the platform.
The user aims to discover high-quality content, while the platform itself may have other objectives, such as maximizing user engagement or boosting content exposure -- goals that may not always fully align with users' preferences.
The platform evaluates videos based on quality metrics (e.g., through historical user feedback or other collected data) or characteristics like captions, sounds, and hashtags. Users, on the other hand, form prior beliefs about the quality of videos based on their past experiences with the platform or external information sources. Here, video quality refers to how engaging, interesting, or relevant users perceive the content to be.
However, the platform often lacks knowledge of users' interests or prior beliefs, especially for new users.
Despite this lack of knowledge, the platform must decide whether to recommend the video to the users. 
If the platform had perfect knowledge of users' preferences, it could reliably make recommendations that optimize its own objectives while keeping users satisfied 
(e.g., by balancing high-quality videos with less engaging ones).
However, given the lack of such knowledge, the platform must learn to make such recommendations over time.
\end{example}
\begin{example}[Recommendation in dating apps]
\label{ex:dating apps}
Online dating apps (e.g., Tinder, Bumble, Coffee Meets Bagel) face similar challenges. Users on these platforms spend a significant amount of time every day on swiping and checking profiles, 
with the goal of building meaningful connections, while the platform itself may have additional goals, 
such as increasing revenue by giving extra visibility to paying subscribers, promoting new or underrepresented users to keep both sides of the market active, or favoring profiles that are predicted to generate high long-run engagement. These objectives can lead the platform to promote certain profiles for visibility even when they are not the user's top matches, so the platform's incentives are not perfectly aligned with the users' preferences.
The platform can observe various characteristics of the suggested matches (e.g., detailed profiles and past interactions with other users). Each user, in turn, may have prior beliefs about the quality or compatibility of potential matches, based on their experiences with the app or external sources. However, the user's preference and prior belief are often unknown to the platform, particularly for new users or those with limited interactions.
Despite this uncertainty, the platform must decide which profiles to recommend. If the platform had full knowledge of users' preferences and beliefs, it could consistently suggest matches that balance user satisfaction with its own goals. However, without such knowledge, the platform must gradually learn how to make better match suggestions over time.
\end{example}

Motivated by  applications like the above, we introduce and study the \emph{online Bayesian recommendation} problem, 
where a recommendation platform sequentially interacts with a homogeneous population of users, each sharing identical private preferences and beliefs.
Here we describe the problem in the context of video recommendation. 
At each time, a video is displayed by the platform to an incoming user.\footnote{In practice, platforms make joint decisions on which video to display and how to recommend. This work decouples them and focuses on the recommendation problem by assuming that the display decision is made exogenously. In \Cref{sec_generalization}, we discuss how our results could be generalized if the platform can also choose which video to display.}
To capture the uncertain characteristics of the video, we study a \emph{Bayesian} model.
In this model, the payoff-relevant characteristics of a video are represented by a (random) \emph{state} of the video. This state comprehensively encodes all critical payoff-relevant information, features, and characteristics of the video -- such as content relevance, video quality, engagement metrics, emotional impact, etc. These elements are pivotal in influencing a user's experience and satisfaction when they choose to watch a video.
The platform and user each have their own preferences over the video states, which are captured by their utility functions, respectively.  We assume a natural \emph{information asymmetry} between the platform and users --- only the platform can privately observe the realized state of each video, whereas all users only have a prior belief about the video state. Notably, the platform also has its own prior belief over the video state, 
which may differ from the users’ beliefs since they may form such beliefs from completely different sources. 
The platform designs a recommendation policy 
that makes different ``levels'' of recommendation
based on her private information about the video, i.e., its realized state.
After observing the recommendation level,
together with his initial belief,
the user forms a posterior belief about the video 
and decides either to watch this video or skip
it.

In the idealized situation where the platform knew both the users' preferences and prior beliefs, this sequential Bayesian recommendation problem reduces to a standard Bayesian persuasion problem and thus can be solved by   linear programming \citep{KG-11,AC-16,DX-19}.  
This paper, however,  addresses the more realistic yet challenging situation in which the platform does not know users' preferences nor users' prior beliefs.  
Therefore, the platform needs to adaptively update recommendation strategies based on the past users' behavior, so as to maximize her own accumulated utility. 
To formalize the platform's goal, we use the \emph{Stackelberg regret} to evaluate the performance of the platform's recommendation policy, 
and the goal of this paper is to design online policies to achieve no Stackelberg regret for the platform.\footnote{The Stackelberg regret is a regret notion recently developed for strategic settings   \citep{dong2018strategic,chen2020learning}, which compares to the optimal policy in hindsight, assuming users will correspondingly adapt their behaviors to the benchmark policy (thus the ``Stackelberg'' in its name). While previous works  demonstrated the difficulty of obtaining sublinear Stackelberg regret in online classification problems \citep{chen2020learning}, we surprisingly show that our problem admits efficient online learning algorithms with Stackelberg regret that only has logarithmic dependence  
on the number of rounds $\totaltime$.}

\subsection{Our Contributions and Techniques}
\label{subsec:contribution}

In this work, we introduce a novel online Bayesian recommendation framework that addresses the challenge of making recommendations to arriving users with unknown yet homogeneous preferences, under information asymmetry between the  platform and  users. 
Below, we overview our results and contributions for the case with only unknown user preferences but assuming users share the same common prior with the platform.\footnote{The situations with unknown different user priors are mathematically equivalent, and their results are provided in \Cref{sec_generalization}.}

When the platform has complete knowledge of users' preferences, the platform's problem can be viewed as a classic Bayesian persuasion problem with binary action \citep{KG-11}.
The optimal recommendation policies in all rounds are identical and can be solved independently as solving the optimal signaling scheme (that maps the realized state to a possibly random signal) in the Bayesian persuasion problem.
The correspondence between the terminology in recommendation systems and in Bayesian persuasion is outlined in \Cref{table:correspondence table}.
By the revelation principle, the optimal signaling scheme in hindsight is a  \emph{direct} signaling scheme with a binary recommendation level---that is, a binary signal where each signal corresponds directly to a specific action recommendation.
In particular, the optimal signaling scheme associates each state with the users' {\em preference difference}, which represents how much the user prefers watching the video over  not watching the video given this particular state.
Then, the optimal signaling scheme specifies an 
{\em order} over all states based on the users' preference differences  as well as a {\em threshold state} so that it recommends every state above the threshold state in this order.
The threshold state is selected to ensure that whenever the signaling scheme recommends a video, the user is indifferent between watching and skipping it.

When the platform lacks knowledge of users' preferences, it must use adaptive signaling schemes to learn the correct order of the states and the threshold state to achieve optimal long-term revenue.
To understand and solve the platform's problem, we focus on two natural scenarios:  
(1) known ordinal preference -- the order of the users' preference difference is known to the platform, but the exact differences are unknown;\footnote{In practice, the online platform may infer the user's ordinal preference---that is, the relative ranking of videos based on the perceived interest or relevance, from offline  data like past interactions. However, the platform may find it hard to learn the degree to which a user enjoys a particular video---that is, how much the user prefers watching this video than not watching the video.} 
and (2) unknown ordinal preference -- the order of the users' preference difference is unknown to the platform. 
We summarize our results in \Cref{table}.

\begin{table}[!ht]
\renewcommand{\arraystretch}{1.2}
\centering
\begin{tabular}{cc|c|c}
\hline
\multicolumn{2}{c|}{} & Upper bound & Lower bound \\ \hline
\multicolumn{2}{c|}{Known ordinal preference} & $O(\log\log T) ^{*}$ {\scriptsize \text{[\Cref{thm:loglog T regret upper bound}]}} & \multirow{3}{*}{\parbox{2.5cm}{\centering \vspace{30pt}$\Omega(\log\log T)^{*}$\\ {\scriptsize \text{[\Cref{thm:loglog T regret lower bound}]}}}} \\ \cline{1-3}
\multirow{2}{*}{\parbox{3cm}{\centering \vspace{22pt} Unknown\\ ordinal preference}} & \multicolumn{1}{|l|}{\hspace{5pt} Affine preference{$^\dagger$}} & $O(\log\log T)^{*}$ {\scriptsize \text{[\Cref{prop:ordinal preference}]}} & \\ \cline{2-3}
& \multicolumn{1}{|c|}{Arbitrary preference} & {\parbox{6.1cm}{\centering \vspace{5pt}$O(m 2^{m-1}\log\log T \wedge m^6 \log^{O(1)} (mT))$ \\\vspace{-10pt} {\scriptsize \text{[\Cref{thm:loglog T regret upper bound unknown order}, \Cref{thm:log T}]}}\vspace{5pt}}} & \\[0.3cm] \hline
\end{tabular}
\vspace{5pt}
\caption{Regret bounds of the online Bayesian recommendation problem. Here $m$ is the number of states, and $T$ is the time horizon. (*): These bounds do not depend on the number of states $m$. ($\dagger$): the user's utility is an affine function over the state.}
\label{table}
\vspace{-10pt}
\end{table}

Before diving into the detailed discussion of our results, we first highlight one crucial feature in our model---the feedback to the \Sender\ is {\em limited} and {\em probabilistic}.
This feedback structure is a significant challenge that algorithms with low regret must overcome or bypass, setting our model apart from other classic models in the online learning literature.
Specifically,  in our model, a signaling scheme maps each video state to a possibly random signal; therefore, the platform only observes the user's response to this realized signal and not to other signals. 
The feedback is also probabilistic, as the realized signal
depends on the realized state, which is drawn from an exogenous prior distribution.
Consequently, the platform may incur substantial regret in order to learn the user's response to a specific signal realization or their preference for a specific state.

\xhdr{Known ordinal preference}
For the known ordinal preference scenario, it suffices for the platform to learn which state is the threshold state in the optimal signaling scheme, and how to recommend when this threshold state is realized. 
For this scenario, we show that there exists 
a \texttt{Con}servative \texttt{R}ecommendation \texttt{P}olicy, henceforth \CRP\ (\Cref{alg:loglog T}), which achieves the double-logarithmic $O(\log\log\totaltime)$ regret, as stated in our first main result below  (see Theorem \ref{thm:loglog T regret upper bound} for the formal result).
\begin{repeattheorem}[Theorem (Informal)]
\CRP\ achieves $O(\log\log \totaltime)$ regret.
\end{repeattheorem}

The key intuition behind the \CRP\ is an ``unbalanced'' utility structure, in which the platform has zero payoff under overoptimistic recommendations (e.g., a predicted high-user-value video ends up not being watched) but only slightly sub-optimal payoff under overpessimistic recommendations.
This ``unbalanced'' utility structure shares similarity to the seller's payoff (as a function of the charged price) in a classic dynamic pricing problem with unknown yet fixed buyer valuation \citep{KL-03}. 
Inspired by that literature,   algorithm {\CRP} uses a ``conservative'' binary search to identify the threshold state and determine how to recommend on this threshold state. 
What sets our method apart from previous studies are the additional careful treatments to handle the aforementioned challenge on the limited and probabilistic feedback structure in our problem.

\xhdr{Unknown ordinal preference}
For the unknown ordinal preference scenario,  the order and the threshold state specified in the optimal signaling scheme remain unknown. 
Due to the limited and probabilistic feedback feature, designing an online policy to pin down this order with logarithmic regret may appear impossible at first glance. 
However, we show that when the users' preferences are affinely dependent across the states (in this case, the user's expected utility can be uniquely determined by the mean of his posterior belief),
one can still achieve regret with double logarithmic 
dependence on the number of rounds $\totaltime$, and moreover,
this regret is independent of the number of states, 
and also holds even for a continuous state space. 
On the other hand, for arbitrary users' preferences, a modified 
\CRP, which enumerates over all possible orders and prunes out the bad orders in the process,  
can lead to regret with still double logarithmic 
dependence on the number of rounds $\totaltime$ but with an exponential 
dependence on the number of states $\totalstate$. 
The informal statement of these results is as follows (see \Cref{prop:ordinal preference}  and \Cref{thm:loglog T regret upper bound unknown order}
for the formal results).
\begin{repeattheorem}[Proposition (Informal)]
For affine preferences, a modified \CRP\ has expected regret $O(\log\log \totaltime)$; for arbitrary preferences, 
a modified \CRP\ has expected regret $O(m 2^{m-1}\cdot \log\log \totaltime)$.
\end{repeattheorem}
A caveat of the above results is that the regret exponentially depends on the number of states $\totalstate$.
For a wide range of applications, it is reasonable to focus on problem instances with small $\totalstate$.\footnote{In our recommendation problem, two videos should be considered as having different states if (a) the platform has enough information to distinguish them, and (b) the user's utility for watching them is different.}
To shed light on instances with large $\totalstate$, we propose another policy, a \texttt{L}inear \texttt{P}rogram-based \texttt{R}ecommendation \texttt{P}olicy, henceforth \LPRP\ (\Cref{alg:log T}), whose regret polynomially depends on $\totalstate$ but logarithmically on $\totaltime$. 
We design \LPRP\ by formulating the problem as optimizing a linear program with membership oracle access. Specifically, the unique form of user feedback in our setting can be exploited to construct a membership oracle for the platform's linear program  problem.
\begin{repeattheorem}[Theorem (Informal)] 
\LPRP\ achieves $\polymregret$ regret.
\end{repeattheorem}
We note that similar ideas of formulating learning problems as optimizing linear programs have also been applied to other online learning problems
such as contextual dynamic pricing \citep[e.g.,][]{LS-18} and
security game \citep[e.g.,][]{BHP-14}.
However, our \LPRP\ requires additional special treatment to overcome the issue of probabilistic signals.  
Moreover, compared with using the separation oracle as in  \citep{LS-18,BHP-14},
our problem of linear optimization with membership oracle access
is considerably harder. For instance, one key technical hurdle, which does not appear in previous works but our \LPRP\ has to overcome, is to construct
an \emph{interior point} inside the feasible region. 

\xhdr{Lower bound}
Similar to the optimal signaling scheme, \CRP\ and \LPRP\ only use direct signaling schemes with \emph{binary} signals.
Such direct signaling schemes are prevalent in many real-world applications such as ``For You'' in TikTok.
However, when the platform does not know and has to learn users' preferences, the revelation principle does not necessarily hold --- i.e.,  it is unclear whether restricting to direct signaling schemes with binary signals is still without loss of generality during learning.  
Our third main result provides an affirmative answer, showing that introducing more signals cannot improve the regret dependence on $\totaltime$ for all scenarios we mentioned before.

\begin{repeattheorem}[Theorem (Informal)] 
No online policy can achieve a regret better than 
$\Omega(\log\log \totaltime)$ even for
problem instances with 
binary state.
\end{repeattheorem}

\xhdr{Simulations}
In \Cref{subsec:simulation},
we also provide 
numerical experiments to evaluate the empirical
performance of our proposed algorithm and highlight some of its salient features. In particular, 
we evaluate the performance of our proposed 
algorithm \CRP, and compare its performance with several 
benchmarks, including the benchmarks that use simple 
searching strategies to find out the optimal signaling scheme without considering the unique structure of our problem, and also the benchmarks that use simple signaling schemes. 
We observe that our algorithm significantly outperforms all of these benchmarks.
The results not only 
demonstrate the benefits of using partial 
information revealing in the platform's problem, but also show the efficiency of our algorithm.

\subsection{Related Work}

Our work connects to several strands 
of existing literature.
First, when the \Receiver's preference is known
and they share the same prior belief with the \Sender, 
the one-shot instantiation 
of our problem exactly follows the formulation
of the canonical Bayesian persuasion problem \citep{KG-11}.
Bayesian persuasion concerns the problem 
where an informed sender (i.e., \Sender) 
designs an information structure (i.e., signaling scheme)
to influence the behavior of a receiver (i.e., \Receiver).
There is a growing literature, including our work, 
on studying the relaxation of 
one fundamental assumption 
in the Bayesian persuasion model -- 
the sender perfectly knows the receiver's preference
and their prior belief.
There are generally two approaches to deal with the uncertainty from the sender's perspective:
the robust approach \citep{DP-20,BTXZ-21,K-18,HW-21} 
which tries to design signaling schemes
that perform robustly well for all possible  receiver preferences; 
the online learning approach \citep{CCMG-20,CMCG-21,ZIX-21}
which studies the regret minimization 
when the sender repeatedly interacts with receivers.
\footnote{We refer the reader to the work by \citet{DP-20,BTXZ-21}
for a comprehensive overview 
on different methods in the robust approach.}
Our work falls into the second approach.
In particular, \citet{CCMG-20} address the 
sender interacting with 
receivers who have an unknown type.
They provide an algorithm with a regret guarantee
$O(T^{\sfrac{4}{5}})$ but it has exponential running-time 
over the number of states.
\citet{ZIX-21} study a setting where the sender
has an unknown prior distribution, 
and they require the sender to make \IC\ signaling 
schemes at each round.
They provide an algorithm with an $O(\sqrt{T})$ regret bound, 
and also demonstrate that it is tight 
whenever the receiver has five (or more) actions. 
Our work differs from the above works in many ways.
First, instead of assuming unknown types, 
our setting directly relaxes
the knowledge of the \Receiver's utilities.
Second, we do not
require the \Sender's signaling scheme to be \IC\
at each round.
Third, 
we achieve logarithmic regrets 
over the time horizon and this is possible due to the special structure of the Bayesian recommendation problem. 

Second, our work also relates to research on Bayesian 
exploration in multi-armed bandit \citep{KMP-14,MSSV-15,MSSW-21}.
In both our Bayesian recommendation and their 
Bayesian exploration, the \Sender\ utilizes her 
information advantage to persuade the \Receiver\ 
to take the desired action, and the \Receiver\  
observes the \Sender's message and forms their posterior
which will be used to pick their optimal action. 
However, in Bayesian exploration, 
the platform is learning the true state of nature, which is realized at the very beginning and never changes afterward, and is required 
to make incentive-compatible action recommendation  
at each round.
While in our setup, the platform is learning the users' preferences and beliefs, and the state 
is realized independently across the time horizon.
Additionally, 
our problem does not require 
the platform to
make incentive-compatible recommendations. 
Thus, the analysis and the technique 
of this work are quite 
different from theirs. 

Third, our setting with homogeneous users  shares  similarities with the fixed valuation in (contextual) dynamic pricing literature \citep{LRV-18, KL-03,LS-18,LLS-21},
where the logarithmic regrets are also achievable. 
Part of our analysis is also related to this line of literature.
In particular, we prove our lower bound via
a non-trivial reduction to the single-item 
dynamic pricing problem. 
Though it seems that our problem for 
multiple states shares similarities with the contextual 
dynamic pricing (e.g., we both need to learn an unknown vector: in our setting, it is the \Receiver's preference for product states, and in contextual pricing, it is the buyer's preference for product features), 
we note that there are significant differences in our 
problem structure such as the \Sender's actions 
and the probabilistic feedback from \Receiver s
(see \Cref{apx:dynamic pricing} for 
the detailed comparisons).

\section{Preliminary}
\label{sec_prelim}
\label{sec:prelim}

Motivated by the applications of 
short-video platforms,
this paper introduces and studies the   \emph{Bayesian 
recommendation} problem.

\subsection{Basic Setup}
We begin by describing a static model
and then introduce the online setup
studied in this work. 

In the static model, there are two players: 
a platform and a user.\footnote{In this paper, we use ``she'' 
to denote the platform\ 
and ``he'' to denote the user.}
The platform wants to recommend a video 
to the user.
The video
is associated with a private
state $\state$ drawn from 
a finite set $[\totalstate] 
\triangleq \{1, \ldots, \totalstate\}$ 
according to a prior distribution $\prior \in \stochastic([\totalstate])$,
which is common knowledge among both players. 
We use the notation $\state$ to denote the state as a random variable, 
and $\realizedstatei, \realizedstatej\in [\totalstate]$
as its possible realizations.
The user 
has 
a binary-action set $\cA = \{0, 1\}$
(i.e., not watch or watch),
and 
a utility function $\RUtility:[\totalstate]\times \cA
\rightarrow \reals$
mapping from the state of the video 
and his action to his utility.
The platform has the utility function
$\SUtility: [\totalstate] \times\cA \rightarrow \reals$
that depends on the user's action and the realized state.
In addition, the platform has an information advantage in the sense that the platform can observe the realized video state but the user cannot.
This advantage enables the platform
to recommend the video in different levels based on the realized state.
In particular, 
the platform can design a finite\footnote{For ease of presentation,
we restrict the signal space to be finite.
For continuous signal space $\signalSpace$, 
$\scheme(\realizedstatei, \cdot)$ can be interpreted as a probability
density function (conditional on state $\realizedstatei$) 
over the continuous signal space $\signalSpace$.
We note that this restriction is without loss as the optimal signaling scheme
characterized in \Cref{lem:optimal in hindsight characterization} indeed 
has finite signal space.}
signal space $\signalSpace$ where each signal $\signal\in\signalSpace$ represents a recommendation level for the video,
and a signaling scheme
$\scheme = \{\pi(i, \cdot)\}_{i\in[\totalstate]} $ where $\pi(i, \cdot)\in\Delta(\signalSpace), i\in[\totalstate]$ denotes the conditional probability distribution over the signals when the state $i$ is realized,
and we use
$\scheme(\realizedstatei, \signal)$ as the conditional probability of sending signal $\signal\in\signalSpace$ at state~$\realizedstatei$.

For ease of presentation, our main context will focus on a stylized setup where 
(i)
the platform
and the user
share the same prior belief $\prior$ over 
$[\totalstate]$;
(ii)
the platform
has state-independent utility function
and only benefits from the user's action $1$
(i.e., watch/click), namely, $\SUtility(\state, \action) \equiv \SUtility(\action), \forall \state\in[\totalstate]$ and $\SUtility(1) \ge \SUtility(0)$.
In \Cref{sec_generalization}, 
we illustrate how our algorithms 
and results can be easily extended to general settings where (i) users might have different unknown prior beliefs and (ii) the platform has arbitrary utility functions.
Without loss of generality, we normalize 
$\SUtility(\action) = \action$.
Namely, the platform gains 1 unit of revenue if the user
watches the video.

In this work, we consider the following repeated interaction 
between the platform and a population of users.
All users share the same utility function $\RUtility(\cdot,\cdot)$ which is \emph{unknown} to the platform.
The setting proceeds for $\totaltime$ rounds. 
For each round $\period = 1, \ldots, \totaltime$:
\begin{enumerate}
    \item The platform commits to
    a signaling scheme $\schemei$ with a 
    signal space $\signalSpacei$.

    \item A video with state $\statei \sim \prior$ is realized (only observed by the platform),
    and a signal $\signali \sim \schemei(\statei, \cdot)$
    is realized (observed by both players) subsequently according to designed signaling scheme $\schemei$.
    
    \item Upon seeing the signal $\signali$, together with the knowledge of $\scheme_t$,
    the \Receiver\ $t$ updates his belief regarding the underlying state to a {\em posterior distribution} – a.k.a.\ the posterior – denoted by $\posteriori(\signali,\cdot) \in \Delta([\totalstate])$ where 
    $\posteriori(\signali,\realizedstatei) \triangleq
    \frac{
    \prior(\realizedstatei) \schemei(\realizedstatei, \signali)
    }
    {
    \sum_{\realizedstatej\in[\totalstate]}
    \prior(\realizedstatej)
    \schemei(\realizedstatej, \signali)
    }$ is the posterior probability on state $\realizedstatei$.
    
    \item 
    With the posterior $\posteriori(\signali, \cdot)$, 
    the \Receiver\ $t$
    chooses an action $\actioni(\signali)$ that maximizes his expected utility, i.e., 
    $\actioni(\signali) = \argmax_{\action\in\cA} 
    \expect[\state\sim \posteriori(\signali, \cdot)]{\RUtilityi}$.
\item The \Sender\  then derives the utility
    $\actioni$.
\end{enumerate}
Given a signaling scheme $\scheme$, 
let $\SExpUtility(\scheme)$
denote the \Sender's one-shot expected payoff,
\begin{align*}
    \SExpUtility(\scheme) = 
    \sum\nolimits_{\realizedstatei\in[\totalstate]} 
    \sum\nolimits_{\signal\in\signalSpace}\prior(\realizedstatei)\scheme(\realizedstatei,\signal) \cdot\action(\signal)
\end{align*}
where $\signalSpace$ is the associated signal space and $\action(\signal)\in\cA$ is the agent's best response action upon seeing signal~$\signal$ from the signaling scheme~$\scheme$. The goal of the \Sender\ is to design an online policy  that constructs signaling schemes $\{\schemei\}_{\period \in [\totaltime]}$ to maximize her long-term expected utility $\sum_{\period\in[\totaltime]}
\SExpUtility(\schemei)$.
We conclude this subsection with the following two remarks.

\begin{remark}[Bayesian-rational behavior]
Upon seeing a signal realization, the user is able to form his Bayesian posterior belief,
and then make his decision by maximizing the expected utility 
based on the current belief. 
This Bayesian-rational behavior follows the common assumption adopted in Bayesian persuasion literature, and also other literature that studies signaling as a way to reveal product quality/characteristic information.\footnote{For example, in the literature of Bayesian social learning with pricing, the customers are usually assumed to be able to update their beliefs in a Bayesian manner upon seeing certain new information about the product quality \citep{IMSZ-19,SVZ-22}, and then make the purchase decision by maximizing the corresponding expected utility \citep{IMSZ-19}. 
Other literature includes signaling in queues \citep{DPR-12,LI-19}.
We believe that the inclusion of Bayesian rational users inherently adds complexity to the problem, yielding rich insights and results. 
We view incorporating realistic and relevant behavioral biases as a next step in this research direction.}
\end{remark}
\begin{remark}[Commitment power]
In the above interaction, 
the platform is assumed to have the commitment power, and the designed signaling scheme is known to the user. In practice, the platform might lack this commitment power and may not be able to change her signaling scheme daily. However, we note that users may engage with the platform over a specific duration (i.e., a time cycle), e.g., staying on the platform for a while or returning to the platform several times.
Throughout this duration, the platform may stick to the same signaling scheme, allowing the user to discern the underlying signaling scheme as well as his best response to it.
Thus, a single round in the above theoretical model 
can be equated to one practical time cycle. 
\end{remark}

\subsection{Stackelberg Regret and 
Benchmark}
\label{subsec: benchmark}

We evaluate the performance of an online policy by its \emph{Stackelberg regret}
\citep{chen2020learning} 
against the \emph{optimal policy in hindsight}.
The optimal policy in hindsight 
knows \Receiver s' utility function $\RUtility(\cdot,\cdot)$,
and maximizes the \Sender's long-term expected utility. 
Since \Receiver s are all identical, 
the optimal policy in hindsight commits to the same signaling
scheme $\optoff$
(see its characterization in
program
\ref{eq:optimal in hindsight}
and \Cref{lem:optimal in hindsight characterization}) 
for every round $\period\in[\totaltime]$.
\begin{definition}
Given the \Receiver's utility function $\RUtility$,
let $\optoff$ be the optimal signaling scheme.
The \emph{Stackelberg regret} of online policy $\ALG$
is 
\begin{align*}
    \Reg{\ALG} \triangleq 
    \displaystyle \sum\nolimits_{\period\in[\totaltime]}
    \SExpUtility(\optoff)
    -
    \expect[\scheme_1,\dots, \scheme_\totaltime]{
    \displaystyle \sum\nolimits_{\period\in[\totaltime]}
    \SExpUtility(\schemei)}
\end{align*}
where $\schemei$
is the signaling scheme committed by $\ALG$ in each round $\period\in[\totaltime]$.
\end{definition}
Different from the regret notation (e.g., external regret) in classic single-agent no-regret learning literature
\citep[cf.][]{BM-07}, 
the Stackelberg regret compares
to the optimal policy in hindsight, 
where \Receiver s have the opportunity to re-generate 
a different history by best-responding to the new 
signaling scheme $\optoff$.
In the remainder of the paper, we simplify the terminology
Stackelberg regret as regret.

To facilitate our analysis, we introduce one auxiliary variable,
$\RUtilityDiff(\realizedstatei) \triangleq \RUtility(\realizedstatei, 1) 
- \RUtility(\realizedstatei, 0)$, 
to represent how much the user prefers action $1$ over action $0$ given state $i$.
To make the problem non-trivial, 
we make the following two assumptions 
on the \Receiver's utility function
throughout this paper.
\begin{assumption}
\label{asp:receiver utility 1}
There exists at least one state $\realizedstatei\in[\totalstate]$ such that 
$\RUtilityDiff(\realizedstatei)\prior(\realizedstatei) > 0$.
\end{assumption}
\begin{assumption}
\label{asp:receiver utility 2}
The \Receiver's utility function satisfies that $\sum_{\realizedstatei\in[\totalstate]}
\RUtilityDiff(\realizedstatei)\prior(\realizedstatei) < 0$.
\end{assumption}
\Cref{asp:receiver utility 1} ensures that there is at least one state where the user prefers action 1 over action 0. Without this, the problem becomes trivial, as the \Receiver\ would always choose action 0 regardless of the signaling scheme, and any online policy would achieve zero regret.
\Cref{asp:receiver utility 2} ensures that under a no-information-revealing signaling scheme, the user's best response (i.e., the prior-best action) is to choose action 0. Without this, the platform could commit to a no-information-revealing scheme to achieve optimal utility, as the user would always choose action 1.

When the \Receiver s' utility function $\RUtility(\cdot,\cdot)$ 
is known to the \Sender, 
the optimal signaling scheme in hindsight can be characterized as follows.
By the revelation  principle  \citep{KG-11}, the optimal signaling scheme is a direct signaling scheme with binary signal space $\signalSpace = \{0, 1\} \equiv\cA$ that each signal corresponds to an action recommendation.
In particular, it can be solved by a linear program as follows,
\begin{align}
\tag{$\mathcal{P}^{\texttt{opt}}$}
\label{eq:optimal in hindsight}
    \begin{array}{llll}
      \optoff = \argmax\limits_{\scheme}   &  
      \displaystyle\sum\nolimits_{\realizedstatei\in[\totalstate]}
      \prior(\realizedstatei)\scheme(\realizedstatei,1)
      &  \text{s.t.}
      &\\
        \text{(\IC)}
         & \displaystyle\sum\nolimits_{\realizedstatei\in[\totalstate]}
\RUtilityDiff(\realizedstatei)
         \prior(\realizedstatei)
         \scheme(\realizedstatei,1)
         \geq 0
         \quad\quad
         & & \\
         &\scheme(\realizedstatei, 1)
         +
         \scheme(\realizedstatei, 0) = 1
         & \realizedstatei\in[\totalstate]
         &\\
         &\scheme(\realizedstatei, 1)
         \geq 0, 
         \scheme(\realizedstatei, 0) \geq 0
         &
         \realizedstatei\in[\totalstate]
         &
    \end{array}
\end{align}
The objective in the above program $U(\scheme) = \sum_{\realizedstatei\in[\totalstate]} \prior(\realizedstatei) \scheme(\realizedstatei, 1)$ is the platform's expected utility from an incentive compatible, direct signaling scheme $\pi$.
The first constraint is an incentive compatibility (\IC) constraint to ensure that it is the user's best response to follow the recommended action realized from the signaling scheme,
i.e., 
taking action 1 is indeed the \Receiver's optimal action given his 
posterior
when action 1 is recommended. We note that in program
\ref{eq:optimal in hindsight},
the {\IC} constraint for action 0 is omitted
since the \IC\ constraint for action $1$ and \Cref{asp:receiver utility 2} can
jointly imply the {\IC} constraint for action $0$. To see this, 
notice that the {\IC} constraint for action $0$ is:
$
\sum_{\realizedstatei\in[\totalstate]} 
\RUtilityDiff(\realizedstatei)
\prior(\realizedstatei)
\scheme(\realizedstatei,0)
< 0$. Rearranging the \IC\ constraint for action $1$, we obtain
$\sum_{\realizedstatei\in[\totalstate]} 
\RUtilityDiff(\realizedstatei)
\prior(\realizedstatei)
\scheme(\realizedstatei,0) \le \sum_{\realizedstatei\in[\totalstate]} 
\RUtilityDiff(\realizedstatei)
\prior(\realizedstatei)$, where the right-hand side is weakly smaller than zero due to \Cref{asp:receiver utility 2}.

For ease of presentation, 
with slight abuse of notation,
we use $\optoffi \triangleq \optoff(\realizedstatei, 1)$
and thus $\optoff(\realizedstatei, 0) 
\equiv 1 - \optoffi$.
Program \ref{eq:optimal in hindsight}
can be interpreted as 
a fractional knapsack problem,
where 
the budget is zero,
and each state $\realizedstatei$
corresponds to
an item
with
value $\prior(\realizedstatei)$
and (possibly negative)
cost $\RUtilityDiff(\realizedstatei)\prior(\realizedstatei)$.
Thus, its optimal solution $\optoff$
has the following characterization.
\begin{lemma}[See for example \citealp{RSV-17}]
\label{lem:optimal in hindsight characterization}
The optimal signaling scheme $\optoff$ 
in hindsight is the optimal solution of the linear program~\ref{eq:optimal in hindsight}.
There exists a threshold 
state $\realizedstatei\primed
\in[\totalstate]$ such that
(a)
for every state $\realizedstatei\not=\realizedstatei\primed$,
$\optoffi = \indicator{
\RUtilityDiff(\realizedstatei) \geq
\RUtilityDiff(\realizedstatei\primed)}$,
and 
(b) 
$\optoff(\realizedstatei\primed) = 
-
\frac{\sum_{\realizedstatei\not= \realizedstatei\primed} \RUtilityDiff(\realizedstatei)\prior(\realizedstatei) \optoffi}{
\RUtilityDiff(\realizedstatei\primed)\prior(\realizedstatei\primed)}$.
\end{lemma}
In words, 
\Cref{lem:optimal in hindsight characterization}
states that 
the signaling scheme $\optoff$ 
reveals whether 
the state is above or below\footnote{Throughout this paper, we say 
a state $\realizedstatei$ 
is above (resp.\ below) state $\realizedstatej$ if it satisfies that
$\RUtilityDiff(\realizedstatei)
\geq 
\RUtilityDiff(\realizedstatej)$ 
(resp.\ 
$\RUtilityDiff(\realizedstatei)
< 
\RUtilityDiff(\realizedstatej)$).}
a threshold
state $\realizedstatei\primed$, 
with possible randomization at
state $\realizedstatei\primed$.

The above \Cref{lem:optimal in hindsight characterization} highlights the significance of both the cardinal value and the order of user preference differences $\{\RUtilityDiff(\realizedstatei)\}$ to characterize the optimal signaling scheme in hindsight. 
As we mentioned earlier, the user preference $\RUtility(\cdot, \cdot)$
is unknown to the platform.
Consequently, the platform is unaware of both the cardinal value and the order of $\{\RUtilityDiff(\realizedstatei)\}$, 
and thus needs to learn these quantities by adaptively changing her signaling schemes. 
To address this intricate learning challenge, 
we initially study the scenario where only the cardinal value of $\{\RUtilityDiff(\realizedstatei)\}$ is unknown to the platform, while its order is known to the platform (see \Cref{sec:knwon order}).
Building on the insights from this scenario, we then discuss the more complex situation where the platform lacks knowledge of both the cardinal value and the order of 
$\{\RUtilityDiff(\realizedstatei)\}$
(see \Cref{sec:unknwon order} and \Cref{sec:log T}).

\subsection{A Useful Subroutine for Checking the Incentive Compatibility}
When the \Receiver's utility function $\RUtility(\cdot,\cdot)$
is unknown, 
the standard revelation principle fails.
As a consequence,
restricting to binary signal space
(e.g., 
\{``recommended'', 
``not recommended''\})
is 
\emph{not} without loss of generality.
Nonetheless, 
as we formally show later, 
restricting to 
the subclass of signaling schemes
with binary signal space 
does not hurt the optimal regret.
We now formally define such 
signaling schemes as follows.
\begin{definition}
A \emph{direct signaling scheme}
$\scheme$ is a mapping from states into probability distributions over actions recommended 
to the \Receiver.
\end{definition}
With slight abuse of notation, 
for every direct signaling scheme $\scheme$,
we use $\scheme(\realizedstatei) \triangleq\scheme(\realizedstatei,1)$
and thus $\scheme(\realizedstatei,0)
\equiv 1 - \scheme(\realizedstatei)$.
When facing a direct signaling scheme $\scheme$, the \Receiver\ takes the action that maximizes his expected utility given his posterior. 
We say a signaling scheme $\scheme$ is \emph{\IC} if 
the \Receiver\ takes action 1 as long as  
action 1 is recommended by the signaling scheme $\scheme$.
The proofs of \Cref{lem:persuasive costState},
\Cref{lem:persuasive check correctness} and 
\Cref{lem:persuasive check regret} 
are deferred 
to \Cref{apx:proofs in prelim}.

\begin{restatable}{lemma}{persuasiveCostState}
\label{lem:persuasive costState}
A direct signaling scheme $\scheme$
is \IC\ if and only if $\sum_{\realizedstatei\in[\totalstate]}
\costStatei\scheme(i) \geq 0$.
\end{restatable}

Before we finish the preliminary section,
we provide
Procedure~\ref{alg:persuasive check}
as a useful subroutine that will be used in 
our online 
policies.
Procedure~\ref{alg:persuasive check} takes a direct signaling scheme
as input, and determines whether 
this direct signaling scheme is \IC. 
Its correctness guarantee is given in 
\Cref{lem:persuasive check correctness}
and the regret guarantee incurred by implementing Procedure~\ref{alg:persuasive check} is given in 
\Cref{lem:persuasive check regret}.

\begin{myprocedure}
\linespread{0.8}\selectfont
\caption{$\CheckPersu(\scheme)$}
\label{algo_persuasive_check}
\label{alg:persuasive check}
\SetAlgoLined\DontPrintSemicolon
\KwIn{a direct signaling scheme $\scheme$}
\KwOut{\True/\False\ -- 
whether $\scheme$ is \IC;
or \texttt{round-exhausted} if 
there is no round left}
\While{there are rounds remaining}{
\tcc{suppose now is round $t$}
Commit to signaling scheme $\scheme$
towards \Receiver\ $t$.
\\
\textbf{if} $\signali = 1$ and $\actioni = 1$, \textbf{then}
    return \True \\ 
\textbf{else if} $\signali = 1$ and $\actioni = 0$, \textbf{then} return \False \\ 
\textbf{else if}  $\signali = 0$ and $\actioni = 1$, \textbf{then} return \False
 
\textbf{else} move to next round, i.e., $t \gets t + 1$ 
} 
return \texttt{round-exhausted}
\end{myprocedure}
\begin{restatable}{lemma}{persuasiveCheckCorrectness}
\label{lem:persuasive check correctness}
Given a direct signaling scheme $\scheme$,
Procedure~\ref{alg:persuasive check}
returns \True\ only if $\scheme$ is \IC,
and returns \False\ only if $\scheme$ is not \IC. 
\end{restatable}
We note that 
Procedure~\ref{alg:persuasive check} does not 
include the case for 
$\signali = 0$ and $\actioni = 0$ as when this case happens, the user's response (i.e., taking $\actioni = 0$) does not convey any information about whether the signaling scheme is \IC\ or not. 
In particular, when this case happens, we know that $\sum_{i\in[\totalstate]} \prior(i)\RUtilityDiff(i)(1-\scheme(i)) < 0$, and this inequality does not imply the \IC\ inequality $\sum_{\realizedstatei\in[\totalstate]}
\costStatei\scheme(i) \geq 0$ defined as in \Cref{lem:persuasive costState} due to $\sum_{i\in[\totalstate]} \prior(i)\RUtilityDiff(i) < 0$ in \Cref{asp:receiver utility 2}.
Thus, Procedure~\ref{alg:persuasive check} stops either when it encounters
one of the listed three cases (in which the user's response can be used to infer whether the committed signaling scheme is \IC\ or not) 
or when the time rounds are exhausted. 
As long as Procedure~\ref{alg:persuasive check} returns 
\True\textbackslash\False, the platform 
knows for sure whether the signaling 
scheme is \IC\ or not.

\newcommand{\issuedrounds}{Q}

\begin{restatable}{lemma}{persuasiveCheckRegret}
\label{lem_approx_guarantee}
\label{lem:persuasive check regret}
Given a direct signaling scheme $\scheme$,
the expected regret of 
Procedure~\ref{alg:persuasive check}
is at most 
$\frac{\SExpUtility(\optoff)}{
\sum_{\realizedstatei}\prior(\realizedstatei)\scheme(\realizedstatei)}
-
\indicator{\text{$\scheme$ is \IC}}$.
\end{restatable}
The intuition behind  \Cref{lem:persuasive check regret}
is as follows. 
Given a direct signaling scheme $\scheme$,
as long as its probability (i.e., the value
$\sum_{\realizedstatei}\prior(\realizedstatei)\scheme(\realizedstatei)$) 
for recommending action $1$
is a constant approximation to the optimal payoff 
$\SExpUtility(\optoff)$, then the expected regret of 
Procedure~\ref{alg:persuasive check} for 
checking its incentive compatibility can be upper bounded by this constant.
However, if this probability is small 
compared to 
$\SExpUtility(\optoff)$, then the expected incurred regret 
can be very large,
regardless of 
the incentive compatibility of $\scheme$
or the value of $\SExpUtility(\optoff)$.

\section{Conservative Recommendation Policy for Known Ordinal Preference}
\label{sec:knwon order}

\label{sec:loglog T}

In this section, we provide our first result
--
\texttt{Con}servative \texttt{R}ecommendation \texttt{P}olicy (\CRP),
which is an online policy that can achieve low regret
when the order of the user's preference differences $\{\RUtilityDiff(i)\}$ 
is known to the platform.
We upperbound its regret by $O(\log \log \totaltime)$.
Combining with the lower bound $\Omega(\log\log \totaltime)$ 
of regret for any online policy
developed in \Cref{sec:loglog T lower bound}, 
our policy \CRP\ is regret-optimal. The results of this section will be served as a building block for the algorithm design when the order of $\{\RUtilityDiff(i)\}$ is not known in advance.

Before diving into our results, 
let us highlight one of the main challenges in the 
design of a good online policy.
In our problem, 
the \Sender's feedback is \emph{limited} and \emph{probabilistic}.
Specifically, 
when signaling scheme $\schemei$ is used in round $t$
and signal $\signali \sim \schemei(\statei)$ is realized,
the platform only observes the \Receiver's action under signal $\signali$
and learns her corresponding payoff, but nothing about her payoff under other signals.
Meanwhile, this feedback is also probabilistic,
since the realized signal $\signali$
depends on the realized state $\statei$. 
Because of
these two features of the feedback, some natural tasks toward learning the \Receiver's utility may not be completed easily. 
Here are two illustrative examples.
    
\xhdr{Identifying the signs of $\{\RUtilityDiff(\realizedstatei)\}$}
Recall that the optimal signaling scheme in hindsight $\optoff$ 
follows from a threshold signaling scheme
-- 
it recommends action $1$ deterministically
for all states above 
a threshold state $\realizedstatei\primed$, 
recommends action $1$ randomly at threshold state, and 
recommends action $0$ deterministically for all states 
below $\realizedstatei\primed$.
Following the same logic, a natural attempt 
to design a good online policy
is to try to identify the threshold state and the states 
that are above the threshold state. 
However, it is unclear how to identify the threshold state.
In fact,
it is even challenging to identify the sign of 
$\RUtilityDiff(\realizedstatei)$ for a state $\realizedstatei$. 
To see this,
ideally, 
identifying the sign of 
$\RUtilityDiff(\realizedstatei)$  needs to solicit 
the \Receiver's action 
when the \Receiver's posterior belief is concentrated on 
state $\realizedstatei$
when a particular signal is realized. 
A signaling scheme $\scheme$ with $\scheme(\realizedstatej) 
= \indicator{\realizedstatej = \realizedstatei}, 
\forall \realizedstatej\in[m]$
can shape the \Receiver's posterior belief to be concentrated on state $\realizedstatei$
when a signal $1$ is realized. 
However, since the feedback is limited,
such a signaling scheme cannot 
collect useful information 
whenever other signal is realized.
Consequently, it bears a large regret if it happens to be the case when  $\prior(\realizedstatei)$ is small. 

\xhdr{Determining if $\SExpUtility(\optoff) 
\geq C$}
Consider a problem instance with $\totalstate = 2$ states.
Suppose the platform knows that $\RUtilityDiff(1) > 0$,
and $\RUtilityDiff(2) < 0$. This implies that 
the threshold state $\realizedstatei\primed = 2$,
and state 1 is above threshold state 2.
Note that a good online policy should be able to 
approximately identify the value of $\SExpUtility(\optoff)$.
Now, suppose the platform only wants to determine whether
$\SExpUtility(\optoff) 
\geq C$.
If the \Sender\ can determine whether following this   
natural signaling scheme
$\scheme(1) = 1$ and $\scheme(2) = 
\sfrac{(C - \prior(1))}{\prior(2)}$, 
is \IC\ or not,
then the \Sender\ can determine whether
$\SExpUtility(\optoff) 
\geq C$.\footnote{Under this signaling scheme $\scheme$, if 
$\scheme$ is \IC, 
then we have 
$\SExpUtility(\optoff) \ge \SExpUtility(\scheme) \ge C$, 
otherwise $\SExpUtility(\optoff) <\SExpUtility(\scheme) < C$.}
However, since the feedback is probabilistic, 
it takes $\sfrac{1}{C}$ rounds (in expectation) 
to learn the incentive compatibility 
of $\scheme$.
Thus, even if $\SExpUtility(\optoff)$ is small 
(i.e., $\SExpUtility(\optoff) = o(1)$),
by \Cref{lem:persuasive check regret},
the aforementioned attempt bears a 
superconstant regret 
as long as $C = o(\SExpUtility(\optoff))$.

\subsection{Towards \texorpdfstring{$O(\log\log \totaltime)$}\ \ Regret}
\label{sec:loglog T alg}

Despite the aforementioned challenges, 
in this subsection,
we present \texttt{Con}servative \texttt{R}ecommendation \texttt{P}olicy (\CRP) that can have $O(\log\log\totaltime)$
regret guarantee. 
Our algorithm is inspired by \cite{KL-03} yet requires necessary and non-trivial
modifications to account for the unique feedback structure in our problem. In \Cref{apx:dynamic pricing}, we further discuss the connection between our model with the dynamic pricing problem, and how our new algorithmic ingredients can be used to solve a new variant of the dynamic pricing problem.

\newcommand{\permutationSet}{\mathcal{P}}
\newcommand{\permutation}{r}
\newcommand{\SExpUtilityUnderbar}{\underbar{\SExpUtility}}
\newcommand{\permus}{\permutationSet}
\newcommand{\permu}{\permutation}
\newcommand{\permuTrue}{\permu^*}
\newcommand{\Flag}{\texttt{Flag}}
\newcommand{\perc}{\varepsilon}
\newcommand{\stepperc}{\delta}

\newcommand{\underbarScheme}{\underline{\scheme}}

\xhdr{Overview of the algorithm}
In the \CRP,
the whole $\totaltime$ rounds
are divided into
the
\emph{exploring} phase
and
the \emph{exploiting} phase.
The exploring phase has two subphases.
The first subphase
(i.e., \emph{exploring phase I})
identifies
a lower bound and an upper bound
of $\SExpUtility(\optoff)$,
i.e.,
it identifies an \IC\ signaling 
scheme $\underbarScheme$ with 
$\SExpUtilityUnderbar:= \SExpUtility(\underbarScheme)$
such that 
$\SExpUtilityUnderbar
\leq \SExpUtility(\optoff) 
\leq 2 \SExpUtilityUnderbar$.
Note that once we narrow down the 
value of $\SExpUtility(\optoff)$
to be in the interval $[\SExpUtilityUnderbar,2\SExpUtilityUnderbar]$,
with the \IC\ signaling scheme 
$\underbarScheme$, we can ensure 
that expected regret 
to check the incentive compatibility 
of the signaling schemes in the later 
rounds is at most a constant, which
addresses the second challenge 
(i.e., 
determining if $\SExpUtility(\optoff) \geq C$)
we just mentioned before.
We will show that 
the expected cumulative regret in 
exploring phase I is $O(1)$.
The second subphase
(i.e., \emph{exploring phase II})
identifies a signaling scheme 
$\scheme\primed$
whose 
per-round expected regret is $\sfrac{1}{\totaltime}$, 
and we will show that its 
expected cumulative regret is 
at most $O(\log\log \totaltime)$.
The identified signaling scheme $\scheme\primed$
from the exploring phase II
is used in the remaining rounds which are considered as the exploiting phase 
that induces $O(1)$ expected cumulative regret.
See \Cref{alg:loglog T} 
for a formal description.

\xhdr{A subclass of direct 
signaling schemes $\{\scheme^{(u)}\}$}
Our online policy will repeatedly consider
a subclass of direct 
signaling schemes.
Recall program~\ref{eq:optimal in hindsight}
indicates that 
the optimal signaling scheme in hindsight $\optoff$
can be thought of as the optimal solution of a fractional
knapsack problem, 
where each state $\realizedstatei$
corresponds to an item with value $\prior(\realizedstatei)$
and cost $\costState(\realizedstatei)$.
This observation implies that there must 
exist a total order $\permuTrue$
over all states with respect to their true  bang-per-buck $\RUtilityDiff(\realizedstatei) = \sfrac{\prior(\realizedstatei)}
{\costState(\realizedstatei)}$.
Given an arbitrary number $u\in[0, 1]$,
we define $\scheme^{(u)}$
to be the direct signaling scheme as follows: 
there exists a threshold state $\realizedstatei\primed$
such that
(a) for every state
$\realizedstatei\not=\realizedstatei\primed$,
$\scheme^{(u)}(\realizedstatei)
= 
\indicator{\RUtilityDiff(\realizedstatei) >\RUtilityDiff(\realizedstatei\primed)}$ 
and (b) 
$\scheme^{(u)}(\realizedstatei\primed)
=
\frac{
u - \sum_{\realizedstatei\not=\realizedstatei\primed}
\prior(\realizedstatei)
\scheme^{(u)}(\realizedstatei)
}{\prior(\realizedstatei\primed)}$
(recall that since the platform knows the order of user's preference difference $\{\RUtilityDiff(\realizedstatei)\}$, this signaling scheme is well-defined).
As a sanity check, 
observe that 
the signaling scheme $\scheme^{(\SExpUtility(\optoff))}$
is exactly the optimal signaling scheme in hindsight
$\optoff$.
By construction, 
it is also guaranteed that
$
\sum_{\realizedstatei\in[\totalstate]}
\prior(\realizedstatei)\scheme^{(u)}(\realizedstatei)
= u$.
We note that by focusing on the signaling scheme $\scheme^{(u)}$, we bypass the challenge on identifying the value, order or even signs of $\{\RUtilityDiff(\realizedstatei)\}$.
Indeed, for general problem instances, our \CRP\ does not explicitly learn those quantities, nor can they be inferred from the outcome of \CRP.

\newcommand{\AlgoTwoName}{\texttt{Con}servative \texttt{R}ecommendation \texttt{P}olicy (\CRP)}
\begin{algorithm}
\caption{\AlgoTwoName}
\label{alg:loglog T}
\linespread{0.8}\selectfont
\SetAlgoLined\DontPrintSemicolon
\KwIn{number of rounds $\totaltime$, 
number of states $\totalstate$,
prior distribution $\prior$}
\tcc{exploring phase I
--
identify 
$\SExpUtilityUnderbar$
such that 
$\SExpUtilityUnderbar
\leq \SExpUtility(\optoff) 
\leq 2 \SExpUtilityUnderbar$}
Initialize $\SExpUtilityUnderbar \gets \frac{1}{2}$
\\
\While{
    $\CheckPersu\left(\scheme^{(\SExpUtilityUnderbar)}\right) =
    \False$}{
        $\SExpUtilityUnderbar\gets \frac{\SExpUtilityUnderbar}{2}$
    }
\tcc{exploring phase II --
identify a signaling scheme $\scheme\primed$
such that $\SExpUtility(\scheme\primed) \geq \SExpUtility(\optoff) - \frac{1}{\totaltime}$}
Initialize
$R \gets 2\SExpUtilityUnderbar$,
 $L \gets \SExpUtilityUnderbar$,
$\stepperc \gets 1$
\\
\While{$R - L \geq \frac{1}{\totaltime}$
}{
    $\perc \gets \frac{\stepperc}{2}$,
    $S\gets \lfloor\frac{R - L}{\perc L}\rfloor$,
    $\ell \gets 1$. \\
    \While{
        $\CheckPersu\left(\scheme^{(L + \ell\perc L)}\right) =
        \True$
    }{
        $R\gets L + \ell\perc L$,
        $L \gets L + (\ell-1)\perc L$,
        $\stepperc \gets \perc^2$,
        $\ell \gets \ell +1$.
    }
}
Set $\scheme\primed \gets \scheme^{(L)}$.
\\
\tcc{exploiting phase}
Use signaling scheme $\scheme\primed$
for all remaining rounds.
\end{algorithm}
\vspace{-10pt}
\begin{remark}
In the \CRP, 
the first subphase (i.e., exploring phase I) is used to identify a lower bound and an upper bound of the optimal
payoff $\SExpUtility(\optoff)$ of the platform. 
This step is crucial for us to establish the $O(\log\log \totaltime)$ regret. In \Cref{subsec:simulation}, we present simulation studies showing that without this step, the algorithm may perform very badly. Moreover, this subphase is also used to identify an interior point
in designing our second main algorithm (\Cref{alg:log T}).
\end{remark}
We are now ready to describe the main result of this section.

\begin{theorem}
\label{thm:loglog T regret upper bound}
The expected regret 
of \CRP\
is at most 
$O(\log\log \totaltime)$.
\end{theorem}
\begin{proof}
We analyze  
the expected 
regret in 
exploring phase I,
exploring phase II,
and exploiting phase separately.
We first assume that \CRP\
finishes exploring phase I and II before $\totaltime$
rounds are exhausted.
A similar argument follows for the other case where 
exploring phase I or exploring phase II is completed 
due to the exhaustion of rounds.

\xhdr{Exploring phase I}
Let $K = -\lceil\log(\SExpUtility(\optoff))\rceil$.
By definition, 
$\CheckPersu(\scheme^{(2^{-k})})=\False$
for $k \in [K - 1]$,
and 
$\CheckPersu(\scheme^{(2^{-K})}) = \True$.
Thus, at the end of exploring phase I,
$\SExpUtilityUnderbar$ 
 is $2^{-K}$,
 and 
there are $K$ iterations
in the while loop.
For each iteration $k\in[K]$, 
$\CheckPersu(\scheme^{(2^{-k})})$ 
is called once.
By \Cref{lem:persuasive check regret}, the total expected regret is 
\begin{align*}
    \sum\nolimits_{k \in [K]}
    \frac{\SExpUtility(\optoff)}{
    \sum_{\realizedstatei\in[\totalstate]}
\prior(\realizedstatei)\scheme^{(2^{-k})}(\realizedstatei)
    }
    \overset{(a)}{\leq} 
    \sum\nolimits_{k \in [K]}
    \frac{2^{-(K-1)}}{
    2^{-k}
    }
    =
    \sum\nolimits_{k\in[K]}
    2^{-(K - k - 1)}
    = O(1)
\end{align*}
where the denominator 
in the right-hand side 
of inequality~(a) is due to the 
construction of $\scheme^{(2^{-k})}$.

\xhdr{Exploring phase II}
By construction, 
there are $O(\log\log \totaltime)$
iterations in the while loop.
Thus, it is sufficient to show 
the expected regret in each iteration 
is $O(1)$.

In each iteration $k$,
let $\ell\primed\in[S]$
be the smallest index 
that the signaling scheme
$\scheme^{(L + \ell\primed\perc L)}$
is not \IC.
The expected regret in iteration $k$ 
is at most 
\begin{align*}
    &~\sum\nolimits_{\ell=1}^{\ell\primed - 1}
    \left(
    \frac{\SExpUtility(\optoff)}{
    \sum\nolimits_{\realizedstatei\in[\totalstate]}
\prior(\realizedstatei)\scheme^{
(
L +\ell\perc L)}(\realizedstatei)
    }
    -1
    \right)
    +
    \frac{\SExpUtility(\optoff)}{
    \sum\nolimits_{\realizedstatei\in[\totalstate]}
\prior(\realizedstatei)
\scheme^{
(
L +\ell\primed\perc L)}
(\realizedstatei)
    }
    \\
    \overset{(a)}{=} &
    \sum\nolimits_{\ell=1}^{\ell\primed - 1}
    \left(
    \frac{\SExpUtility(\optoff)}{
L +\ell\perc L
    }
    -1
    \right)
    +
    \frac{\SExpUtility(\optoff)}{
L +\ell\primed\perc L
    }
\overset{(b)}{\leq} 
\sum\nolimits_{\ell=1}^{\ell\primed - 1}
    \left(
    \frac{R}{
L
    }
    -1
    \right)
    +
    \frac{R}{
L
    }
\overset{(c)}{\leq}
(S - 1) \frac{R - L}{L} + 2
\overset{(d)}{\leq}
~ \frac{(R - L)^2}{\perc L^2} + 2
\end{align*}
where equality~(a) holds due to 
the construction of 
$\scheme^{
(
L +\ell\perc L)}$
and $\scheme^{
(
L +\ell\primed\perc L)}$;
inequality~(b) holds 
since $\SExpUtility(\optoff) \leq R$;
inequality~(c) holds since 
$\ell\primed \leq S$ and $R \leq 2 L$;
and
inequality~(d) holds since $S = 
\lfloor\frac{R - L}{\perc L}\rfloor$.

We finish this part 
by showing
$R - L \leq \sqrt{2\perc} L$
by induction.
Let $L^{(k)}, R^{(k)}$, $\stepperc^{(k)}$
and $\perc^{(k)}$
be the value of $L, R, \stepperc,\perc$
in each iteration $k$.
The claim is satisfied for
iteration $k = 1$,
since $R^{(1)} - L^{(1)} = 2\SExpUtilityUnderbar - \SExpUtilityUnderbar 
= L^{(1)}$
and $\perc^{(1)} = \sfrac{1}{2}$.
Suppose the claim holds for iteration $k - 1$.
Now, for iteration $k$, we know that 
$R^{(k)} - L^{(k)} = \perc^{(k-1)} L^{(k - 1)}
\leq \perc^{(k-1)} L^{(k)}
=
\sqrt{\stepperc^{(k)}} L^{(k)}
=
\sqrt{2\perc^{(k)}} L^{(k)}$,
which finishes the induction.

\xhdr{Exploiting phase}
In this phase, we know that 
$\scheme\primed$ is \IC\ and 
$\SExpUtility(\scheme\primed) \geq 
\SExpUtility(\optoff) - \sfrac{1}{\totaltime}$,
which concludes the proof.
\end{proof}

\section{Implementing \CRP\ for Unknown Ordinal Preference}
\label{sec:unknwon order}

In this section, we discuss how
to adapt the \CRP\
when the user's ordinal preference, i.e., the 
order of $\{\RUtilityDiff(\realizedstatei)\}$,
is not known to the platform.
In \Cref{sec:affine preference},
we first present the application of \CRP\ for a general class of user preferences.
In \Cref{sec:unknown preference CRP},
we discuss a more challenging setting 
where the platform has completely no knowledge about the user's preference.
In both scenarios, we show that regret $O(\log\log T)$ is achievable by implementing \CRP. 

\subsection{Application: Affine State-dependent User Preference}
\label{sec:affine preference}
In this subsection, 
we show that for a general class of 
users' preferences which are affine 
with respect to the state,
\CRP\
is able to achieve $O(\log\log \totaltime)$ regret, 
and this regret bound does not depend on the number of possible states. 

\xhdr{Affine state-dependent user preference}
With a slight abuse of the model, in this subsection, we consider a continuous state space $\stateSpace \subseteq \reals_+$ instead of a finite state space $[m]$. We say a user has an \emph{affine state-dependent preference} if 
her expected utility  $\expect[\state\sim\posteriorState]{\RUtility(\state,\action)}$,
given posterior $\posteriorState \in \Delta(\stateSpace)$ and action $\action$,
is equal to $\RUtility(\expect[\state\sim\posteriorState]{\state},\action)$.
The affine state-dependent preference
is a fundamental setting 
in information design literature \citep{can-22, ABSY-23,KA-18,GK-16,CS-21,KMZL-17}. 

For a user with an affine state-dependent preference, her optimal action $\action^* = \argmax_\action \expect[\state\sim\posteriorState]{\RUtility(\state,\action)}$ (e.g., whether to watch the video or not), depends only on the expected state $\expect[\state\sim\posteriorState]{\state}$ (e.g., expected quality of the video) of the user's posterior belief $\posteriorState$ over the underlying states. 
It is worth highlighting that an affine state-dependent preference could be decomposed as 
$\RUtility(\state, \action) = \RUtility_1(\action) \state + \RUtility_2(\action)$
where functions $\RUtility_1, \RUtility_2: \cA \rightarrow \R$ are unknown to the platform. 

\xhdr{\textbf{\CRP}\ for affine state-dependent preference}
To see how \CRP\ could be adapted to 
solve this setting, notice that under affine state-dependent preference, the user's preference difference is essentially
$\RUtilityDiff(\realizedstatei) = (\RUtility_1(1) -  \RUtility_1(0)) \realizedstatei  + \RUtility_2(1) - \RUtility_2(0)$. This shows that 
there exists at most two possible orders of the user's preference difference $\{\RUtilityDiff(\realizedstatei)\}$, depending on the sign of $\RUtility_1(1) -  \RUtility_1(0)$. 
Thus, one can just run \CRP\ over 
these two possible orders of $\{\RUtilityDiff(\realizedstatei)\}$ in 
a round-robin manner and use the payoff of any identified 
\IC\ signaling scheme to prune out the incorrect order.

The main result in this subsection 
is summarized as follows:
\begin{proposition}
\label{prop:ordinal preference}
For the affine state-dependent preference,
the expected regret of \CRP\
is $O(\log\log\totaltime)$, this regret also holds even when the state space $\stateSpace$ 
is continuous. 
\end{proposition}
Lastly, we observe that the above results
can also be extended to scenarios wherein the user's preferences depend on a potentially non-linear transformation of the state, denoted as $f(\state)$, in an arbitrary manner. 
Fundamental to our conclusions is the observation that the order of the user's preference difference $\{\RUtilityDiff(\realizedstatei)\}$ is not changing even with such non-linear transformation. 

\subsection{Application: 
User with Unknown Ordinal Preference}
\label{sec:unknown preference CRP}
\label{subsec:completely unknown order}
In this subsection, we show that when the order 
of user's preference differences
$\{\RUtilityDiff(\realizedstatei)\}$ is unknown to the platform, a regret $O(\totalstate2^{\totalstate-1}\cdot \log\log \totaltime)$ is achievable 
by enumerating all possible orders. 

\begin{restatable}{proposition}{oglogTregretupperboundunknownorder}
\label{thm:loglog T regret upper bound unknown order}
When the platform has no knowledge of the user's preference,
the expected regret of a modified version (see \Cref{alg:loglog T w unknown order} in \Cref{apx-missing-algo}) of \CRP\ is $O(\totalstate2^{\totalstate-1}\cdot \log\log \totaltime)$.
\end{restatable}
Below we briefly discuss how the \CRP\
could be adapted to obtain the above regret bound.
The formal proof of \Cref{thm:loglog T regret upper bound unknown order} is provided in \Cref{apx-missing-algo}. 
Similar to \CRP, the modified algorithm (\Cref{alg:loglog T w unknown order})
also relies on a subclass of direct signaling schemes $\scheme^{(\permu,u)}$, but with an additional parameter 
$\permu$ to represent a possible (total) order of the user's 
preference differences $\{\RUtilityDiff(\realizedstatei)\}$. 
In other words, for 
every possible order $\permu$ of the user's preference differences 
$\{\RUtilityDiff(\realizedstatei)\}$, one can identify 
a direct signaling scheme $\scheme^{(\permu,u)}$ such that 
the expected payoff to the platform will exactly equal $u$
if this signaling scheme $\scheme^{(\permu,u)}$ is {\em \IC}.
In the modified algorithm, whenever the algorithm identifies 
an \IC\ signaling scheme $\scheme^{(\permu,u)}$ (see Line $4$ and $15$ in \Cref{alg:loglog T w unknown order}), 
it then naturally gives
a lower bound 
of the platform's expected payoff of the optimal
signaling scheme, namely, $\SExpUtility(\optoff) \ge u$. 
Then one can use the payoff (i.e., $u$) of
such \IC\ signaling scheme to further prune out 
other signaling schemes that are constructed with different order of 
$\{\RUtilityDiff(\realizedstatei)\}$ but 
are either non-\IC\ or have payoff less than $u$.
We would like to note that without pruning, 
the regret of the algorithm might 
have linear dependence on the time horizon $\totaltime$.

In more detail,
recall that when \Receiver's utility function is unknown,
the bang-per-buck as well as the true total order $\permuTrue$ 
are unknown to the platform.
In the exploring phase of \Cref{alg:loglog T w unknown order}, our online policy maintains 
a subset $\permus$ of total orders over $[m]$
that contains the optimal order $\permuTrue$.
In particular, the algorithm initializes the 
set $\permus$ such that it contains all possible orders, 
and each order $\permu$ in the set
$\permus$ specifies a state $\realizedstatei\primed\in[\totalstate]$, 
a subset of states
that are below the state $\realizedstatei\primed$, and a subset of states that are above the state $\realizedstatei\primed$.
Given an arbitrary order $\permu$,
we define $\scheme^{(\permu,u)}$
to be the direct signaling scheme as follows: 
let state $\realizedstatei\primed$ be the state associated 
with this order $\permu$, then
(a) for every state
$\realizedstatei\not=\realizedstatei\primed$,
$\scheme^{(\permu,u)}(\realizedstatei)
= 
\indicator{\permu(\realizedstatei) >\permu(\realizedstatei\primed)}$,\footnote{Given an order $\permu$, we denote $\permu(\realizedstatei)$ by
the rank of state $\realizedstatei$.}
and (b) 
$\scheme^{(\permu,u)}(\realizedstatei\primed)
=
\frac{
u - \sum_{\realizedstatei\not=\realizedstatei\primed}
\prior(\realizedstatei)
\scheme^{(\permu,u)}(\realizedstatei)
}{\prior(\realizedstatei\primed)}$.
As a sanity check, 
observe that 
the signaling scheme $\scheme^{(\permuTrue,
\SExpUtility(\optoff))}$
is exactly the optimal signaling scheme in hindsight
$\optoff$, and it is also guaranteed that
$
\sum_{\realizedstatei\in[\totalstate]}
\prior(\realizedstatei)\scheme^{(\permu,u)}(\realizedstatei)
=
u$ by construction.

Note that even though the number of all possible total orders
could be as large as $O(m!)$, we can have a more 
succinct representation on user's ordinal preference. 
To see this, note that each possible order $\permu$ 
can first specify a state 
$\realizedstatei\primed$, 
a subset of states that are below the state $\realizedstatei\primed$, 
and remaining states that are above the state $\realizedstatei\primed$.
Then, in total, there are at most $O(m2^{m-2})$ such possible orders.

\begin{remark}
In both exploring phase I and II, 
\Cref{alg:loglog T w unknown order} checks 
whether there exists $\permu\in\permus$
such that $\CheckPersu(\scheme^{(\permu,u)}) =
\True$ for some $u$.
Our regret bound has an $(2^\totalstate)$ dependence due to
brute-force searching over $\permu$.
We would like to note that 
(i)
in many practical applications,
the number of states $\totalstate$
is small or even constant,
and thus our main focus 
in this section is the
optimal dependence on the number of rounds $\totaltime$,
and 
(ii)
when $N~(\leq \totalstate2^{\totalstate-1})$ identical problem instances 
are allowed to run in parallel, 
the regret dependence on $\totalstate$ becomes 
$\sfrac{\totalstate2^{\totalstate-1}}{N}$. 
\end{remark}

\section{Alternative LP-based Algorithm}
\label{sec_logn}

\newcommand{\lpsolver}{\MemberLP}
\label{sec:log T}

In this section, 
we provide our second main result 
--
a linear program-based recommendation policy (\LPRP) with
$\polymregret$ regret when the user's preferences
(including both the cardinal preference and ordinal preference)
are unknown to the platform.
The main result  is as follows:
\begin{theorem}
\label{thm_reg_upper_bound}
\label{thm:log T}
The expected regret of \LPRP\
is at most 
$
    O\left(
\totalstate^6
     \log^{O(1)}
    \left(\totalstate\totaltime
    \right)
    \right)$.
\end{theorem}
The proposed \LPRP\ uses a  subroutine $\lpsolver$
-- an algorithm 
\citep[e.g.,][]{LSV-18}
to solve
linear program 
with membership oracle access.\footnote{\LPRP\
uses $\lpsolver$ as a blackbox.
Namely, it can be replaced 
by other algorithms for linear program 
with membership oracle access.}
We first formally introduce 
the linear program optimization with membership 
oracle access,
and 
discuss its connection to our online
Bayesian recommendation problem.
Then
we provide the formal description
and the explanation of \LPRP\,
where
we also 
present the proof of \Cref{thm:log T}.

\xhdr{Linear program optimization
with membership oracle access}
\newcommand{\convexset}{H}
\newcommand{\zeroed}{^{(0)}}
\newcommand{\TwoNormBall}{\mathbf{B}_2}
Optimizing a linear function 
$f(\cdot)$
within an unknown convex set $\convexset$ 
has been studied 
extensively in the literature.
There are two standard oracle assumptions:
\emph{membership oracle}
and 
\emph{separation oracle}.
A membership oracle returns 
whether a queried point $y$ is contained in 
convex set $\convexset$.
In contrast, a 
separation oracle not only returns
whether
a queried point $y$ is contained in 
convex set $\convexset$,
but also 
returns 
a hyperplane that separates $y$ from $\convexset$
if $y\not\in\convexset$.

Recall that in our problem, 
the optimal signaling scheme in hindsight
$\optoff$
is the optimal solution of the linear program~\ref{eq:optimal in hindsight}.
From the platform's perspective,
the only unknown component in this program  
is $\{\costStatei\}$ in 
the IC constraint.
Nonetheless, using Procedure~\ref{alg:persuasive check},
the platform can determine the incentive compatibility 
(i.e., whether the IC constraint is satisfied)
of any direct signaling scheme.
In other words, Procedure~\ref{alg:persuasive check} 
works like a membership oracle for the convex set 
which contains all \IC\ signaling schemes.
Thus, finding the optimal signaling scheme $\optoff$
can be formulated as optimizing a linear program with 
membership oracle access.
In particular, we leverage 
the algorithm introduced in
\citet{LSV-18} with the following guarantee.
\begin{theorem}[\citealp{LSV-18}]
\label{thm:membership oracle query bound}
For any linear function $f$,
and convex set 
$\convexset \subseteq \reals^\totalstate$,
given 
an interior point $x\zeroed$,
a lower bound $r$, an upper bound $R$
such that $\TwoNormBall(x\zeroed, r)\subseteq \convexset \subseteq
\TwoNormBall(x\zeroed,R)$,\footnote{$\TwoNormBall(x\zeroed, r)$ is 
the ball of radius $r$ centered at $x\zeroed$.}
and given a membership oracle,
there exists an algorithm $\lpsolver$
that 
finds an $\epsilon$-approximate optimal solution 
for $f$ in $\convexset$
with
probability $1-\delta$,
using $O(\totalstate^2\log^{O(1)}
\left(\sfrac{\totalstate R}{\epsilon\delta r}\right))$
queries to the oracle.
\end{theorem}

We note that
when only membership oracle is given,
the interior point $x\zeroed$
as well as 
lower bound $r$, and upper bound $R$
such that $\TwoNormBall(x\zeroed, r)\subseteq \convexset \subseteq
\TwoNormBall(x\zeroed,R)$ is necessary for any algorithms.
Otherwise, there is an information-theoretic barrier
\citep[see][]{GLS-88}.

\subsection{Towards \texorpdfstring{$\polymregret$}{} Regret}
\label{sec:log T proof}

Before we describe our algorithm, let
us highlight two major hurdles in applying 
the membership oracle approach 
to solve
our Bayesian recommendation problem.
\begin{enumerate}
\item Though \Cref{thm:membership oracle query bound}
upper bounds
the total number of queries 
to the membership oracle (a.k.a.,
Procedure~\ref{alg:persuasive check}),
as illustrated in our second example presented 
in \Cref{sec:loglog T},
the regret from one execution of 
Procedure~\ref{alg:persuasive check}
may be superconstant.

\item \Cref{thm:membership oracle query bound}
requires an interior point $x$
as well as a
lower bound radius $r$, and an upper bound radius $R$
such that $\TwoNormBall(x, r)\subseteq \convexset \subseteq
\TwoNormBall(x,R)$,
and the number of queries  depends 
on the value of $r$ and $R$.
However in our problem, the interior point is 
{\em not} given explicitly. 
How to find a proper interior point $x$
with non-trivial lower bound radius $r$ 
(without incurring too much regret)
is not obvious in our problem.
\end{enumerate}
\xhdr{Overview of the algorithm}
We now sketch \LPRP.
\ifinforms
The formal description of {\LPRP} and its detailed proof is provided in \Cref{apx:proofs in log T}.
\else
\fi
The high-level idea of this algorithm is to
use $\lpsolver$ as a subroutine to 
identify a signaling scheme 
$\scheme\primed$ whose 
per-round expected regret is $\sfrac{1}{\totaltime}$.
In more detail, \LPRP\ divides the whole $\totaltime$ 
rounds into an \emph{exploring phase}
and an \emph{exploiting phase}.
In exploring phase, we use $\lpsolver$ to identify 
a persuasive signaling scheme $\scheme\primed$, 
and exploiting phase uses $\scheme\primed$ 
until the rounds are exhausted.
As mentioned in the above two hurdles,
to use $\lpsolver$, we need to ensure that 
each query (i.e., a signaling scheme) 
to the $\lpsolver$ cannot incur too much regret, 
i.e., Procedure~\ref{alg:persuasive check} for checking 
the persuasiveness of a queried signaling scheme cannot be large; 
and we need to find a proper interior point 
with non-trivial lower bound radius. 
To achieve this, there are three subphases in the exploring phase:

\smallskip

\vspace{-4pt}
   \noindent\underline{\textsl{Exploring phase I -- Lowerbounding $\SExpUtility(\optoff)$}}: 
    Similar to \CRP, the first step of 
    \LPRP\ is to identify 
    a lower bound and an upper bound 
    of $\SExpUtility(\optoff)$.
    But different from \CRP, here, 
    we identify a set $\anchorsetHat$ of persuasive direct signaling schemes
    such that for every signaling scheme
    $\schemeanchorpair \in \anchorsetHat$, it has the following two 
    properties: 
    $(i)$ it has the same payoff $\SExpUtilityUnderbar$
    with other signaling schemes in
    set $\anchorsetHat$, i.e., 
    $\SExpUtilityUnderbar \equiv \SExpUtility(\schemeanchorpair), \forall \schemeanchorpair \in \anchorsetHat$, and 
    $\SExpUtilityUnderbar$ is relatively good, i.e.,
    $\SExpUtilityUnderbar \ge \frac{\SExpUtility(\optoff)}{\totalstate^2}$;
    $(ii)$ signaling scheme $\schemeanchorpair$ has a specific 
    structure where it has non-zero probability 
    for recommending action $1$ on at most two states.

    \begin{lemma}[informal]
\label{lem:existence of anchor pair main}
When exploring phase I terminates, 
$\SExpUtilityUnderbar \geq
\frac{\SExpUtility(\optoff)}{\totalstate^2}$
and 
$\anchorsetHat$ is not empty.
\end{lemma}
    
    At a high level, the property $(i)$ implies that
    $\SExpUtilityUnderbar \leq \SExpUtility(\optoff)
    \leq \totalstate^2\,\SExpUtilityUnderbar$, which can 
    guarantee us whenever we use Procedure~\ref{alg:persuasive check}
    (as a membership oracle)
    to check the persuasiveness of 
    a direct signaling scheme in the later rounds,
    the expected regret is at most $O(\totalstate^2)$.
    The property $(ii)$ can guarantee that we
    find an interior point 
    with non-trivial lower bound radius $r$
    in the later subphase.

    \smallskip
    
\vspace{-4pt}
    \noindent\underline{\textsl{Exploring phase II --  
    Excluding Degenerate States}}:
    To find the interior point for the
    program \ref{eq:optimal in hindsight},
    however, we first note that 
    it is possible the convex set 
    in the program~\ref{eq:optimal in hindsight}
    is degenerate and thus no interior point exists.
    Nonetheless, those degenerate dimensions (i.e., states) 
    must contribute little to $\SExpUtility(\optoff)$.
    Thus, in this exploring phase, we use the signaling schemes
    in $\anchorsetHat$ obtained in Exploring phase I to exclude 
    those states and obtain a set $\meaningfulset \subseteq [m]$ 
    that contains all relatively good states 
    (i.e., states whose $\costState$ 
    cannot be too negative).
    
    \begin{restatable}{lemma}{meaningfulSetCharacterization}
\label{lem:meaningful set characterization}
When exploring phase II terminates,
\begin{itemize}
\item for each state $\realizedstatei
    \in\meaningfulset$:~
    $\costStatei \geq -\totalstate\totaltime
    \cdot 
    \max_{\realizedstatej\in[\totalstate]}
    \costStatej$;
    \item
     for each state $\realizedstatei\not\in
    \meaningfulset$:~
    $\costStatei < - 
    \frac{\totalstate\totaltime}{3} 
    \cdot
    \max_{\realizedstatej\in[\totalstate]}
    \costStatej$.
    \end{itemize}
\end{restatable}
    
    With the obtained $\meaningfulset$ at hand, 
    we show that 
    there exists a persuasive direct 
    signaling scheme $\schemestartingpointTilde$ 
    such that  
    $
    \frac{1}{8\totalstate^2\totaltime}\leq
    \schemestartingpointTilde(\realizedstatei) 
    \leq 
    1 - \frac{1}{8\totalstate^2\totaltime}$
    for every $\realizedstatei
    \in
    \meaningfulset$,
    and 
    $\schemestartingpointTilde(\realizedstatei) = 0$
    for every $\realizedstatei\not\in \meaningfulset$.
    Furthermore, 
    signaling scheme $\schemestartingpointTilde$  
    is an interior point\footnote{Here 
    we mean $\schemestartingpointTilde$
    is an interior point of convex set 
    in program~\ref{eq:optimum in hindsight within meaningfulset} 
    when we restrict to 
    states in $\meaningfulset$.} 
    of the following linear program.
\begin{align}
    \tag{$\mathcal{P}^{\texttt{opt}}_{\meaningfulset}$}
    \label{eq:optimum in hindsight within meaningfulset}
        \begin{array}{llll}
          \max\limits_{\scheme:
          \scheme(\realizedstatei) =
          0~\forall\realizedstatei\not\in\meaningfulset}   &  
          \displaystyle\sum\nolimits_{\realizedstatei\in\meaningfulset}
          \prior(\realizedstatei)\scheme(\realizedstatei)\quad
          &  \text{s.t.}
          &\\
             & \displaystyle\sum\nolimits_{\realizedstatei\in\meaningfulset}
             \costStatei
             \scheme(\realizedstatei)
             \geq 0\quad 
             & & \\
             & \displaystyle\sum\nolimits_{\realizedstatei\in\meaningfulset}
             \prior(\realizedstatei)
             \scheme(\realizedstatei)
             \geq \frac{1}{16}\,\SExpUtilityUnderbar\quad 
             \qquad& & \\
             &\scheme(\realizedstatei)\in [0,1]
             \qquad&
              \realizedstatei\in\meaningfulset
             &
        \end{array}
    \end{align}
From \Cref{lem:meaningful set characterization},
    we know that the optimal objective value 
    of program~\ref{eq:optimum in hindsight within meaningfulset}
    is close to $\SExpUtility(\optoff)$.

\begin{restatable}{lemma}{meaningfulSetApproxOpt}
\label{lem:meaningful set approximate optimal}
Let $\meaningfulsetopt$
be the optimal solution in 
program~\ref{eq:optimum in hindsight within meaningfulset},
i.e., $\meaningfulsetopt
= \argmax$
\ref{eq:optimum in hindsight within meaningfulset}.
Then
$\SExpUtility(\meaningfulsetopt) \geq 
\SExpUtility(\optoff) - O(\frac{1}{\totaltime})$.
\end{restatable}
    
    \smallskip
    
\vspace{-4pt}
    \noindent\underline{\textsl{Exploring phase III
    --
    Executing $\lpsolver$
    with Interior Point Candidates}}: 
    In this phase, we identify a direct signaling scheme $\scheme\primed$
    whose per-round expected regret is $O(\frac{1}{\totaltime})$
    (i.e., $\SExpUtility(\scheme\primed) \geq \SExpUtility(\optoff) - O(\frac{1}{\totaltime})$)
    with probability $1-\frac{1}{\totaltime}$.
    To do this, we solve a $\frac{1}{\totaltime}$-approximate solution in
    program~\ref{eq:optimum in hindsight within meaningfulset}
    by using $\lpsolver$ as a subroutine.
    However, we note that we cannot directly identify 
    the interior point $\schemestartingpointTilde$ to 
    the program~\ref{eq:optimum in hindsight within meaningfulset}
mentioned in exploring phase II.
    Instead, 
    we introduce a specific modification
    for signaling schemes in $\anchorsetHat$
    --
for every 
    signaling scheme $\schemeanchorpair\in \anchorsetHat$,
    its modification 
    $\schemestartingpoint$
    is an interior point candidate.
    In particular,
    because of \Cref{lem:meaningful set characterization},
    there exists 
    a signaling scheme $\schemeanchorpair \in \anchorsetHat$
whose 
    modification
    $\schemestartingpoint$ is 
    indeed
    an interior point $\schemestartingpointTilde$.

\vspace{-2pt}
\begin{lemma}[informal]
\label{lem:interior point main}
There exists 
a signaling scheme
$\schemeanchorpair\in
\anchorsetHat$
such that 
its modification
$\schemestartingpoint$ is an interior point 
of program~\ref{eq:optimum in hindsight within meaningfulset}.
In particular, let $r = \frac{1}{16\totalstate^2\totaltime}$,
then $\TwoNormBall(\schemestartingpoint,r)\subseteq 
\convexset($\ref{eq:optimum in hindsight within meaningfulset}$)$.
\end{lemma}
    
\vspace{-2pt}
    Finally, we run $\lpsolver$ based on 
    every interior point candidate
    $\schemestartingpoint$
(and in the end, we  
    pick the best solution as $\scheme\primed$),
    where we set the interior
    point $x\zeroed\gets 
    \schemestartingpoint$,
    lower-bound radius $r\gets \frac{1}{16\totalstate^2\totaltime}$,
    upper-bound radius $R \gets \sqrt{\totalstate}$,
    precision $\epsilon\gets \frac{1}{\totaltime}$
    and success probability $\delta \gets \frac{1}{\totaltime}$.\footnote{When an incorrect 
    interior point is given,
    $\lpsolver$ terminates with
    a suboptimal solution.
    The number of queries to the oracle
    is the same as 
    the one in \Cref{thm:membership oracle query bound}.}

\section{\texorpdfstring{$\Omega(\log\log \totaltime)$}\ \ Regret Lower Bound}
\label{sec:loglog T lower bound}
\ifinforms

We now show a tight lower bound 
of
$\Omega(\log\log \totaltime)$
regret of any online policy,
even when the number of states is 2 and the user has the affine state-dependent preference.
Here we allow the online policy 
to be randomized (i.e., 
can commit to different signaling schemes at random)
and 
have non-binary (but finite) signal spaces
(i.e., can have multiple recommendation 
levels).

\begin{theorem}
\label{thm:loglog T regret lower bound}
No online policy
can achieve an expected regret better than $\Omega(\log\log \totaltime)$, 
even for the family of binary-state problem instances.
\end{theorem}
We note that for every binary-state instance, 
the user's utility can be represented as 
an affine function over the state space.
Thus, the above regret lower bound also holds
when the user's utility is an affine function over the state space.

\xhdr{Overview of the proof}
To show \Cref{thm:loglog T regret lower bound},
we focus on problem instances 
with binary state.
Our proof mainly consists of  two steps.
In the first step, we show 
that
for problem instances with binary state, 
any online policy 
can be transformed into a randomized online policy
that only uses signaling schemes with
binary signal space.
This statement is no longer true 
for general problem instances 
with non-binary state, since
the classic revelation principle fails.
The key technical ingredient
(\Cref{lem:multi label signal to binary label signal}) 
is to show 
that any posterior distribution of binary 
state can be induced by a convex combination
of signaling schemes with binary signal space,
which may be of independent  interest.
In the second step, we show a
reduction from
the single-item dynamic pricing problem
to
our online Bayesian recommendation 
problem with binary state.
Thus, the $\Omega(\log\log \totaltime)$ 
regret lower bound known in dynamic pricing problem
\citep{KL-03}
can be extended to our problem. (See more discussions about the connection between our problem and the dynamic pricing problem in \Cref{apx:dynamic pricing}.)

Below we provide detailed discussion and related lemmas
for the above mentioned two steps.
At the end of this subsection, 
we combine all pieces together to conclude the proof of 
\Cref{thm:loglog T regret lower bound}.

\xhdr{Step 1: Binary signals suffice}  
Our first step is to show that every online policy 
can be transformed into
a randomized online policy with binary signal space.
While this might 
appear obvious at first as binary signal suffices
in the optimal signaling scheme in hindsight, it is 
not a-priori clear whether restricting to binary signals is without loss
in an online policy without knowing \Receiver's utility. 

\begin{restatable}{lemma}
{multilabelpolicytobinarylabelpolicy}
\label{lem:multi label policy to binary label policy}
Given any problem instance with binary state,
for any online policy $\ALG$,
there exists an online policy $\ALG\primed$
which only uses signaling schemes with binary 
signal space and has regret $\Reg{\ALG\primed} = \Reg{\ALG}$.
\end{restatable}
We now first sketch the intuition behind
\Cref{lem:multi label policy to binary label policy}.
Fix an arbitrary online policy $\ALG$, we construct 
a randomized online policy $\ALG\primed$ 
with binary signal space
that uses the original policy
$\ALG$ as a blackbox.
Briefly speaking, 
in each round $t$,
policy $\ALG\primed$
first asks 
which signaling scheme $\schemei$ is used by $\ALG$
in this round.
Then, $\ALG\primed$
uses a signaling scheme $\schemei\primed$
with binary signal space 
at random\footnote{Namely, $\ALG\primed$
randomly picks a signaling scheme $\schemei\primed$
and commits to it in round $t$.}
such that the
distribution of \Receiver\ $t$'s posterior belief
induced in $\schemei\primed$ (over the randomness
of state, signaling scheme $\schemei\primed$ used by $\ALG$,
and $\schemei\primed$ itself)
is the same as the one induced by $\schemei$.
Note that from \Receiver\ $t$'s perspective, 
his best response is uniquely determined by 
his posterior belief. Thus, the distribution of 
\Receiver\ $t$'s action is the same in both $\ALG$
and $\ALG\primed$.
Finally, $\ALG\primed$ sends \Receiver\ $t$'s 
action $\actioni$
as the feedback to $\ALG$, and moves to the next round.
The formal proof of \Cref{lem:multi label policy to binary label policy} is in \Cref{apx:proofs lower bound}.

\xhdr{Step 2: Reduction from dynamic pricing} The second step in the proof of \Cref{thm:loglog T regret lower bound}
is a reduction from 
the single-item dynamic pricing problem to
our online Bayesian recommendation problem.
The definition of single-item 
dynamic pricing problem is as follows.
\begin{definition}
In the \emph{single-item dynamic pricing problem},
there is a seller with unlimited units of a single item
and $\totaltime$ buyers.
In each round $t\in[\totaltime]$, 
the seller wants to sell a new unit of the item 
(by setting a price $\price_t$)
to buyer $t$.
Buyer $t$ has a private value $\optval$
that is unknown to the seller,
and will buy the item (and pay $\price_t$) 
if and only if $\optval \geq \price_t$.
The regret of a dynamic pricing mechanism $\ALG$
is 
\begin{align*}
    \Reg{\ALG} \triangleq \totaltime\cdot \optval
    - \expect[\price_1, \dots, \price_\totaltime]
    {\sum\nolimits_{t\in[\totaltime]} \price_t 
    \cdot \indicator{\price_t \leq \optval}}
\end{align*}
where $\price_t$ is the price posted by $\ALG$ in each round $t\in[\totaltime]$.
\end{definition}

\begin{theorem}[\citealp{KL-03}]
\label{lem:dynamic pricing lower bound}
In single-item dynamic pricing problem,
no randomized dynamic pricing mechanism 
can achieve an expected regret better than 
$\Omega(\log\log \totaltime)$.
\end{theorem}

The following lemma (its formal proof 
is deferred to \Cref{apx:proofs lower bound}.) formally states
the reduction from 
the single-item dynamic pricing problem to
our online Bayesian recommendation problem.
\begin{restatable}{lemma}{dynamicPricingReduction}
\label{lem:dynamic pricing reduction}
For every single-item dynamic pricing 
problem instance $\instance$,
there exists an online Bayesian recommendation problem instance 
$\instance\primed$
with binary state.
For every online policy $\ALG\primed$
with binary signal space 
and regret $\Reg[\instance\primed]{\ALG\primed}$
on online Bayesian recommendation instance $\instance\primed$,
there exists a dynamic pricing mechanism $\ALG$
with regret $\Reg[\instance]{\ALG} 
\leq \Reg[\instance\primed]{\ALG\primed} + 1$ on dynamic pricing instance $\instance$.
\end{restatable}

\begin{remark}
We would like to note that the binary-state 
instance constructed in the proof of \Cref{lem:dynamic pricing reduction}
also satisfies that the user's utility function is an 
affine function over the states. 
\end{remark}
Putting all pieces together, we are ready to prove
\Cref{thm:loglog T regret lower bound}.

\begin{proof}[Proof of \Cref{thm:loglog T regret lower bound}]
Combining \Cref{lem:dynamic pricing lower bound}
and \Cref{lem:dynamic pricing reduction},
in the online Bayesian recommendation problem 
with binary state, 
no randomized online policy 
with binary signal space 
can achieve
an expected regret better than 
$\Omega(\log\log \totaltime)$.
Invoking \Cref{lem:multi label policy to binary label policy}
finishes the proof.
\end{proof}

 \fi

\section{Extensions}
\label{sec_generalization}

In this section,
we briefly discuss three extensions:
(i) the platform has a state-dependent utility function,
(ii) users have an unknown misspecified prior belief, and
(iii) the platform has access to different video categories and must decide which category to display and how.
For the first two extensions, the regret bounds for both \CRP\ and \LPRP\ remain valid with some algorithm modifications. In the third extension, the double logarithmic regret dependence still holds but is multiplied by the number of video categories.
More details and all formal results are in \Cref{apx:extension}.

\xhdr{Platform with state-dependent utility function}
Recall that our baseline model 
assumes that the platform's utility function 
$\SUtility$ is 
state-independent, i.e., $\SUtility(\theta, a) \equiv a$.
We can relax this assumption and consider a more general state-dependent platform's utility function $\SUtility:[\totalstate]\times\cA \rightarrow\reals$.
Additionally, we assume that $\SUtility(\realizedstatei, 0) = 0$
and $\SUtility(\realizedstatei, 1) \in [0, 1]$
for every state $\realizedstatei\in[\totalstate]$.
In this more general model, it can be verified that 
the regret guarantees in \Cref{thm:loglog T regret upper bound} and \Cref{thm:log T}
continue to hold for 
modified versions of \Cref{alg:loglog T} and \Cref{alg:log T}.
We also note that our state-dependent utility formulation also captures a scenario where the platform is a benevolent player that aims to maximize the social welfare (see \Cref{apx:extension} for detailed discussions).

\xhdr{Users with misspecified beliefs}
Our algorithm and results can be extended to the setting
where users have misspecified beliefs for the underlying state realization~\citep{AC-16}.
In particular, we allow that the users have a misspecified prior belief $\prior\primed\in\stochastic([\totalstate])$, which is unknown to the platform, while the platform has the prior belief $\prior\in\stochastic([\totalstate])$, which may not be necessarily the same as the users' prior belief $\prior\primed$.
In this setup, we can show that our \Cref{alg:loglog T} and \Cref{alg:log T}, and corresponding regret guarantees continue to hold.

\label{sec:which-video}

\xhdr{Which video to display}
Our results can also be generalized to a setting where there are different categories of videos, and each category could include videos that are in a particular genre or style. 
For example, sports video can be categorized based on the specific type of sport, e.g., soccer, basketball, baseball, American football, etc; or it can be categorized into sub-genres based on the content and purpose of the videos, e.g., highlights and clips, training and tutorials, interviews and profiles, etc.
Meanwhile, each category may include videos that have different characteristics (i.e., states). 
The users not only have different preferences over different states of the videos, but also have preferences over different video categories (e.g., a basketball fan may prefer watching videos in the basketball category). 
The platform initially does not know the user's preferences and can only display one video to the users at each time round to learn the user's preferences.\footnote{In practice, the platform may  display an assortment of videos to the user, 
with the user's choice of which video to watch potentially following a particular choice model.
However, to focus on how the platform can leverage its information advantage about the video state to persuade the user to watch the video, we simplify the scenario by abstracting away from the assortment modeling. 
Instead, we assume that the platform can only display one video at a time.}
The platform's goal is to maximize cumulative payoff. This involves identifying the video category with the highest payoff and optimizing recommendations within that category to achieve the best results. We show that our \CRP\ algorithm can still guarantee a double logarithmic expected regret of $O(\videoNum\log\log T)$, where $\videoNum$ represents the number of categories.

\section{Conclusions and Future Work}
\label{sec:conclusion}

In this paper, we have studied the online Bayesian recommendation
problem with featuring a two-sided information asymmetry 
where the platform knows the payoff-relevant state
but does not know the user's preference (and belief),
and the user knows his preference but is uncertain 
about the payoff-relevant state. 
Focusing on policies that minimize the Stackelberg regret,
we present two algorithms. 
The first algorithm 
is a conservative recommendation policy (\CRP).
We show that this algorithm can achieve $O(\log\log \totaltime)$ regret when the platform knows the user's ordinal preference over the states. Moreover, this 
algorithm can also be readily adapted to the setting 
with unknown ordinal preference. In particular, the same regret $O(\log\log \totaltime)$ can be achieved when the user's preference is affine with respect to the state, and regret $O(\totalstate2^{\totalstate-1}\cdot \log\log \totaltime)$ can be achieved for arbitrary preference. 
Our second algorithm is a linear programming-based algorithm
(\LPRP) that utilizes the problem structure and can achieve
$\polymregret$ regret,
which is more desired when the 
number of states $m$ is large and 
the user's preference is arbitrary. 

\label{apx:future direction}
Our research opens a number of interesting and challenging  questions for future research.

\xhdrQ{Better regret dependency on state space for unknown ordinal preference}
Our lower bound only establishes regret dependency 
on the time horizon $T$, however, it does not rule out the possibility to design an algorithm with achieving $O(\poly(m)\cdot \log\log T)$ regret for unknown ordinal preference setting.
Thus, for this setting, it would be interesting to explore whether one can tighten up the lower bound, or design an algorithm whose regret has only polynomial dependency on the number of states $m$ with still double-logarithmically depending on time horizon $T$. 
Making progress in this direction likely requires 
a more judicious characterization of the 
underlying geometry of our online problem. 

\xhdr{Extension with users' heterogeneity}
In the paper, we focus on the setting where the platform
is interacting with users that have the same preference and the same belief. 
Yet, in some applications, users may be heterogeneous and have different preferences or beliefs over the payoff-relevant states.
Thus, an important extension of our problem with significant practical implications would be to consider a setting that captures the users' heterogeneity. 
In particular, one actionable extension to capture users' heterogeneity is to associate each incoming user at time $t$ with a context $\context_t\in\contextSpace\subseteq \R^d$ (observed by the platform) where $\contextSpace$ is the context space.
The utility function of the user with the context $\context$ is given by $\RUtility: [\totalstate] \times \cA \times \contextSpace \rightarrow \R$.
For example, this utility function could be linear in the context, namely, $\RUtility(i, a, \context_t) 
= \unknownVec_{i, a}^\top\context_t$ where $\{\unknownVec_{i, a}\}_{i\in[\totalstate], a\in\cA}$ with each $\unknownVec_{i, a}\in\R^d$ are unknown vectors associated with each state action pair that the platform needs to learn through the interactions with the users.
With this formulation, our problem then connects to the well-studied contextual dynamic pricing literature. In this line of literature without market noise \citep{LRV-18,LS-18,CLP-20,LLS-21}, the unknown valuation $v$ of the buyer for the product is oftentimes assumed to be a linear function of the context $\context$ (including product and customer features/characteristics), namely, $v = \unknownVec^\top \context$ with some unknown vector $\unknownVec$.
Notably, the pricing algorithm developed for this problem is a non-trivial geometric extension of the conservative binary search proposed in \cite{KL-03}.
Building on the algorithmic insights we developed in \Cref{sec:loglog T}, we believe that it is possible to develop a counterpart contextual algorithm in our setting.
Given that this may require substantial efforts and is also beyond the scope of this work, we leave this practically important but also technically interesting problem as future work.

\xhdr{Generalized setting with multiple signals}
In our designed algorithms (see, e.g., \Cref{alg:loglog T}), the signaling scheme at each time round only uses at most two signals. 
This is largely motivated by the fact that (i) the users have binary action; (ii) and we consider a population of homogeneous users who share the same utilities and beliefs.
With these facts, the two signals suffice for the offline problem when the platform knew the user's preference and belief.
Yet, in practice, we may often observe that the platform uses more than two signals with more involved signaling structure (e.g., YouTube usually displays a rank- or order-list of videos to the users). 
The potential reasons behind this observation could be that 
(i) the users may usually follow a choice model that they have multiple actions that they can take; 
(ii) each user may have a private type, and the platform may need to design the signaling scheme without observing the user type.
It can be shown that either one of these reasons can lead to the optimal signaling scheme with possibly multiple signals and more involved signaling structure.
For example, the offline counterpart of the setting with private types corresponds to a public persuasion problem where the platform designs a public signaling scheme for all users with different types \citep{DX-17,can-22}, and the online problem of this setting requires learning the users' type distribution.
Exploring the problem that requires general signaling structure may need a different set of algorithmic ideas, thus, we leave this direction as future work.

\bibliography{mybib}

\renewcommand{\theHchapter}{A\arabic{chapter}}
\renewcommand{\theHsection}{A\arabic{section}}
\newpage
\ECSwitch
\ECDisclaimer

\section{The Correspondence Table}
\label{apx:table}

\begin{table}[H]
\renewcommand{\arraystretch}{1.0}
\centering
\begin{tabular}{|c|c|}
\toprule 
{\bf Recommendation}          & {\bf Persuasion} 
\\
\midrule
video characteristic/quality  &     state       \\ \hline
platform  &     sender         \\ \hline
user      &     receiver       \\ \hline
watch/accept   & action 1   \\ \hline
not watch/skip & action 0  \\ \hline
recommendation policy  (how to recommend) &     signaling scheme       \\ \hline
general/binary recommendation levels    &     general/binary signal space      \\ \hline
recommend watch/accept   & send signal 1 \\ \hline
recommend not watch/skip & send signal 0 \\ 
\hline
\end{tabular}
\vspace{10pt}
\caption{Correspondence between the terminology used in the recommendation and the terminology used in persuasion.}
\label{table:correspondence table}
\end{table}

\section{Connections to the (Contextual) Dynamic Pricing}
\label{apx:dynamic pricing}

\newcommand{\valUB}{\bar{v}}
\newcommand{\probFn}{f}
\newcommand{\CPP}{\texttt{ConPP}}
\newcommand{\CheckPurchase}{\texttt{CheckPurchase}}
\newcommand{\threshold}{\tau}
\newcommand{\peakthreshold}{\threshold\primed}

This section is presented as an auxiliary extension that connects our online Bayesian recommendation framework to variants of (contextual) dynamic pricing and highlights the broader applicability of our approach.\footnote{It is not needed for any result in the main paper: all theorems, proofs, and guarantees stated in the main text are self-contained and do not depend on this section.}
The purpose of this section is to provide additional intuition and to show how the ``Exploring Phase I'' (conservative exploration) idea can be applied to settings with censored or partial feedback. We emphasize that the pricing model studied here is stylized and is included primarily to make the conceptual connection transparent.

\subsection{Dynamic Pricing with Probabilistic Feedback}
\label{apx:subsec dp probabilistic feedback}
Motivated by \Cref{lem:dynamic pricing reduction}, this subsection provides an extension that connects our online Bayesian recommendation problem (with binary state) to a variant of single-item dynamic pricing. We model probabilistic feedback: after posting a price $\price$, the seller learns whether the buyer's value exceeds $\price$ (i.e., purchase/no-purchase) only with probability $f(\price)$, which can be interpreted as a platform-induced exposure/impression probability that depends on price competitiveness or visibility policies. This captures censored learning under exposure constraints and differs from the classical dynamic pricing model of \citet{KL-03}, where the seller always observes the outcome whenever a buyer arrives. We show that standard $O(\log\log\totaltime)$ pricing mechanisms can suffer substantially larger regret under such censoring, whereas the key ``Exploring Phase I'' idea in \CRP\ can be adapted to address this issue.

The formal definition of dynamic pricing problem with probabilistic feedback is as follows.

\begin{definition}[Dynamic pricing with probabilistic feedback]
\label{def:dynamic pricing with probabilistic feedback}
In this problem, the buyer has an unknown fixed valuation $\val \in [0, 1]$. The seller can post a price $\price_t\in [0, 1]$ at each time round~$t$. The feedback structure and the revenue are as follows: there exists right-continuous function~$\probFn: [0, 1] \rightarrow [0, 1]$ that is known by the seller. Given a price $\price_t\in[0, 1]$, with probability $\probFn(\price_t)$, the seller observes a binary feedback on $\indicator{\price_t\le \val}$ and if $\price_t\le v$, the seller obtains revenue $\price_t$, otherwise collects zero revenue; with probability $1 - \probFn(\price_t)$, the seller observes nothing and collects zero revenue. Thus, the seller's expected revenue at time $t$ is given by $\price_t \probFn(\price_t) \cdot \indicator{\price_t \le \val}$. The goal of the seller is to design a pricing algorithm that has low expected regret compared against the optimal fixed price $\price^*=\argmax_{\price\in[0, 1]} \price \probFn(\price)\cdot \indicator{\price\leq \val}$ that she would like to choose if she knew the buyer's value $\val$. 
\end{definition}
Notice that when $\probFn(p)\equiv 1$ for any price $p$, the above problem degenerates to the dynamic pricing problem with fixed unknown valuation studied in \cite{KL-03}.

One possible practical motivation for probabilistic feedback comes from a common form of censored/exposure-driven feedback that arises when a seller's posted price does not deterministically generate an informative interaction. Consider a seller who is deciding how to price her product on a third-party platform (e.g., Amazon or eBay), where many other sellers offer similar items. To optimize the experience for both buyers and third-party sellers, the platform may prioritize goods that are more ``competitively'' or ``neutrally'' priced: if the price $\price$ set by a seller is too high (resp.\ too low), the induced buyer surplus (resp.\ seller surplus) might be poor at that trading price, and the platform may reduce the product's impressions, ranking, or eligibility for certain traffic. \footnote{There is much empirical evidence showing that price can indeed impact the visibility of the product \citep{AmazonPriceVisibility}.}

In this environment, the seller only occasionally receives informative feedback, namely, whether 
$\price \le v$ via a purchase/no-purchase outcome, because such feedback is observed only when the product is actually shown to buyers. Equivalently, the probability $\probFn(\price)$ that the product gets exposed depends on the posted price, and conditional on exposure the seller observes whether a transaction occurs. This is exactly the probabilistic-feedback model in \Cref{def:dynamic pricing with probabilistic feedback}, where after posting $\price$ the seller learns the purchase indicator only with probability $\probFn(\price)$ (capturing exposure/visibility as a function of price competitiveness).

We note that a direct application of the conservative binary search proposed by \cite{KL-03} can lead to linear regret for the above problem even if the function $\probFn$ is increasing in price $\price$.
Instead, using an idea similar to {\CRP}, we design {\CPP} (\Cref{alg:loglog T dynamic pricing with probabilistic feedback}) that achieves $O(\log\log T)$ regret without any structural assumption on the 
probability function $\probFn$. To formally describe {\CPP}, we introduce three useful auxiliary notations. First, for every threshold $\threshold\in[0, 1]$, with slight abuse of notation, we define $\SExpUtility(\threshold)$ as the optimal expected revenue when the buyer's value is $\threshold$, i.e., 
\begin{align*}
    \SExpUtility(\threshold) = \max_{\price\in[0, 1]} \price\probFn(\price)\cdot \indicator{\price\leq \threshold}~.
\end{align*}
As a sanity check, $\SExpUtility(\threshold)$ is right-continuous, weakly increasing in threshold $\threshold$ and is between 0 and 1 for all $\threshold$. 
By definition, we can see that the per-round expected revenue in the optimum offline benchmark is $\SExpUtility(\val)$. For any threshold $\threshold\in[0, 1]$, with slight abuse of notation, we also define
\begin{align*}
    \price(\threshold) = \argmax_{\price\in[0, 1]} \price\probFn(\price)\cdot \indicator{\price\leq \threshold}~.
\end{align*}
Finally, similar to {\CheckPersu} used in our base model, we define a subroutine {\CheckPurchase} as follows. The input of subroutine {\CheckPurchase} is a per-round expected revenue $\SExpUtility\primed\in[0, 1]$. The goal is to check whether $\SExpUtility(\val)$ is weakly larger than input $\SExpUtility\primed$. To obtain this information, the subroutine first computes threshold $\threshold\primed = \min\{\threshold\in[0, 1]:\SExpUtility(\threshold) \geq \SExpUtility\primed\}$ that is the smallest threshold $\threshold$ with $\SExpUtility(\threshold) \geq \SExpUtility\primed$. (Threshold $\threshold\primed$ is well-defined since the function $\SExpUtility(\threshold)$ is right-continuous and weakly increasing.) Then the subroutine keeps posting a fixed price $\price\primed = \price(\threshold\primed)$ until a binary purchase feedback has been returned (or rounds exhausted). By construction, we have $\SExpUtility(\val) \geq \SExpUtility\primed$ if and only if the buyer purchases at price $\price\primed$. Moreover, the regret of subroutine {\CheckPurchase} can be bounded in the following lemma (similar to \Cref{lem:persuasive check regret} for our base model).
\begin{lemma}
\label{lem:dynamic pricing check regret}
    Given any $\SExpUtility\primed\in[0, 1]$, the expected regret of subroutine {\CheckPurchase} with input $\SExpUtility\primed$ is at most 
    $\frac{\SExpUtility(\val)}{\SExpUtility\primed} - \indicator{\SExpUtility\primed \leq \SExpUtility(\val)}$.
\end{lemma}
\begin{proof}
Let $\issuedrounds$ be the number of rounds used in subroutine {\CheckPurchase} with input $\SExpUtility\primed$.
We start by upper bounding $\expect{\issuedrounds}$.
Note that subroutine {\CheckPurchase} returns binary feedback with probability $\probFn(\price\primed)$ in each round, where price $\price\primed = \price(\threshold\primed)$ and threshold $\threshold\primed = \min\{\threshold\in[0, 1]:\SExpUtility(\threshold) \geq \SExpUtility\primed\}$ by construction. 
Therefore, 
\begin{align*}
    \expect{\issuedrounds} = \frac{1}{\probFn(\price\primed)}
    = \frac{\price\primed}{\SExpUtility(\threshold\primed)}
    \leq 
    \frac{1}{\SExpUtility\primed}~,
\end{align*}
where the second equality holds due to the definition of $\price\primed$, $\threshold\primed$ and $\SExpUtility(\cdot)$; and the inequality holds since $\price\primed \leq 1$ and $\SExpUtility(\threshold\primed) \geq \SExpUtility\primed$. Consequently, the expected regret can be upper bounded by
\begin{align*}
\expect{\issuedrounds} \cdot 
\left(\SExpUtility(\val)-\SExpUtility(\threshold\primed)\cdot \indicator{\SExpUtility(\threshold\primed) \leq \SExpUtility(\val)}\right)
& {} = 
\expect{\issuedrounds} \cdot 
\left(\SExpUtility(\val)-\SExpUtility(\threshold\primed)\cdot \indicator{\SExpUtility\primed \leq \SExpUtility(\val)}\right)
\\
& {} \leq 
\frac{\SExpUtility(\val)}{\SExpUtility\primed} - \indicator{\SExpUtility\primed \leq \SExpUtility(\val)}~,
\end{align*}
where the equality holds since $\indicator{\SExpUtility(\threshold\primed) \leq \SExpUtility(\val)} = \indicator{\SExpUtility\primed \leq \SExpUtility(\val)}$ by construction; and the inequality holds since $\expect{\issuedrounds} \leq \frac{1}{\SExpUtility\primed}$ and $\SExpUtility(\threshold\primed) \geq \SExpUtility\primed$ as we argued above.
\end{proof}

The formal description of {\CPP} is given in \Cref{alg:loglog T dynamic pricing with probabilistic feedback} and its regret guarantee is in \Cref{thm:loglog T dynamic pricing with probabilistic feedback}. Both the algorithm and the analysis are almost the same as the ones (\Cref{alg:loglog T} and \Cref{thm:loglog T regret upper bound}) for our base model. Nonetheless, we include them for completeness.

\begin{theorem}
\label{thm:loglog T dynamic pricing with probabilistic feedback}
For the dynamic pricing problem with probabilistic feedback (\Cref{def:dynamic pricing with probabilistic feedback}), 
there exists an algorithm (\CPP, see \Cref{alg:loglog T dynamic pricing with probabilistic feedback}) that has the expected regret of  at most $O(\log\log T)$.
\end{theorem}

\begin{proof}
We analyze  
the expected 
regret in 
exploring phase I,
exploring phase II,
and exploiting phase separately.
We first assume that \CPP\
finishes exploring phases I and II before $\totaltime$
rounds are exhausted.
A similar argument follows for the other case where 
exploring phase I or exploring phase II is completed 
due to the exhaustion of rounds.

\xhdr{Exploring phase I}
Let $K = -\lceil\log(\SExpUtility(\val))\rceil$.
By definition, 
$\CheckPurchase(2^{-k})=\False$
for $k \in [K - 1]$,
and 
$\CheckPurchase(2^{-K})) = \True$.
Thus, at the end of exploring phase I,
$\SExpUtilityUnderbar$ 
 is $2^{-K}$,
 and 
there are $K$ iterations
in the while loop.
For each iteration $k\in[K]$, 
subroutine $\CheckPurchase(2^{-k})$ 
is called once.
By \Cref{lem:dynamic pricing check regret}, the total expected regret is 
\begin{align*}
    \sum\nolimits_{k \in [K]}
    \frac{\SExpUtility(\val)}{2^{-k}
    }
    \overset{}{\leq} 
    \sum\nolimits_{k \in [K]}
    \frac{2^{-(K-1)}}{
    2^{-k}
    }
    =
    \sum\nolimits_{k\in[K]}
    2^{-(K - k - 1)}
    = O(1)
\end{align*}

\xhdr{Exploring phase II}
By construction, 
there are $O(\log\log \totaltime)$
iterations in the while loop.
Thus, it is sufficient to show 
the expected regret in each iteration 
is $O(1)$.

In each iteration $k$,
let $\ell\primed\in[S]$
be the index 
that subroutine $\CheckPurchase(L + \ell\primed\perc L)$ returns $\False$.
Invoking \Cref{lem:dynamic pricing check regret}, the expected regret in iteration $k$ 
is at most 
\begin{align*}
    &~\sum\nolimits_{\ell=1}^{\ell\primed - 1}
    \left(
    \frac{\SExpUtility(\val)}{L +\ell\perc L
    }
    -1
    \right)
    +
    \frac{\SExpUtility(\val)}{
L +\ell\primed\perc L
    }
\overset{(a)}{\leq} 
\sum\nolimits_{\ell=1}^{\ell\primed - 1}
    \left(
    \frac{R}{
L
    }
    -1
    \right)
    +
    \frac{R}{
L
    }
\overset{(b)}{\leq}
(S - 1) \frac{R - L}{L} + 2
\overset{(c)}{\leq}
~ \frac{(R - L)^2}{\perc L^2} + 2
\end{align*}
where 
inequality~(b) holds 
since $\SExpUtility(\val) \leq R$;
inequality~(c) holds since 
$\ell\primed \leq S$ and $R \leq 2 L$;
and
inequality~(d) holds since $S = 
\lfloor\frac{R - L}{\perc L}\rfloor$.

We finish this part 
by showing
$R - L \leq \sqrt{2\perc} L$
by induction.
Let $L^{(k)}, R^{(k)}$, $\stepperc^{(k)}$
and $\perc^{(k)}$
be the value of $L, R, \stepperc,\perc$
in each iteration $k$.
The claim is satisfied for
iteration $k = 1$,
since $R^{(1)} - L^{(1)} = 2\SExpUtilityUnderbar - \SExpUtilityUnderbar 
= L^{(1)}$
and $\perc^{(1)} = \sfrac{1}{2}$.
Suppose the claim holds for iteration $k - 1$.
Now, for iteration $k$, we know that 
$R^{(k)} - L^{(k)} = \perc^{(k-1)} L^{(k - 1)}
\leq \perc^{(k-1)} L^{(k)}
=
\sqrt{\stepperc^{(k)}} L^{(k)}
=
\sqrt{2\perc^{(k)}} L^{(k)}$,
which finishes the induction.

\xhdr{Exploiting phase}
In this phase, we know that 
$\price\primed \leq \val$ and 
$\SExpUtility(\price\primed) =  \SExpUtility(\threshold\primed) \geq 
\SExpUtility(\val) - \sfrac{1}{\totaltime}$,
which concludes the proof.
\end{proof}

\begin{algorithm}
\caption{\texttt{Con}servative \texttt{P}ricing \texttt{P}olicy (\CPP)}
\label{alg:loglog T dynamic pricing with probabilistic feedback}
\linespread{0.8}\selectfont
\SetAlgoLined\DontPrintSemicolon
\KwIn{number of rounds $\totaltime$}
\tcc{exploring phase I
--
identify 
per-round expected utility lower bound $\SExpUtilityUnderbar$
such that 
$\SExpUtilityUnderbar
\leq \SExpUtility(\val) 
\leq 2 \SExpUtilityUnderbar$}

Initialize $\SExpUtilityUnderbar \gets \frac{1}{2}$
\\
\While{
    $\CheckPurchase\left(\SExpUtilityUnderbar\right) =
    \False$}{
        $\SExpUtilityUnderbar\gets \frac{\SExpUtilityUnderbar}{2}$
    }
\tcc{exploring phase II --
identify a threshold $\threshold\primed$
such that $\SExpUtility(\threshold\primed) \geq \SExpUtility(\val) - \frac{1}{\totaltime}$}
Initialize
$R \gets 2\SExpUtilityUnderbar$,
 $L \gets \SExpUtilityUnderbar$,
$\stepperc \gets 1$
\\
\While{$R - L \geq \frac{1}{\totaltime}$
}{
    $\perc \gets \frac{\stepperc}{2}$,
    $S\gets \lfloor\frac{R - L}{\perc L}\rfloor$,
    $\ell \gets 1$. \\
    \While{
        $\CheckPurchase\left(L + \ell\perc L\right) =
        \True$
    }{
        $R\gets L + \ell\perc L$,
        $L \gets L + (\ell-1)\perc L$,
        $\stepperc \gets \perc^2$,
        $\ell \gets \ell +1$.
    }
}
Set threshold $\threshold\primed \gets 
 \min\{\threshold\in[0, 1]:\SExpUtility(\threshold) \geq \SExpUtility\primed\}$
 and 
price $\price\primed \gets \price(\threshold\primed)$.
\\
\tcc{exploiting phase}
Post pricing $\price\primed$
for all remaining rounds.
\end{algorithm}

\subsection{Comparing to Contextual Dynamic Pricing}
\label{apx:subsec contextual dp}
When the size of state space (i.e., $\totalstate$) is large,
\CRP\
incurs an $O(\totalstate\cdot 2^{\totalstate-1})$ regret dependence,
which may not be ideal.
A natural question is whether we can 
improve the dependence on $\totalstate$
to $\poly(\totalstate)$.
To answer this question, one natural attempt is to 
revisit
the multi-dimension generalization of 
the single-item dynamic pricing problem --
contextual dynamic pricing problem,
in which 
\citet{LS-18,LLS-21} design
a contextual dynamic pricing mechanism with 
$O(\poly(\totalstate)\log\log\totaltime)$ regret.
In the contextual dynamic pricing problem,
the item has $\totalstate$ features,
and buyers have private value $\optval(i)$
for each feature $i$. In each round $t\in[\totaltime]$, 
nature selects a vector 
$(x_t(1), \dots, x_t(\totalstate))\in\reals_{\geq0}^\totalstate$,
and 
the seller wants to sell a new unit of the item 
by setting a price $\price_t$ to buyer $t$,
who will buy the item (and pay $\price_t$) 
if and only if 
$\sum_{i\in[\totalstate]}\optval(i)x_t(i) \geq \price_t$.

The contextual dynamic pricing problem 
shares some similarity to our problem with multiple states.
Specifically,
there is an unknown vector $\{\optval(i)\}$
(resp. $\{\RUtilityDiff(i)\}$),
and the optimal in hindsight benchmarks 
can be formulated as similar linear programs depending
on $\{\optval(i)\}$
(resp. $\{\RUtilityDiff(i)\}$).
Nonetheless, there exist fundamental differences
between the two problems besides the 
probabilistic and limited feedback feature 
mentioned before.
In particular, in each round,
the contextual dynamic pricing mechanism
chooses a price $\price_t$
which is a scalar,
while the online Bayesian recommendation policy
chooses a signaling scheme 
(i.e., a high-dimensional function).
In \citet{LS-18}, the authors obtain $O(\poly(\totalstate)\log\log \totaltime)$
regret by formulating 
the contextual dynamic pricing as 
solving linear programs with a separation oracle.\footnote{In particular, vector $(x_t(1), \dots, x_t(\totalstate))$
is served as 
the separating hyperplane for the oracle.}
However, in our 
problem, 
it is unclear if such a simple separation oracle exists. 
In \Cref{sec:log T},
we introduce an online policy with
$\polymregret$
regret by formulating our problem as 
solving linear programs with a membership oracle.\footnote{Membership oracle is weaker than separation oracle. See more discussion between
the two oracles
in \Cref{sec_logn}.}

\section{Omitted Proofs in Section~\ref{sec:prelim}}
\label{apx:proofs in prelim}

In this section, we 
present the omitted proofs 
of \Cref{lem:persuasive costState},
\Cref{lem:persuasive check correctness},
and \Cref{lem:persuasive check regret}
in \Cref{sec:prelim}.

\persuasiveCostState*
\begin{proof}
When action $1$ is recommended,
the posterior distribution is $\posteriori(1,i) 
= 
\frac{
\prior(\realizedstatei)\scheme(\realizedstatei)
}{
\sum_{\realizedstatej\in[\totalstate]}
\prior(\realizedstatej)\scheme(\realizedstatej)
}$.
Thus, the \Receiver\ 
takes action 1 if and only if 
\begin{align*}
\frac{1}{\sum_{\realizedstatej\in[\totalstate]}\prior(\realizedstatej)\scheme(\realizedstatej)}
\sum_{\realizedstatei\in[\totalstate]}
\RUtility(\realizedstatei, 1)
\prior(\realizedstatei)\scheme(\realizedstatei)
\geq 
\frac{1}{\sum_{\realizedstatej\in[\totalstate]}\prior(\realizedstatej)\scheme(\realizedstatej)}
\sum_{\realizedstatei\in[\totalstate]}
\RUtility(\realizedstatei, 0)
\prior(\realizedstatei)\scheme(\realizedstatei)
\end{align*}
Rearranging the terms finishes the proof.
\end{proof}

\persuasiveCheckCorrectness*
\begin{proof}
By construction, Procedure~\ref{alg:persuasive check}
returns \True\ if $\signali = 1 = \actioni$,
which is exactly the definition of incentive compatibility.
Similarly,  Procedure~\ref{alg:persuasive check}
returns \False\ if $\signali = 1 = 1 - \actioni$
or $\signali = 0 = 1 - \actioni$.
The correctness of the former case holds due to
the definition of incentive compatibility.
To see the correctness of the latter case,
note that when action 0 is recommended, 
the \Receiver\ takes action 1 if and only if 
$\sum_{\realizedstatei\in[\totalstate]}
\costStatei(1 - \scheme(\realizedstatei)) \geq 0$.
Hence,
\begin{align*}
\sum_{\realizedstatei\in[\totalstate]}
\costStatei\scheme(\realizedstatei)
\leq 
\sum_{\realizedstatei\in[\totalstate]}
\costStatei
< 0
\end{align*}
where the last inequality holds due to \Cref{asp:receiver utility 2}.
Invoking \Cref{lem:persuasive costState}
finishes the proof.
\end{proof}

\persuasiveCheckRegret*
\begin{proof}
Let $\issuedrounds$
be the number of rounds used in 
Procedure~\ref{alg:persuasive check}.
We start by
upper bounding $\expect{\issuedrounds}$.
Note that 
Procedure~\ref{alg:persuasive check} 
returns if action 1 is recommended,
which happens with probability 
$\sum_{\realizedstatei}\prior(\realizedstatei)\scheme(\realizedstatei)$ in each round.
Thus, $\expect{\issuedrounds} \leq \frac{1}{
\sum_{\realizedstatei}
\prior(\realizedstatei)\scheme(\realizedstatei)}$,
and the expected regret is at most 
\begin{align*}
\expect{\issuedrounds} 
(\SExpUtility(\optoff)-\SExpUtility(\scheme))
\leq 
\frac{\SExpUtility(\optoff)}{
\sum_{\realizedstatei}\prior(\realizedstatei)\scheme(\realizedstatei)}
-
\frac{\SExpUtility(\scheme)}{
\sum_{\realizedstatei}\prior(\realizedstatei)\scheme(\realizedstatei)}
\leq
\frac{\SExpUtility(\optoff)}{
\sum_{\realizedstatei}\prior(\realizedstatei)\scheme(\realizedstatei)
}
-
\indicator{\text{$\scheme$ is \IC}}
\end{align*}
where the last inequality holds 
since $\SExpUtility(\scheme) 
\geq 
(\sum_{\realizedstatei}\prior(\realizedstatei)\scheme(\realizedstatei))
\cdot 
\indicator{\text{$\scheme$ is \IC}}$.
\end{proof} 
\section{Omitted Proofs and Algorithm in Section~\ref{sec:unknwon order}}
\label{apx:proofs unknwon order}

\label{apx-missing-algo}

\newcommand{\AlgoTwoNameUnknown}{\texttt{Con}servative \texttt{R}ecommendation \texttt{P}olicy (\CRP) for unknown order preference}
\begin{algorithm}
\caption{\AlgoTwoNameUnknown}
\label{alg:loglog T w unknown order}
\linespread{0.8}\selectfont
\SetAlgoLined\DontPrintSemicolon
\KwIn{number of rounds $\totaltime$, 
number of states $\totalstate$,
prior distribution $\prior$}
Initialize the set $\permus$ to 
contain all possible orders described as above.
\\
\tcc{exploring phase I
--
identify 
$\SExpUtilityUnderbar$
such that 
$\SExpUtilityUnderbar
\leq \SExpUtility(\optoff) 
\leq 2 \SExpUtilityUnderbar$}
Initialize $\SExpUtilityUnderbar \gets \frac{1}{2}$
\\
\While{\True
}{
    \If{there exists $\permu\in \permus$ such that 
    $\CheckPersu\left(\scheme^{(\permu,\SExpUtilityUnderbar)}\right) =
    \True$}{
        $\permus\gets
        \left\{\permu \in\permus:
        \CheckPersu\left(\scheme^{(\permu,\SExpUtilityUnderbar)}\right) = 
        \True
        \right\}$\\
        \cc{break}
    }
    \uElse{
    $\SExpUtilityUnderbar\gets \frac{\SExpUtilityUnderbar}{2}$
    }
}
\tcc{exploring phase II --
identify a signaling scheme $\scheme\primed$
such that $\SExpUtility(\scheme\primed) \geq \SExpUtility(\optoff) - \frac{1}{\totaltime}$}
Initialize
$R \gets 2\SExpUtilityUnderbar$,
 $L \gets \SExpUtilityUnderbar$,
$\stepperc \gets 1$
\\
\While{$R - L \geq \frac{1}{\totaltime}$
}{
    $\perc \gets \frac{\stepperc}{2}$,
    $S\gets \lfloor\frac{R - L}{\perc L}\rfloor$
    \\
    \For{$\ell = 1, 2, \dots, S$}{
        \If{there exists $\permu\in \permus$ such that 
        $\CheckPersu\left(\scheme^{(\permu,L + \ell\perc L)}\right) =
        \True$}{
            $\permus\gets \left\{\permu\in\permus:
            \CheckPersu\left(\scheme^{(\permu,L + \ell\perc L)}\right) =
            \True
        \right\}$
        }\uElse{
            $R\gets L + \ell\perc L$,
            $L \gets L + (\ell-1)\perc L$,
            $\stepperc \gets \perc^2$\\
            \cc{break}
        }
    }
}
Set $\scheme\primed \gets \scheme^{(\permu,L)}$
for an arbitrary $\permu\in \permus$
\\
\tcc{exploiting phase}
Use signaling scheme $\scheme\primed$
for all remaining rounds.
\end{algorithm}

\oglogTregretupperboundunknownorder*
\begin{proof}
We analyze  
the expected 
regret in 
exploring phase I,
exploring phase II,
and exploiting phase separately.
We first assume that \Cref{alg:loglog T w unknown order}
finishes exploring phase I and II before $\totaltime$
rounds are exhausted.
A similar argument follows for the other case where 
exploring phase I or exploring phase II is completed 
due to the exhaustion of rounds.

\xhdr{Exploring phase I}
Let $K = -\lceil\log(\SExpUtility(\optoff))\rceil$.
By definition, 
$\CheckPersu(\scheme^{(\permu,2^{-k})})=\False$
for all $\permu\in\permus$
and $k \in [K - 1]$,
and 
$\CheckPersu(\scheme^{(\permuTrue,2^{-K})}) = \True$.
Thus, at the end of exploring phase I,
$\SExpUtilityUnderbar$ 
 is $2^{-K}$,
 and 
there are $K$ iterations
in the while loop.
For each iteration $k\in[K]$, 
$\CheckPersu(\scheme^{(\permu,2^{-k})})$ 
is called 
for every $\permu\in\permus$.
By \Cref{lem:persuasive check regret}, the total expected regret is 
\begin{align*}
    \sum_{k \in [K]}
    \sum_{\permu\in\permus}
    \frac{\SExpUtility(\optoff)}{
    \sum_{\realizedstatei\in[\totalstate]}
\prior(\realizedstatei)\scheme^{(\permu,2^{-k})}(\realizedstatei)
    }
    \overset{(a)}{\leq} 
    \sum_{k \in [K]}
    \sum_{\permu\in\permus}
    \frac{2^{-(K-1)}}{
    2^{-k}
    }
    =
    |\permus|\sum_{k\in[K]}
    2^{-(K - k - 1)}
    \leq
    4\cdot \totalstate!
\end{align*}
where the denominator 
in the right-hand side 
of inequality~(a) is due to the 
construction of $\scheme^{(\permu,2^{-k})}$.

\xhdr{Exploring phase II}
By construction, 
there are $O(\log\log \totaltime)$
iterations in the while loop.
Thus, it is sufficient to show 
the expected regret in each iteration 
is $O(\totalstate2^{\totalstate-1})$.

In each iteration $k$,
for every $\permu \in \permus$,
let $\ell\primed\in[S]$
be the smallest index 
that the signaling scheme
$\scheme^{(\permu,L + \ell\primed\perc L)}$
is not \IC.
The expected regret in iteration $k$ 
for $\permu\in \permus$
is at most 
\begin{align*}
    &~\sum_{\ell=1}^{\ell\primed - 1}
    \left(
    \frac{\SExpUtility(\optoff)}{
    \sum_{\realizedstatei\in[\totalstate]}
\prior(\realizedstatei)\scheme^{
(\permu,
L +\ell\perc L)}(\realizedstatei)
    }
    -1
    \right)
    +
    \frac{\SExpUtility(\optoff)}{
    \sum_{\realizedstatei\in[\totalstate]}
\prior(\realizedstatei)
\scheme^{
(\permu,
L +\ell\primed\perc L)}
(\realizedstatei)
    }
    \\
    \overset{(a)}{=} &
    \sum_{\ell=1}^{\ell\primed - 1}
    \left(
    \frac{\SExpUtility(\optoff)}{
L +\ell\perc L
    }
    -1
    \right)
    +
    \frac{\SExpUtility(\optoff)}{
L +\ell\primed\perc L
    }
\overset{(b)}{\leq} 
\sum_{\ell=1}^{\ell\primed - 1}
    \left(
    \frac{R}{
L
    }
    -1
    \right)
    +
    \frac{R}{
L
    }
\overset{(c)}{\leq}
(S - 1) \frac{R - L}{L} + 2
\overset{(d)}{\leq}
~ \frac{(R - L)^2}{\perc L^2} + 2
\end{align*}
where equality~(a) holds due to 
the construction of 
$\scheme^{
(\permu,
L +\ell\perc L)}$
and $\scheme^{
(\permu,
L +\ell\primed\perc L)}$;
inequality~(b) holds 
since $\SExpUtility(\optoff) \leq R$;
inequality~(c) holds since 
$\ell\primed \leq S$ and $R \leq 2 L$;
and
inequality~(d) holds since $S = 
\lfloor\frac{R - L}{\perc L}\rfloor$.

We finish this part 
by showing
$R - L \leq \sqrt{2\perc} L$
by induction.
Let $L^{(k)}, R^{(k)}$, $\stepperc^{(k)}$
and $\perc^{(k)}$
be the value of $L, R, \stepperc,\perc$
in each iteration $k$.
The claim is satisfied for
iteration $k = 1$,
since $R^{(1)} - L^{(1)} = 2\SExpUtilityUnderbar - \SExpUtilityUnderbar 
= L^{(1)}$
and $\perc^{(1)} = \sfrac{1}{2}$.
Suppose the claim holds for iteration $k - 1$.
Now, for iteration $k$, we know that 
$R^{(k)} - L^{(k)} = \perc^{(k-1)} L^{(k - 1)}
\leq \perc^{(k-1)} L^{(k)}
=
\sqrt{\stepperc^{(k)}} L^{(k)}
=
\sqrt{2\perc^{(k)}} L^{(k)}$,
which finishes the induction.

\xhdr{Exploiting phase}
By construction, $\permus$ is not empty
at the end of exploring phase II, and 
thus signaling scheme $\permu\primed$
is well-defined.
Additionally, we know that 
$\scheme\primed$ is \IC\ and 
$\SExpUtility(\scheme\primed) \geq 
\SExpUtility(\optoff) - \sfrac{1}{\totaltime}$,
which concludes the proof.
\end{proof}

\section{Missing Technical Details for Alternative LP-based Algorithm}
\label{apx:proofs in log T}

We present a formal description of \LPRP\ below (see \Cref{alg:log T}). 
In this algorithm,
three specific subclasses 
of direct signaling schemes
$\{\schemeanchorpair\}$, $\{\schemeanchorpairk\}$,
and $\{\schemestartingpoint\}$
are used,
whose constructions are as follows.
We note that the
aforementioned signaling scheme subset 
$\anchorsetHat$ in the
algorithm overview is 
not explicitly defined in \LPRP.
Its formal definition is 
$\anchorsetHat\triangleq
\{\schemeanchorpair \text{ induced by}
(\realizedstateanchori,\realizedstateanchorj)\in
\anchorset\}$.

\mybox{
\setlength{\itemsep}{1pt}
  \setlength{\parskip}{1pt}
  \setlength{\parsep}{1pt}
  \openup 0.3em
\vspace{1mm}
\begin{itemize}[leftmargin=*]
\setlength{\itemsep}{1pt}
\item 
Given $(\realizedstateanchori,
\realizedstateanchorj,\SExpUtilityUnderbar)$,
let 
$\schemeanchorpair$ denote a direct signaling scheme 
with
    $\schemeanchorpair(\realizedstateanchori)= 1$,
    $\schemeanchorpair(\realizedstateanchorj)= \frac{\SExpUtilityUnderbar}
    {\prior(\realizedstateanchorj)}$ if $\realizedstateanchorj\not=\realizedstateanchori$,
    and 
    $\schemeanchorpair(\realizedstatei)= 0$
    for every $\realizedstatei\not
    \in\{\realizedstateanchori,\realizedstateanchorj\}$.

\item 
Given $(\realizedstateanchori,
\realizedstateanchorj,\SExpUtilityUnderbar,\realizedstatei)$,
let 
$\schemeanchorpairk$ denote a direct signaling scheme 
with
    $\schemeanchorpairk(\realizedstateanchori)= 1$,
    $\schemeanchorpairk(\realizedstateanchorj)= \frac{\SExpUtilityUnderbar}
    {2\prior(\realizedstateanchorj)}$ 
    if $\realizedstateanchorj\not=\realizedstateanchori$,
    $\schemeanchorpairk(\realizedstatei)=
    \frac{3}{2\totalstate\totaltime}$,
    and 
    $\schemeanchorpairk(\realizedstatej)= 0$
    for every $\realizedstatej\not
    \in\{\realizedstateanchori,\realizedstateanchorj,\realizedstatei\}$.

\item 
Given $(\realizedstateanchori,
\realizedstateanchorj,\SExpUtilityUnderbar,\meaningfulset)$,
let
$\schemestartingpoint$
denote a direct signaling scheme 
with
$\schemestartingpoint(\realizedstateanchori) = \frac{1}{2}
+\frac{1}{16\totalstate^2\totaltime}$,
$\schemestartingpoint(\realizedstateanchorj) =
\frac{\SExpUtilityUnderbar}{8\prior(\realizedstateanchorj)}$
if $\realizedstateanchorj\not=\realizedstateanchori$,
$\schemestartingpoint(\realizedstatei) = \frac{1}{8\totalstate^2\totaltime}$
for every $\realizedstatei\in\meaningfulset\backslash
\{\realizedstateanchori,\realizedstateanchorj\}$,
and $\schemestartingpoint(\realizedstatei) = 0$
for every $\realizedstatei\not\in\meaningfulset$.
\end{itemize}
}

In the remainder of this section, we explain each phase of \Cref{alg:log T}, i.e., \LPRP,
in detail.
By combining the regret analysis in all phases, 
we prove \Cref{thm:log T}
in the end of this section.

\SetKwInput{KWSchemeI}{Notation $\schemeanchorpair$}
\SetKwInput{KWSchemeII}{Notation $\schemeanchorpairk$}
\SetKwInput{KWSchemeStart}{Notation $\schemestartingpoint$}
\begin{algorithm}[ht]
\linespread{0.8}\selectfont
\caption{\texttt{LP}-based \texttt{R}ecommendation \texttt{P}olicy (\LPRP)}
\label{alg:log T}
\SetAlgoLined\DontPrintSemicolon
\KwIn{number of rounds $\totaltime$, 
number of states $\totalstate$,
prior distribution $\prior$,
and 
linear program solver $\lpsolver$ \citep{LSV-18}
with membership oracle access
}

\tcc{exploring phase I}
Initialize $\SExpUtilityUnderbar \gets \frac{1}{2}$
\\
\While{$\SExpUtilityUnderbar\geq\frac{1}{\totalstate^2\totaltime}$
}{
    \If{there exists
    $(\realizedstateanchori,\realizedstateanchorj)
    \in[\totalstate]\times[\totalstate]$
    such that
    $\CheckPersu\left(
    \schemeanchorpair
    \right) =
    \True$}{
        $\anchorset \gets
        \left\{(\realizedstateanchori,\realizedstateanchorj)
    \in[\totalstate]\times[\totalstate]:
        \CheckPersu\left(
       \schemeanchorpair
    \right) =
    \True
        \right\}$\\
        \cc{break}
    }
    \uElse{
    $\SExpUtilityUnderbar\gets \frac{\SExpUtilityUnderbar}{2}$
    }
}
\If{$\anchorset = \emptyset$}{
Set $\scheme\primed:[\totalstate]\rightarrow\{0\}$ to be the signaling scheme
which reveals no information.
\\
Move to exploiting phase.
}
\tcc{exploring phase II}
Initialize
$\meaningfulset \gets \emptyset $
\\
\For{each state pair $(\realizedstateanchori,\realizedstateanchorj)\in\anchorset$}{
$\meaningfulset\gets\meaningfulset\cup
\left\{\realizedstatei\in[\totalstate]:
\realizedstatei=\realizedstateanchori\,\,
\mbox{or}\,\,
\realizedstatei=\realizedstateanchorj\,\,
\mbox{or}\,\,
     \CheckPersu\left(\schemeanchorpairk\right) = \True
\right\}$
}
\tcc{exploring phase III}
\For{for each pair $(\realizedstateanchori,\realizedstateanchorj)
\in\anchorset$
}{
Solve $\scheme^{(\realizedstateanchori,\realizedstateanchorj)
}$
by linear program solver $\lpsolver$ 
for program~\ref{eq:optimum in hindsight within meaningfulset}:
set the interior
point $x\zeroed\gets 
\schemestartingpoint$,
lower-bound radius $r\gets \frac{1}{16\totalstate^2\totaltime}$,
upper-bound radius $R \gets \sqrt{\totalstate}$,
precision $\epsilon\gets \frac{1}{\totaltime}$
and success probability $\delta \gets \frac{1}{\totaltime}$.
}
Set $\scheme\primed$
be best $\scheme^{(\realizedstateanchori,\realizedstateanchorj)
}$
\big(i.e., maximizing $\SExpUtility(^{(\realizedstateanchori,\realizedstateanchorj)
})$
\big)
for all $(\realizedstateanchori,\realizedstateanchorj)\in
\anchorset$
\\
\tcc{exploiting phase}
Use signaling scheme $\scheme\primed$
for all remaining rounds.
\end{algorithm}

\vspace{10pt} 

\xhdr{The analysis of exploring phase I}
We use the following lemma to characterize 
exploring phase~I.
\begin{lemma}[restatement of \Cref{lem:existence of anchor pair main}]
\label{lem:existence of anchor pair}
Suppose $\SExpUtility(\optoff) \geq \frac{1}{\totaltime}$.
Let $\realizedstateanchori =
\argmax_{\realizedstatei\in[\totalstate]}\costStatei$.
When exploring phase I terminates, 
$\SExpUtilityUnderbar \geq
\frac{\SExpUtility(\optoff)}{\totalstate^2}$
and there exists a state $\realizedstateanchorj\in[\totalstate]$
such that $(\realizedstateanchori,\realizedstateanchorj)\in
\anchorset$.
\end{lemma}

\begin{proof}
Let state $\realizedstatej' =
\argmax_{\realizedstatej\in[\totalstate]}
\prior(\realizedstatej)\optoff(\realizedstatej)$,
namely, $\realizedstatej'$
is the state that contributes the most to
$\SExpUtility(\optoff)$.
Consider a direct signaling scheme $\scheme$
with $\scheme(\realizedstateanchori) = 1$,
$\scheme(\realizedstatej') = \frac{\optoff(\realizedstatej')}{\totalstate - 1}$
if $\realizedstatej'\not=\realizedstateanchori$,
and $\scheme(\realizedstatei) = 0$
for every $\realizedstatei\not\in\{
\realizedstateanchori,\realizedstatej'\}$.
We claim that $\scheme$ is \IC. To see this,
note that if $\realizedstatej'\not=\realizedstateanchori$,
\begin{align*}
\sum\nolimits_{\realizedstatei\in[\totalstate]}
\costStatei\scheme(\realizedstatei)
&=
\costState(\realizedstateanchori)
\scheme(\realizedstateanchori)
+
    \costState(\realizedstatej')
    \scheme(\realizedstatej')
    = 
    \costState(\realizedstateanchori) 
+
\frac{1}{\totalstate - 1}
    \costState(\realizedstatej')
    \optoff(\realizedstatej')
    \\
    &\overset{(a)}{\geq}
    \frac{1}{\totalstate-1}
    \sum\nolimits_{\realizedstatei\in[\totalstate]:
    \realizedstatei\not=\realizedstatej'}
    \costState(\realizedstatei)\optoff(\realizedstatei)
+
\frac{1}{\totalstate - 1}
    \costState(\realizedstatej')
    \optoff(\realizedstatej')
    \overset{(b)}{\geq} 0
\end{align*}
where inequality~(a) holds since
$\costState(\realizedstateanchori)\geq
\costState(\realizedstatei)\optoff(\realizedstatei)$
for all $\realizedstatei\in[\totalstate]$ by definition,
and inequality~(b) holds since $\optoff$ is \IC.
A similar argument holds for $\realizedstatej'=\realizedstateanchori$.
Moreover, we know that
\begin{align*}
    \SExpUtility(\scheme) \geq 
    \prior(\realizedstatej')\scheme(\realizedstatej')
    \geq 
    \frac{1}{\totalstate - 1}
     \prior(\realizedstatej')\optoff(\realizedstatej')
     \overset{(a)}{\geq} 
     \frac{1}{\totalstate - 1}
     \frac{1}{\totalstate}
     \sum\nolimits_{\realizedstatei\in[\totalstate]}
     \prior(\realizedstatei)\optoff(\realizedstatei)
     =
     \frac{1}{\totalstate(\totalstate-1)}
     \SExpUtility(\optoff)
\end{align*}
where inequality~(a) holds due to the definition
of state $\realizedstatej'$.

The existence of the \IC\ signaling scheme $\scheme$
constructed above 
implies that when the if-condition (line 3
in \LPRP)
is
satisfied if $\SExpUtilityUnderbar < \frac{\SExpUtility}{\totalstate(\totalstate - 1)}$.
Hence, if $\SExpUtility\geq \frac{1}{\totaltime}$,
when exploring phase I terminates, 
$\SExpUtilityUnderbar \geq
\frac{\SExpUtility(\optoff)}{\totalstate(\totalstate - 1)}
\geq 
\frac{\SExpUtility(\optoff)}{\totalstate^2}$.

Next, we argue the second part of the lemma statement --
when exploring phase I terminates,
there exists a state $\realizedstateanchorj\in[\totalstate]$
such that $(\realizedstateanchori,\realizedstateanchorj)\in
\anchorset$.
For each pair of states $(\realizedstatei'',\realizedstatej'')$
such that $\costState(\realizedstatei'') \geq 0$,
consider a direct signaling scheme 
$\scheme^{(\realizedstatei'',\realizedstatej'')}$
with 
\begin{align*}
\scheme^{(\realizedstatei'',\realizedstatej'')}
(\realizedstatei'') &= 1
\qquad
\scheme^{(\realizedstatei'',\realizedstatej'')}
(\realizedstatej'')= 
\prior(\realizedstatej)
\left(
\indicator{\costState(\realizedstatej) \geq 0}
+
\min\left\{
\frac{-\costState(\realizedstateanchori)}{
\costState(\realizedstatej)},
1\right\}
\cdot 
\indicator{\costState(\realizedstatej) < 0}
\right)
\\
\scheme^{(\realizedstatei'',\realizedstatej'')}
(\realizedstatei) &= 0
\qquad
\text{for every state }\realizedstatei\not\in\{
\realizedstatei'',\realizedstatej''\}
\end{align*}
Namely, 
$\scheme^{(\realizedstatei'',\realizedstatej'')}$
is the direct \IC\ signaling scheme that maximizes
$\prior(\realizedstatej'')\scheme(\realizedstatej'')$
when action 1 is only allowed to be recommended in 
state $\realizedstatei''$ or $\realizedstatej''$.
By definition, among all pairs of states 
$(\realizedstatei'', \realizedstatej'')$,
the pair that maximizes $\prior(\realizedstatej'')\scheme(\realizedstatej'')$
must be $\realizedstatei'' = \realizedstateanchori$,
which shows the second part of the lemma statement.
\end{proof}

Due to \Cref{lem:existence of anchor pair},
when exploring phase I terminates,
$\SExpUtilityUnderbar \geq
\frac{\SExpUtility(\optoff)}{\totalstate^2}$.
This enables us to build the regret bound
of exploring phase I as follows.
\begin{lemma}
\label{lem:regret log T phase I}
In \LPRP,
the expected regret in exploring phase I
is at most $O(\totalstate^4)$.
\end{lemma}
\begin{proof}
Let $K_1 = -\log(\SExpUtility(\optoff))$,
and $K_2 = -\lceil\log\left(
\frac{\SExpUtility(\optoff)}{\totalstate^2}
\right)\rceil$.
By \Cref{lem:existence of anchor pair}
when exploring phase I terminates,
$\SExpUtilityUnderbar \geq 
\frac{\SExpUtility(\optoff)}{\totalstate^2}$,
and thus
there are at most $K_2$ iterations
in the while loop (line 2 in \LPRP).
For each iteration $k\in[K]$, 
$\CheckPersu(\schemeanchorpair)$ 
is called 
for every pair of states $(\realizedstatei,\realizedstatej)$
with $\SExpUtilityUnderbar\leq \prior(\realizedstatej)$.
By \Cref{lem:persuasive check regret}, the total expected regret is at most
\begin{align*}
    \totalstate^2\cdot 
    \sum_{k \in [K_2]}
    \frac{\SExpUtility(\optoff)}{
    \sum_{\realizedstatei\in[\totalstate]}
\prior(\realizedstatei)\schemeanchorpair(\realizedstatei)
    }
    \overset{(a)}{\leq}
    \totalstate^2\cdot 
    \sum_{k \in [K_2]}
    \frac{2^{-K_1}}{
    2^{-k}
    }
    =
    \totalstate^2\cdot 
    \sum_{k=K_1-K_2}^{K_1}
    2^{-k}
    =
    O(\totalstate^4)
\end{align*}
where the denominator 
in the right-hand side 
of inequality~(a) is due to the 
construction of $\schemeanchorpair$.
\end{proof}

\xhdr{The analysis of exploring phase II}

The goal of exploring phase II is to 
identify an interior point $\schemestartingpoint$
for the linear program solver $\lpsolver$
with membership oracle access.
To achieve this, 
\LPRP\ 
excludes degenerate states which 
contribute little to $\SExpUtility(\optoff)$,
and the remaining states form the subset 
$\meaningfulset$.
We characterize $\meaningfulset$
by the following lemma.

\begin{numberedlemma}{\ref{lem:meaningful set characterization}}
When exploring phase II terminates,
\begin{itemize}
\item for each state $\realizedstatei
    \in\meaningfulset$:~
    $\costStatei \geq -\totalstate\totaltime
    \cdot 
    \max_{\realizedstatej\in[\totalstate]}
    \costState(\realizedstatej)$;
    \item
     for each state $\realizedstatei\not\in
    \meaningfulset$:~
    $\costStatei < - 
    \frac{\totalstate\totaltime}{3} 
    \cdot
    \max_{\realizedstatej\in[\totalstate]}
    \costState(\realizedstatej)$.
    \end{itemize}
\end{numberedlemma}

\begin{proof}
Let $\realizedstateanchori =
\argmax_{\realizedstatei\in[\totalstate]}\costStatei$.
For each state $\realizedstatei\in\meaningfulset$,
suppose it is added into $\meaningfulset$ due to a pair of state $(\realizedstatei',\realizedstatej')\in 
\anchorset$. 
Suppose $\realizedstatei'\not=\realizedstatej'$
(A similar argument holds for
$\realizedstatei'=\realizedstatej'$).
In this case, we know that
the signaling scheme $\schemeanchorpairk$ 
corresponding to $(\realizedstatei',
\realizedstatej',\realizedstatei,\SExpUtilityUnderbar)$ 
is \IC,
i.e., 
\begin{align*}
    0 
    &\leq \sum_{\realizedstatej\in[\totalstate]}
    \costState(\realizedstatej)
    \schemeanchorpairk(\realizedstatej)
    \overset{(a)}{=}
    \costState(\realizedstatei')
    +
    \costState(\realizedstatej')
    \frac{\SExpUtilityUnderbar}
    {2\prior(\realizedstatej')}
    +
    \costState(\realizedstatei)
    \frac{3}{2\totalstate\totaltime}
\overset{(b)}{\leq} 
    \costState(\realizedstateanchori)
    +
    \frac{1}{2}\costState(\realizedstateanchori)
    +
    \costState(\realizedstatei)
    \frac{3}{2\totalstate\totaltime}
\end{align*}
which implies that
$\costState(\realizedstatei) \geq
-\totalstate\totaltime\costState(\realizedstateanchori)$.
Here equality~(a) holds due to the construction of $\schemeanchorpairk$,
and inequality~(b) holds since 
$\costState(\realizedstatei')\leq
\costState(\realizedstateanchori)$,
$\costState(\realizedstatej')\leq
\costState(\realizedstateanchori)$,
and $\SExpUtilityUnderbar \leq \prior(\realizedstatej')$.

By \Cref{lem:existence of anchor pair},
there exists a state $\realizedstateanchorj$
such that $(\realizedstateanchori,\realizedstateanchorj)
\in \anchorset$.
For each state $\realizedstatei\not\in
\meaningfulset$,
we know that 
the signaling scheme $\schemeanchorpairk$
corresponded to 
$(\realizedstateanchori,\realizedstateanchorj,
\realizedstatei,\SExpUtilityUnderbar)$
is not \IC,
i.e., 
\begin{align*}
    0 
    > \sum_{\realizedstatej\in[\totalstate]}
    \costState(\realizedstatej)
    \schemeanchorpairk(\realizedstatej)
    \overset{(a)}{=}
    \costState(\realizedstateanchori)
    +
    \costState(\realizedstateanchorj)
    \frac{\SExpUtilityUnderbar}
    {2\prior(\realizedstateanchorj)}
    +
    \costState(\realizedstatei)
    \frac{3}{2\totalstate\totaltime}
\overset{(b)}{\geq} 
    \costState(\realizedstateanchori)
    -
    \frac{1}{2}\costState(\realizedstateanchori)
    +
    \costState(\realizedstatei)
    \frac{3}{2\totalstate\totaltime}
\end{align*}
which implies that 
$\costState(\realizedstatei) <
-
\frac{\totalstate\totaltime}{3}
\costState(\realizedstateanchori)$.
Here equality~(a) holds due to the construction of $\schemeanchorpairk$,
and inequality~(b) holds since 
$\costState(\realizedstateanchori)
+
\costState(\realizedstateanchorj)
\frac{\SExpUtilityUnderbar}{\prior(\realizedstateanchorj)}
\geq 0$.
\end{proof}

The first part of \Cref{lem:meaningful set characterization}
guarantees that 
there exists 
a pair of state 
$(\realizedstateanchori,\realizedstateanchorj)\in
\meaningfulset$
such that the corresponding $\schemestartingpoint$
is an interior point of program~\ref{eq:optimum in hindsight within meaningfulset}
(see \Cref{lem:interior point}).
The second part of \Cref{lem:meaningful set characterization}
guarantees that the optimal signaling scheme 
$\meaningfulsetopt$
in program~\ref{eq:optimum in hindsight within meaningfulset}
is close to 
the optimal signaling scheme $\optoff$
in program~\ref{eq:optimal in hindsight}
(see \Cref{lem:meaningful set approximate optimal}).

\begin{numberedlemma}{\ref{lem:meaningful set approximate optimal}}
Let $\meaningfulsetopt$
be the optimal solution in 
program~\ref{eq:optimum in hindsight within meaningfulset},
i.e., $\meaningfulsetopt
= \argmax$
\ref{eq:optimum in hindsight within meaningfulset}.
Then
$\SExpUtility(\meaningfulsetopt) \geq 
\SExpUtility(\optoff) - O(\frac{1}{\totaltime})$.
\end{numberedlemma}
\begin{proof}
By \Cref{lem:optimal in hindsight characterization},
in the optimal signaling scheme $\optoff$,
there exists a threshold state
$\realizedstatei\primed\in[\totalstate]$.
For each state $\realizedstatei$ 
above $\realizedstatei\primed$,
$\optoff(\realizedstatei) = 1$;
and for each state $\realizedstatei$ below
$\realizedstatei\primed$,
$\optoff(\realizedstatei) = 0$.

Let 
$\realizedstateanchori =
\argmax_{\realizedstatei\in[\totalstate]}\costStatei$.
We first show that for each state $\realizedstatei$
above $\realizedstatei\primed$, 
$\realizedstatei\in\meaningfulsetopt$.
To see this, note that $\optoff$ is \IC, i.e., 
\begin{align*}
    0 = 
    \sum_{\realizedstatej\in[\totalstate]}
    \costState(\realizedstatej)
    \optoff(\realizedstatej)
    =
    \costState(\realizedstatei)
    +
    \sum_{\realizedstatej\in[\totalstate]
    \backslash\{\realizedstatei\}}
    \costState(\realizedstatej)
    \optoff(\realizedstatej)
    \leq 
     \costState(\realizedstatei)
     +
     (\totalstate-1)\costState(\realizedstateanchori)
\end{align*}
which implies that $\costStatei\geq 
- (\totalstate-1)\costState(\realizedstateanchori)$.
By \Cref{lem:meaningful set characterization},
we conclude that $\realizedstatei\in\meaningfulset$.\footnote{Here we assume $\totaltime\geq 3$.}
Hence, we can now upperbound
$\SExpUtility(\optoff) - \SExpUtility(\meaningfulsetopt)$
as follows,
\begin{align*}
    \SExpUtility(\optoff) - \SExpUtility(\meaningfulsetopt)
    &
    \overset{(a)}{\leq}
    \prior(\realizedstatei\primed)
    \optoff(\realizedstatei\primed)
    \cdot \indicator{\realizedstatei\primed\not\in
    \meaningfulset}
    \leq 
    \optoff(\realizedstatei\primed)
    \cdot \indicator{\realizedstatei\primed\not\in
    \meaningfulset}
    \overset{(b)}{\leq} 
    -\frac{(\totalstate-1)\costState(\realizedstateanchori)}{
    \costState(\realizedstatei\primed)}
    \cdot \indicator{\realizedstatei\primed\not\in
    \meaningfulset}
    \overset{(c)}{\leq} 
    O\left(\frac{1}{\totaltime}\right)
\end{align*}
where 
inequality~(a) holds since $\realizedstatei\in\meaningfulset$
for every $\realizedstatei$ such that 
$\optoff(\realizedstatei) = 1$;
inequality~(c) holds 
due to \Cref{lem:meaningful set characterization};
and inequality~(b) holds due
to the incentive compatibility of $\optoff$, i.e.,
\begin{align*}
    0 = 
    \sum_{\realizedstatej\in[\totalstate]}
    \costState(\realizedstatej)
    \optoff(\realizedstatej)
    =
    \costState(\realizedstatei\primed)
    \optoff(\realizedstatei\primed)
    +
    \sum_{\realizedstatej\in[\totalstate]
    \backslash\{\realizedstatei\primed\}}
    \costState(\realizedstatej)
    \optoff(\realizedstatej)
    \leq 
     \costState(\realizedstatei\primed)
     \optoff(\realizedstatei\primed)
     +
     (\totalstate-1)\costState(\realizedstateanchori)
\end{align*}
and $\costState(\realizedstatei\primed) < 0$
if $\realizedstatei\primed\not\in\meaningfulset$.
\end{proof}

Finally, we present the regret guarantee
in exploring phase II.

\begin{lemma}
\label{lem:regret log T phase II}
In \LPRP,
the expected regret in exploring phase II
is at most $O(\totalstate^5)$.
\end{lemma}
\begin{proof}
In exploring phase II,
$\CheckPersu(\schemeanchorpairk)$ 
is called 
for every pair of states $(\realizedstateanchori,\realizedstateanchorj)
\in\anchorset$ and $\realizedstatei\in [\totalstate]
\backslash\{\realizedstateanchori,\realizedstateanchorj\}$.
By \Cref{lem:persuasive check regret}, the total expected regret is at most 
\begin{align*}
\sum_{(\realizedstateanchori,\realizedstateanchorj)
\in\anchorset}
\sum_{\realizedstatei\in [\totalstate]
\backslash\{\realizedstateanchori,\realizedstateanchorj\}}
    \frac{\SExpUtility(\optoff)}{
    \sum_{\realizedstatej\in[\totalstate]}
\prior(\realizedstatej)\schemeanchorpairk(\realizedstatej)
    }
    \overset{(a)}{\leq}
    \totalstate^3\cdot 
    \frac{\SExpUtility(\optoff)}{
    \frac{1}{2}\SExpUtilityUnderbar
    }
    \overset{(b)}{\leq}
    O(\totalstate^5)
\end{align*}
where the denominator 
in the right-hand side 
of inequality~(a) is due to the 
construction of $\schemeanchorpairk$,
and inequality~(b) is due to 
\Cref{lem:existence of anchor pair}.
\end{proof}

\xhdr{The analysis of exploring phase III}
Let $\convexset$(\ref{eq:optimum in hindsight within meaningfulset}) be the convex set in 
program~\ref{eq:optimum in hindsight within meaningfulset}.
Here we show that we can find an interior 
point $\schemestartingpoint$
for some pair of states
$(\realizedstateanchori,\realizedstateanchorj)\in\anchorset$.

\begin{lemma}[restatement of \Cref{lem:interior point main}]
\label{lem:interior point}
There exists 
a pair of state 
$(\realizedstateanchori,\realizedstateanchorj)\in
\anchorset$
such that 
$\schemestartingpoint$ is an interior point 
of program~\ref{eq:optimum in hindsight within meaningfulset}.
In particular, let $r = \frac{1}{16\totalstate^2\totaltime}$,
then $\TwoNormBall(\schemestartingpoint,r)\subseteq 
\convexset($\ref{eq:optimum in hindsight within meaningfulset}$)$.
\end{lemma}
\begin{proof}
Let 
$\realizedstateanchori =
\argmax_{\realizedstatei\in[\totalstate]}\costStatei$.
By \Cref{lem:existence of anchor pair},
there exists a state $\realizedstateanchorj$
such that $(\realizedstateanchori,\realizedstateanchorj)
\in \anchorset$.
It is sufficient to show that 
the signaling scheme $\schemestartingpoint$
corresponding to $(\realizedstateanchori,\realizedstateanchorj,\meaningfulset)$ defined here
satisfies the requirement.
In particular, 
Fix an arbitrary 
$\scheme\in\TwoNormBall(\schemestartingpoint,r)$.
Below, we show that every constraint
in program~\ref{eq:optimum in hindsight within meaningfulset}
is satisfied.

We first examine
the feasibility constraint, i.e.,
$\scheme(\realizedstatei)\in[0, 1]$
for every $\realizedstatei\in\meaningfulset$.
For every state 
$\realizedstatei\not=\realizedstateanchorj$,
the feasibility constraint 
is satisfied obviously.
For state $\realizedstateanchorj$, 
note that $\SExpUtilityUnderbar\geq
\frac{1}{\totalstate^2\totaltime}$
and thus $\schemestartingpoint(\realizedstateanchorj) \geq 
\frac{1}{\totalstate^2\totaltime}$,
which guarantees the feasibility constraint.

We next examine the constraint that 
$\sum_{\realizedstatei}
         \prior(\realizedstatei)
         \scheme(\realizedstatei)
         \geq \frac{1}{16}\,\SExpUtilityUnderbar$.
To see this, note that 
\begin{align*}
    \sum_{\realizedstatei\in\meaningfulset}
         \prior(\realizedstatei)
         \scheme(\realizedstatei)
         \geq 
         \prior(\realizedstateanchorj)
         \scheme(\realizedstateanchorj)
         \geq 
         \prior(\realizedstateanchorj)
         \left(
         \frac{\SExpUtilityUnderbar}{
         8\prior(\realizedstateanchorj)}
         -
         r
         \right)
         \geq 
          \prior(\realizedstateanchorj)
         \frac{\SExpUtilityUnderbar}{
         16\prior(\realizedstateanchorj)}
         =\frac{1}{16}\,\SExpUtilityUnderbar
\end{align*}

Finally, we examine the \IC\ constraint.
\begin{align*}
    \sum_{\realizedstatei\in[\totalstate]}
    \costState(\realizedstatei)
    \scheme(\realizedstatei)
    &\overset{}{\geq}
    \frac{1}{2}\costState(\realizedstateanchori)
    +
    \costState(\realizedstateanchorj)
    \scheme(\realizedstateanchorj)
+
    \sum_{\realizedstatei\in
    \meaningfulset\backslash\{\realizedstateanchori,
    \realizedstateanchorj\}}
    \costState(\realizedstatei)
    \scheme(\realizedstatei)
\\
    &\overset{(a)}{\geq}
    \frac{1}{4}
    \costState(\realizedstateanchori)
    -
    \frac{\prior(\realizedstateanchorj)}{\SExpUtilityUnderbar}
    \costState(\realizedstateanchori)
   \scheme(\realizedstateanchorj)
    +
    \sum_{\realizedstatei\in
    \meaningfulset\backslash\{\realizedstateanchori,
    \realizedstateanchorj\}}
    \left(
    \costState(\realizedstateanchori)
    \frac{1}{4\totalstate}
    -
    \totalstate\totaltime\cdot \costState(\realizedstateanchori)
\scheme(\realizedstatei)
    \right)
    \\
    &\overset{}{\geq}
    \frac{1}{4}
    \costState(\realizedstateanchori)
    -
    \costState(\realizedstateanchori)
   \left(
   \frac{1}
    {8}
         +
         r
             \frac{\prior(\realizedstateanchorj)}{\SExpUtilityUnderbar}
   \right)
    +
    \sum_{\realizedstatei\in
    \meaningfulset\backslash\{\realizedstateanchori,
    \realizedstateanchorj\}}
    \left(
    \costState(\realizedstateanchori)
    \frac{1}{4\totalstate}
    -
    \totalstate\totaltime\cdot \costState(\realizedstateanchori)
    \left(\frac{1}{8\totalstate^2\totaltime}
    +r
    \right)
\right)
    \\
    &\overset{(b)}{\geq} 0
\end{align*}
where 
inequality~(a) holds 
since 
$\costState(\realizedstateanchori)
+
\costState(\realizedstateanchorj)
\frac{\SExpUtilityUnderbar}{\prior(\realizedstateanchorj)}
\geq 0$,
and 
$\costState(\realizedstatei)\geq 
-\totalstate\totaltime\cdot \costState(\realizedstateanchori)$
by \Cref{lem:meaningful set characterization};
and inequality~(b) holds since 
$ r\frac{\prior(\realizedstateanchorj)}{\SExpUtilityUnderbar}
\leq \frac{1}{16}$.
\end{proof}

Next, we present the regret guarantee
in exploring phase III.

\begin{lemma}
\label{lem:regret log T phase III}
In \LPRP,
the expected regret in exploring phase III
is at most $
    O\left(
    \totalstate^6
     \log^{O(1)}
    \left(\totalstate\totaltime
    \right)
    \right)$.
\end{lemma}
\begin{proof}
In exploring phase II,
$\lpsolver$ is executed for each 
$(\realizedstateanchori,\realizedstateanchorj)\in
\anchorset$
where $|\anchorset|\leq \totalstate^2$.
Within each execution of $\lpsolver$,
$\CheckPersu(\schemeanchorpairk)$ 
is called 
as the membership oracle.
Note that we 
run Procedure~\ref{alg:persuasive check}
only
if constraint $\sum_{\realizedstatei}
         \prior(\realizedstatei)
         \scheme(\realizedstatei)
         \geq \frac{1}{16}\,\SExpUtilityUnderbar$
         is satisfied.
Thus, by \Cref{lem:persuasive check regret}
and \Cref{lem:existence of anchor pair},
the expected regret in exploring phase III 
is at most $O(\totalstate^2)$
times 
the total number of queries
to the membership oracle in all executions.
Invoking \Cref{thm:membership oracle query bound}
finishes the proof.
\end{proof}

\xhdr{The analysis of exploiting phase}

Here we present the regret guarantee
in exploiting phase.

\begin{lemma}
\label{lem:regret log T phase IIII}
In \LPRP,
the expected regret in exploiting phase
is at most $
    O\left(
   1
    \right)$.
\end{lemma}
\begin{proof}
There are three different cases.
If exploring phase III terminates due to 
$\SExpUtilityUnderbar < \frac{1}{\totalstate^2\totaltime}$,
by \Cref{lem:existence of anchor pair},
we know that $\SExpUtility(\optoff) <\frac{1}{\totaltime}$.
Hence, using any signaling scheme (including 
$\scheme\primed$) induces $O(1)$ regret.

If exploring phase III terminates with
$\anchorset\not=\emptyset$,
then linear program solver $\lpsolver$
is executed.
Recall that $\lpsolver$
is a randomized algorithm with success probability
$1- \frac{1}{\totaltime}$.
If it fails, the regret is at most $\totaltime$,
which happens with probability $\frac{1}{\totaltime}$.
If it succeeds, 
by \Cref{lem:interior point},
$\SExpUtility(\scheme\primed)\geq 
\SExpUtility(\meaningfulsetopt) - O(\frac{1}{\totaltime})
\geq 
\SExpUtility(\optoff) - O(\frac{1}{\totaltime})$.
\end{proof}

\smallskip

Combining the regret analysis in all phases,
we can prove \Cref{thm:log T}.
\begin{proof}[Proof of \Cref{thm:log T}]
Invoking \Cref{lem:regret log T phase I},
\Cref{lem:regret log T phase II},
\Cref{lem:regret log T phase III},
and \Cref{lem:regret log T phase IIII}
finishes the proof.
\end{proof} 

\section{Omitted Proofs in Section~\ref{sec:loglog T lower bound}}
\label{apx:proofs lower bound}

In this section, we present the omitted proofs
in \Cref{sec:loglog T lower bound}.

\multilabelpolicytobinarylabelpolicy*
The proof of \Cref{lem:multi label policy to binary label policy} relies on the following lemma,
which
guarantees that for any distribution $\posterior$
of posterior belief over binary state, there exists 
a distribution of signaling schemes with binary signal space
that implements $\posterior$.

\begin{lemma}
\label{lem:multi label signal to binary label signal}
Let $\scheme:[2] \rightarrow \stochastic(\signalSpace)$ 
be a signaling scheme 
that maps binary state
into probability distributions a over finite signal space $\signalSpace$,
and 
$\posterior:\signalSpace\rightarrow
\stochastic([2])$ be the 
distribution of posterior belief
induced by $\scheme$.
There exists a positive integer $K$,
and a finite set 
$\{\scheme\ked\}_{k\in[K]}$
where each 
$\scheme\ked:
[2]\rightarrow
\stochastic(\{0, 1\})$
is a signaling scheme
with binary signal space.
Let $\posterior\ked$ be 
the distribution of posterior belief
induced by $\scheme\ked$ 
for each $k\in[K]$.
Then,
there exists a distribution $F$
over $[K]$
such that 
for every possible posterior belief realization
$x\in\supp(\posterior)$,
$\prob{\posterior=x} = 
\expect[k\sim F]{
\prob{\posterior\ked = x}}$.
\end{lemma}

The proof of \Cref{lem:multi label signal to binary label signal}
relies on 
\Cref{lem:posterior Bayesian plausibility}
and 
\Cref{lem:multi label signal to binary label signal random variable} as follows.
\ifarxiv
\else
The proof of \Cref{lem:multi label signal to binary label signal random variable}
is deferred to \Cref{apx:proofs in loglog T}.
\fi

\begin{lemma}[\citealp{KG-11}]
\label{lem:posterior Bayesian plausibility}
Let $\prior\in\stochastic([2])$ be a prior distribution 
over binary state space $[2]$.
A distribution of posterior belief
$\posterior\in
\stochastic(\stochastic([2]))$
is implementable (i.e., 
can be induced by some signaling scheme)
if and only if
$\prob[x\sim \posterior,\state\sim x]
{
\state = 1} = 
\prior(1)$.
\end{lemma}

\begin{restatable}{lemma}{multiLabelToBinaryLabelRV}
\label{lem:multi label signal to binary label signal random variable}
Let $X$ be a random variable 
with discrete support 
$\supp(X)$.
There exists a 
positive integer $K$,
a finite set of $K$
random variables $\{X_k\}_{k\in[K]}$,
and 
convex combination coefficients
$\convexcombinbf\in [0, 1]^K$ 
with $\sum_{k\in[K]} \convexcombin_k = 1$
such that:
\begin{enumerate}
    \item \underline{Bayesian-plausibility}: for each $k\in [K]$, $\expect{X_k} = \expect{X}$;
    \item \underline{Binary-support}:
    for each $k\in[K]$, the size of $X_k$'s support is at most 2,
    i.e., $|\supp(X_k)| \leq 2$
    \item \underline{Consistency}:
    for each $x\in \supp(X)$, 
$\prob{X=x} = \sum_{k\in[K]}
\convexcombin_k\cdot 
\prob{X_k = x}$
\end{enumerate}
\end{restatable}
\begin{proof}
For notation simplicity, 
we first introduce the following notations.
We denote $\expect{X}$ as $\mean$,
and $\supp(X)$ as $\support$.
We partition $\support$
into $\supportlareq
\triangleq 
\{\xlareq^{(1)}, \dots, \xlareq^{(n_1)}\}$ and 
$\supportsma 
\triangleq
\{\xsma^{(1)},\dots,\xsma^{(n_2)}\}$
where 
$\forall i\in[n_1]$,
$\xlareq^{(i)} \geq \mean$;
and 
$\forall j\in[n_2]$,
$\xsma^{(j)} < \mean$.
We first show the statement holds
if $\mean \not\in \support$.
A similar argument holds for the other case 
where $\mean \in \support$.

\smallskip
Now suppose 
$\mean \not\in \support$.
Let $\problari$ denote 
$\prob{X = \xlari}$
and $\probsmai$ denote 
$\prob{X = \xsmai}$.
Consider the following linear system
with variable $\{\convexcombini\}
_{i\in [n_1],j\in[n_2]}$:
\begin{align}
\label{eq:binary label linear system original}
\left\{
\begin{array}{ll}
\sum_{j\in[n_2]}
\frac{\mean - \xsmai}{\xlari-\xsmai}
\convexcombini = \problari
&  \forall i\in[n_1]\\
\sum_{i\in[n_1]}
\frac{\xlari - \mean}{\xlari-\xsmai}
\convexcombini = \probsmai
&  \forall j\in[n_2]\\
    \convexcombini\geq 0 &
    \forall i\in[n_1],j\in[n_2]
\end{array}
    \right.
\end{align}
Below we first show how to construct 
random variable $\{X_k\}_{k\in [K]}$
required in the lemma statement
with
any feasible 
solution in
linear system~\eqref{eq:binary label linear system original} as the convex combination coefficients.
Then we show the existence
of the feasible solution 
in
linear system~\eqref{eq:binary label linear system original}.

Fix a feasible solution 
$\{\convexcombini\}_{i\in [n_1],j\in[n_2]}$ 
in
linear system~\eqref{eq:binary label linear system original}.
Let $K = n_1\cdot n_2$.
Consider the set of random variables 
$\{X_{ij}\}_{i\in[n_1],j\in[n_2]}$
as follows,
\begin{align*}
    X_{ij} = \left\{
    \begin{array}{ll}
    \xlari     &  
    \text{with probability~} 
    \frac{\mean - \xsmai}{\xlari - \xsmai}\\
    \xsmai     & 
    \text{otherwise}
    \end{array}
    \right.
\end{align*}
Note that $\{\convexcombini\}$
is valid convex combination coefficient for 
$\{X_{ij}\}_{i\in[n_1],j\in[n_2]}$.
In particular,
\begin{align*}
    \sum_{i\in[n_1]}\sum_{j\in[n_2]}
    \convexcombin_{ij}
    =
    \sum_{i\in[n_1]}\sum_{j\in[n_2]}
    \convexcombin_{ij}
    \left(
    \frac{\mean - \xsmai}{\xlari-\xsmai}
    +
\frac{\xlari - \mean}{\xlari-\xsmai}
    \right)
    =
    \sum_{i\in[n_1]}\problari +
    \sum_{j\in[n_2]}\probsmai
    = 1
\end{align*}
To see why
random variables
$\{X_{ij}\}_{i\in[n_1],j\in[n_2]}$ with 
convex combination coefficients
$\{\convexcombini\}_{i\in [n_1],j\in[n_2]}$
satisfy the statement requirement,
note that
``Bayesian-plausibility''
property
and ``Binary-support'' property
are satisfied
by construction.
To verify ``Consistency'' property,
consider each $\xlari\in \supportlar$,
\begin{align*}
    \sum_{i\in[n_1]}\sum_{j\in[n_2]}
    \convexcombin_{ij}\cdot 
    \prob{X_{ij} = \xlari}
    =&
    \sum_{j\in [n_2]}
    \convexcombin_{ij}\cdot 
    \prob{X_{ij} = \xlari}
    \\
    =&
    \sum_{j\in [n_2]}
    \convexcombin_{ij}\cdot 
    \frac{\mean - \xsmai}{\xlari-\xsmai}
= \problari = \prob{X = \xlari}
\end{align*}
Similarly, for each 
$\xsmai\in \supportsma$,
\begin{align*}
    \sum_{i\in[n_1]}\sum_{j\in[n_2]}
    \convexcombin_{ij}\cdot 
    \prob{X_{ij} = \xsmai}
    =&
    \sum_{i\in [n_1]}
    \convexcombin_{ij}\cdot 
    \prob{X_{ij} = \xsmai}
    \\
    =&
    \sum_{i\in [n_1]}
    \convexcombin_{ij}\cdot 
    \frac{\xlari - \mean}{\xlari-\xsmai}
= \probsmai = \prob{X = \xsmai}
\end{align*}
Hence, we conclude that 
random variables $\{X_{ij}\}_{i\in[n_1],j\in[n_2]}$ with 
convex combination coefficients
$\{\convexcombini\}_{i\in [n_1],j\in[n_2]}$
satisfy the statement requirement.

\smallskip
Next, we show the existence of a feasible solution 
$\{\convexcombini\}_{i\in [n_1],j\in[n_2]}$
in linear system \eqref{eq:binary label linear system original}.
Let $\convexcombinhati = 
\frac{\convexcombini}{\xlari-\xsmai}$.
It is equivalent to showing that
\begin{align}
\label{eq:binary label linear system modified}
\left\{
\begin{array}{ll}
\sum_{j\in[n_2]}
(\mean - \xsmai)
\convexcombinhati = \problari
&  \forall i\in[n_1]\\
\sum_{i\in[n_1]}
(\xlari - \mean)
\convexcombinhati = \probsmai
&  \forall j\in[n_2]\\
    \convexcombini\geq 0 &
    \forall i\in[n_1],j\in[n_2]
\end{array}
    \right.
\end{align}
has a feasible solution.
We show this by an induction argument.

\smallskip
\noindent
\textbf{Induction Hypothesis.}
Fix any positive integer 
$n_1, n_2 \in \natures_{\geq 1}$,
and arbitrary non-negative numbers
\footnote{Here we do not 
assume that
$\sum_{i\in[n_1]}\problari
+
\sum_{j\in[n_2]}\probsmai = 1$.}
$\{\xlari,\problari\}_{i\in[n_1]}$,
$\{\xsmai,\probsmai\}_{i\in[n_2]}$,
and
$\mean$.
If 
$\xlari > \mean$
for all $i\in[n_1]$,
$\xsmai < \mean$
for all $j\in[n_2]$,
and
$
\sum_{i\in[n_1]}\xlari \problari
+
\sum_{j\in[n_2]}\xsmai \probsmai
=
\left(
\sum_{i\in[n_1]}\problari
+
\sum_{j\in[n_2]}\probsmai
\right)\,
\mean$;
then linear system~\eqref{eq:binary label linear system modified}
has a feasible solution.

\noindent\textbf{Base Case 
($n_1 = 1$ or $n_2 = 1$).}
Here we show the induction 
hypothesis holds for $n_1 = 1$.
A similar argument
holds for $n_2 = 1$.
Consider a 
solution $\{\convexcombinhat_{1j}\}_{j\in[n_2]}$
constructed as follows,
\begin{align*}
    \convexcombinhat_{1j} = 
    \frac{\probsmai}{\xlar^{(1)} - \mean}
\end{align*}
It is obvious that 
$\{\convexcombinhati\}$
is non-negative and
the equality for every $j\in[n_2]$
is satisfied.
Now consider the equality for $i=1$. Note that
\begin{align*}
  \sum_{j\in[n_2]}
(\mean - \xsmai)
\convexcombinhat_{1j}
&=
  \sum_{j\in[n_2]}
\frac{(\mean - \xsmai)\probsmai}
{\xlar^{(1)} - \mean}
=
\frac{\sum_{j\in[n_2]}
(\mean - \xsmai)\probsmai}
{\xlar^{(1)} - \mean}
\overset{(a)}{=}
\frac{(\xlar^{(1)} - \mean)\problar^{(1)}
}
{\xlar^{(1)} - \mean}
=
\problar^{(1)}
\end{align*}
where equality~(a) uses the assumption that
$
\xlar^{(1)} \problar^{(1)}
+
\sum_{j\in[n_2]}\xsmai \probsmai
=
\left(
\problar^{(1)}
+
\sum_{j\in[n_2]}\probsmai
\right)\,
\mean$.

\noindent\textbf{Inductive Step 
($n_1 \geq 2$ and $n_2 \geq 2$).}
Here we show the induction 
hypothesis holds for $n_1 \geq 2$
and $n_2\geq 2$.
In addition, we assume that
$\xlarn \problarn
+
\xsman \probsman \geq
(\problarn+\probsman)\mean$.
A similar argument holds for
$\xlarn \problarn
+
\xsman \probsman < 
(\problarn+\probsman)\mean$.
Below we show that there exists 
a feasible solution where
we fix
\begin{align*}
\begin{split}
\forall i\in[n_1-1]:
    ~
    \convexcombinhat_{i\,n_2} = 0
    \qquad\mbox{and}
    \qquad \convexcombinhat_{n_1n_2} = 
    \frac{\probsman}{\xlarn - \mean}
    \end{split}
\end{align*}
To see this,
observe that
the equality for $j = n_2$ is satisfied.
Next, we invoke the induction 
hypothesis 
on instance $(n_1, n_2 - 1,
\xlar^{(1)},
\problar^{(1)},
\dots,
\xlar^{(n_1-1)},
\problar^{(n_1-1)},
\xlarn,
\problartilden,
\xsma^{(1)},
\probsma^{(1)},
\dots,
\xsma^{(n_2-1)},
\probsma^{(n_2-1)},
\mean)$
where $\problartilden = 
\problarn - 
\convexcombinhat_{n_1n_2}(\mean - \xsman)$.
It is sufficient to show that
this instance
satisfies the assumption in the induction 
hypothesis.
In particular, we can verify that
\begin{align*}
    \problartilden = 
\problarn - 
\convexcombinhat_{n_1n_2}(\mean - \xsman)
=
\problarn - 
\frac{\probsman(\mean - \xsman)}
{\xlarn - \mean}
\geq 0
\end{align*}
since we assume $\xlarn \problarn
+
\xsman \probsman \geq
(\problarn+\probsman)\mean$;
and
\begin{align*}
    &
\sum_{i\in[n_1 - 1]}\xlari\cdot \problari
    +
    \xlarn \cdot \problartilden
+
\sum_{j\in[n_2 - 1]}\xsmai\cdot \probsmai
\\
=&
    \sum_{i\in[n_1]}\xlari\cdot \problari
+
\sum_{j\in[n_2]}\xsmai\cdot \probsmai
-
\left(
\xlarn(\problarn - \problartilden)
+
\xsman
\probsman
\right)
\\
=&
 \sum_{i\in[n_1]}\xlari\cdot \problari
+
\sum_{j\in[n_2]}\xsmai\cdot \probsmai
-
\left(
\xlarn
\frac{\probsman(\mean - \xsman)}
{\xlarn - \mean}
+
\xsman
\probsman
\right)
\\
\overset{(a)}{=}
&
\left(
\sum_{i\in[n_1]}\problari
+
\sum_{j\in[n_2]}\probsmai
\right)
\mean
-
\left(
\xlarn
\frac{\probsman(\mean - \xsman)}
{\xlarn - \mean}
+
\xsman
\probsman
\right)
&
\\
\overset{(b)}{=}
&
\left(
\sum_{i\in[n_1]}\problari
+
\sum_{j\in[n_2]}\probsmai
\right)
\mean
-
\left(
\frac{\probsman(\mean - \xsman)}
{\xlarn - \mean}
+
\probsman
\right)\,\mean
\\
=
&
\left(
\sum_{i\in[n_1 - 1]}\problari
+
\problartilden
+
\sum_{j\in[n_2 - 1]}\probsmai
\right)
\mean
\end{align*}
where equality~(a) uses the assumption that
$$\sum_{i\in[n_1]}\xlari \problari
+
\sum_{j\in[n_2]}\xsmai \probsmai
=
\left(
\sum_{i\in[n_1]}\problari
+
\sum_{j\in[n_2]}\probsmai
\right)\,
\mean$$
and equality~(b)
is by algebra.
\end{proof}

Now we are ready to show \Cref{lem:multi label signal to binary label signal} and then \Cref{lem:multi label policy to binary label policy}.

\begin{proof}[Proof of \Cref{lem:multi label signal to binary label signal}]
Let $\prior\in\stochastic([2])$ 
be the prior distribution
over binary state space $[2]$,
and $\state$ be the state drawn from $\prior$.
Fix an arbitrary signaling scheme $\scheme$
and let $\posterior$ 
be the distribution of posterior belief induced 
by $\scheme$.
Let $\signal$ be the signal issued by 
signaling scheme $\scheme$,
and
set random variable 
$X = \prob{\state = 1\condition \signal}$.
By \Cref{lem:posterior Bayesian plausibility},
$\expect[\signal]{X} = \prior(1)$.

\Cref{lem:multi label signal to binary label signal random variable}
ensures that there exists a positive integer $K$,
a finite set of $K$ random variables $\{X_k\}$,
and convex combination coefficients $\convexcombinbf$ that satisfy 
``Bayesian-plausibility'' property,
``binary-support'' property,
and ``consistency'' property.
Invoking \Cref{lem:posterior Bayesian plausibility},
we know that each random variable $X_k$
can be thought of as a distribution of 
posterior belief $\posterior\ked$
which can be induced by some signaling scheme
$\scheme\ked$ due to the ``Bayesian-plausibility'' property.
The ``binary-support'' property ensures 
that $\scheme\ked$ has binary signal space.
Let $F$ be the distribution over $[K]$
such that $\prob[k\sim F]{k=\ell} = f_{\ell}$.
The ``consistency'' property
guarantees that for every possible posterior belief realization
$x\in\supp(\posterior)$,
$\prob{\posterior=x} = 
\expect[k\sim F]{
\prob{\posterior\ked = x}}$.
\end{proof}

\begin{proof}[Proof of \Cref{lem:multi label policy to binary label policy}]
Fix an arbitrary problem instance $\instance$
with binary state,
and an arbitrary online policy $\ALG$.
Below we construct a randomized online policy
$\ALG\primed$
that only uses signaling scheme with binary signal.
Then, through a coupling argument, 
we show that \Receiver s' actions 
under $\ALG$ and 
\Receiver s' actions under $\ALG\primed$
are the same for each sample path, which finishes the 
proof.

We can
consider 
online policy $\ALG:
[\totaltime]
\times \historiesaction
\times \histories
 \rightarrow
\schemes$ as a 
mapping from 
the round index $t$, 
\Receiver s' action history 
$\historyaction :=(\action_1,\dots, \action_{t-1})$
in the previous $t - 1$ rounds,
and 
the other \Receiver-irrelevant
history\footnote{For example, \Receiver-irrelevant
history
may encode 
the realized states in previous rounds
and 
the random seed for the randomness of 
selecting signaling schemes in $\ALG$.} $\history$
to 
the signaling scheme $\scheme$
used by $\ALG$ in round $t$.
Here $\historiesaction$
is the set of all possible \Receiver s' action history,
$\histories$
is the set of all possible \Receiver-irrelevant 
history,
and $\schemes$ is the set of all signaling schemes.

Now we describe the construction of $\ALG\primed$
which uses mapping $\ALG$ as a blackbox.\footnote{We 
use notation $\dagger$ to denote terms under 
constructed policy $\ALG\primed$.}
We also need to define a coupling 
between the sample path under $\ALG$
and the sample path under $\ALG\primed$.
We construct $\ALG\primed$
and its coupling with $\ALG$ inductively 
(i.e., round by round).
Under the construction of $\ALG\primed$,
together with its coupling, each \Receiver\ $t$ 
forms
the same posterior belief and takes
the same action on both the sample path under $\ALG$
and the sample path under $\ALG\primed$.

We start with round 1. Let $\history$
be the \Receiver-irrelevant history under $\ALG$.
Since $\history$ is \Receiver-irrelevant, 
it can be simulated in $\ALG\primed$.
$\ALG\primed$
first determines the signaling scheme 
$\scheme_1\triangleq \ALG(1,\emptyset,\history)$
that $\ALG$ uses in round 1 given 
\Receiver s' action history $\emptyset$ (which is empty in
the beginning of round 1),
and \Receiver-irrelevant history $\history$.
By \Cref{lem:multi label signal to binary label signal},
there exists a distribution $F_1\primed$ over 
signaling schemes with binary signal space
such that the distribution of posterior belief 
is the same as $\scheme_1$.
Then $\ALG\primed$ randomly draws a signaling scheme
$\scheme_1\primed$ from distribution $F_1\primed$,
and commits to it in round $1$.
By coupling state $\theta_1$ with 
state $\theta_1\primed$, 
and properly coupling signal $\signal_1$
from $\scheme_1$
with signal $\signal_1\primed \sim \scheme_1\primed 
~(\sim F_1\primed)$,
we can ensure that the realized posterior belief 
$\posterior_1$ under $\ALG$
is the same as the realized posterior belief
$\posterior_1\primed$
under $\ALG\primed$,
and thus \Receiver\ $1$'s action $\action_1$ under $\ALG$
is the same as his $\action_1\primed$ under $\ALG\primed$.

Suppose we have constructed $\ALG\primed$
together with its coupling for the first $t-1$ rounds.
In round $t$, 
let $\historyaction$
be the \Receiver s' action history 
under $\ALG$,
and $\historyaction\primed$
be the \Receiver s' action history 
under $\ALG\primed$.
Because of the coupling in 
the first $t-1$ rounds,
we have $\historyaction = \historyaction\primed$.
Let $\history$
be the \Receiver-irrelevant history under $\ALG$.
Again, since $\history$ is \Receiver-irrelevant, 
$\ALG\primed$
can
compute the distribution of $\history$
that is consistent with 
the \Receiver s' action history
$\historyaction\primed=\historyaction$,
and sample $\history\primed$ from this distribution.
Here we couple  
the \Receiver-irrelevant history $\history$ under $\ALG$
with the simulated $\history\primed$ in $\ALG\primed$, 
so that $\history\primed = \history$.
$\ALG\primed$
first determines the signaling scheme 
$\schemei\triangleq \ALG(t,\historyaction\primed,\history\primed)$
that $\ALG$ uses in 
this round $t$ given 
\Receiver s' action history $\historyaction\primed$,
and \Receiver-irrelevant history $\history\primed$.
The remaining construction of distribution $F_t\primed$ 
over 
signaling schemes with binary signal space,
realized signaling scheme $\schemei\primed$
and their coupling (so that $\actioni\primed = \actioni$) 
are the same as what we do in round $1$.
We omit them to avoid redundancy.

Given the construction of $\ALG\primed$
and its coupling described above, 
we conclude that \Receiver s' actions 
are the same under $\ALG\primed$ and $\ALG$,
which finishes the proof.
\end{proof}

\dynamicPricingReduction*
\begin{proof}[Proof of \Cref{lem:dynamic pricing reduction}]
Fix an arbitrary single-item dynamic pricing problem instance
$\instance  = (\totaltime, \optval)$
such that there are $\totaltime$ rounds and 
each buyer has private value $\optval$.
Without loss of generality, we assume 
that $\optval \leq \frac{1}{2}$
and $\totaltime \geq 2$.
We first present the 
construction of 
the Bayesian recommendation
problem instance $\instance\primed$.
Then, given any online 
policy $\ALG\primed$ 
with binary signal space for 
the Bayesian recommendation instance $\instance\primed$,
we present the construction of dynamic pricing mechanism $\ALG$
for the original dynamic pricing instance $\instance$
with regret 
$\Reg[\instance]{\ALG} \leq \Reg[\instance\primed]{\ALG\primed} + 1$.

\xhdr{Construction of the Bayesian recommendation instance}
Consider the following 
Bayesian recommendation instance $\instance\primed$.\footnote{We 
use notation $\dagger$ and $\ddagger$ to denote the Bayesian recommendation instance.}
There are $\totalstate\primed = 2$ states, and 
$\totaltime\primed = \totaltime$ rounds.
Let $\epsilon = \sfrac{1}{\totaltime\primed}$.
State 1 is realized with probability 
$\prior\primed(1) = \eps$
and state 2 is realized with probability 
$\prior\primed(2) = 1 - \eps$.
The \Receiver s' utility is defined as follows,
\begin{align*}
    &\text{for state 1:}\quad 
    \RUtility\primed(1, \action\primed) = 
    \indicator{\action\primed = 1} \\
    &\text{for state 2:}\quad 
    \RUtility\primed(2, \action\primed) 
    = -\frac{\epsilon}{\optval} \cdot 
    \indicator{\action\primed = 1} 
\end{align*}
For notation simplicity, in this section, 
we use $\costState(\realizedstatei)$ to denote $\costStatei$.
By construction, $\costState\primed(1) = \eps$,
$\costState\primed(2) = - \frac{\epsilon(1-\epsilon)}{\optval}$,
and the optimal signaling in hindsight $\optoffPrimed$
satisfies that 
$\optoffPrimed(1) = 1$, $\optoffPrimed(2) = \frac{\optval}{1-\epsilon}$,
and $\SExpUtility(\optoffPrimed)=\optval + \eps$.

\xhdr{Construction of the 
dynamic pricing mechanism}
Fix an arbitrary online policy $\ALG\primed$ 
with binary signal space 
for 
the Bayesian recommendation instance $\instance\primed$.
Here we show that there exists a dynamic pricing mechanism $\ALG$
with regret 
$\Reg[\instance]{\ALG} \leq \Reg[\instance\primed]{\ALG\primed}
+ 1$.
Our argument contains two steps. 
First, we show that every online policy $\ALG\primed$
can be converted into an online policy $\ALG\doubleprimed$
within a subclass,
which has weakly smaller regret.
Then, we show how to construct a dynamic pricing mechanism $\ALG$ 
based on online policy $\ALG\doubleprimed$.

\textbullet\ \textbf{[Step I]} 
Notice that in the construction of 
the Bayesian recommendation
instance $\instance\primed$,
\Receiver s prefer action 1 in state 1
and action 0 in state 2. 
Thus, in the optimal signaling scheme
in hindsight $\optoffPrimed$,
the threshold state is state 2
and state 1 is above state 2.
Intuitively speaking, it is natural to consider 
a subclass of signaling schemes $\schemes\doubleprimed$ 
such that 
for every signaling scheme $\scheme\doubleprimed
\in\schemes\doubleprimed$,
it
issues signal $\signal\doubleprimed = 1$
(i.e., recommends action 1)
deterministically 
(i.e., $\scheme\doubleprimed(1) = 1$)
when the state is 1,
and 
issues signal $\signal\doubleprimed = 1$ 
with probability $\scheme\doubleprimed(2)$
when the state is 2.
Following this intuition, below we show that 
for any online policy $\ALG\primed$
(with binary signal space),
there exists an online policy $\ALG\doubleprimed$
which only uses signaling schemes in $\schemes\doubleprimed$
and achieves weakly smaller regret.

Suppose in round $t$, signaling scheme $\schemei\primed$ is used in $\ALG\primed$.
Recall $\posteriori\primed(\signal\primed,\realizedstatei)$ 
is the posterior probability 
$\prob{\statei\primed = \realizedstatei\condition \signal\primed}$
under $\schemei\primed$.
Without loss of generality, we assume that 
$\posterior\primed(1,2)\leq \prior\primed(2) 
\leq \posterior\primed(0, 2)$.
Since $\costState\primed(1) + \costState\primed(2) < 0$,
\Receiver\ $t$ must take action 0
when the realized signal $\signal\primed = 0$.
Thus, the regret under $\schemei\primed$ is
\begin{align*}
    \Reg[\instance\primed]{\schemei\primed} 
    =&~\SExpUtility(\optoffPrimed) - 
    (\prior\primed(1)\schemei\primed(1) + 
    \prior\primed(2)\schemei\primed(2))
    \cdot
    \indicator{\text{\Receiver\ $t$ takes action 1}\condition\signal\primed=1}
    \\
    =&~
    \SExpUtility(\optoffPrimed) 
    -
    (\prior\primed(1)\schemei\primed(1) + 
    \prior\primed(2)\schemei\primed(2))
    \cdot 
    \indicator{\text{\Receiver\ $t$ takes action 1}\condition
    \text{posterior belief is $\posterior\primed(1,\cdot)$}}
\end{align*}
Thus, in round $t$,  
$\ALG\doubleprimed$ can use signaling scheme $\schemei\doubleprimed$ such that 
$\schemei\doubleprimed(1) \triangleq 1$ 
and $\schemei\doubleprimed(2) \triangleq 
\sfrac{\schemei\primed(2)}{\schemei\primed(1)}$.
Since 
$\posterior\primed(1,2)\leq \prior\primed(2)$
and $\posterior\primed(1,2) = 
\sfrac{(\prior\primed(2)\schemei\primed(2))}{
(\prior\primed(1)\schemei\primed(1) 
+ \prior\primed(2)\schemei\primed(2))}$,
we have $\schemei\doubleprimed(2) 
~(= 
\sfrac{\schemei\primed(2)}{\schemei\primed(1)}) \leq 1$
and thus $\schemei\doubleprimed$ is well-defined.
Additionally, by construction, posterior belief 
$\posteriori\doubleprimed(1,\cdot)$ under signal 1
in $\schemei\doubleprimed$ is the same as 
posterior belief $\posteriori\primed(1,\cdot)$
under signal 1 in $\schemei\primed$. 
Thus, 
\begin{align*}
\Reg[\instance\primed]{\schemei\doubleprimed} 
    =&~\SExpUtility(\optoffPrimed) - 
    (\prior\primed(1)\schemei\doubleprimed(1) + 
    \prior\primed(2)\schemei\doubleprimed(2))
    \cdot
    \indicator{\text{\Receiver\ $t$ takes action 1}\condition\signal\doubleprimed=1}
    \\
    =&~
    \SExpUtility(\optoffPrimed) 
    -
    (\prior\primed(1)\schemei\doubleprimed(1) + 
    \prior\primed(2)\schemei\doubleprimed(2))
    \cdot 
    \indicator{\text{\Receiver\ $t$ takes action 1}\condition
    \text{posterior belief is $\posterior\doubleprimed(1,\cdot)$}}
    \\
    \leq&~ \Reg[\instance\primed]{\schemei\primed}
\end{align*}
Therefore, 
suppose
$\ALG\primed$ uses signaling scheme $\schemei\primed$
in round $t$.
$\ALG\doubleprimed$ can mimic $\ALG\primed$ by using 
signaling scheme $\schemei\doubleprimed$ defined above
and suffers a weakly smaller expected regret.
Though posterior belief 
$\posteriori\doubleprimed(0,\cdot)$ 
under signal 0
in
$\schemei\doubleprimed$ 
may not equal posterior belief $\posteriori\primed(0,\cdot)$ 
under signal 0 in 
$\schemei\primed$,
\Receiver\ $t$ must take action 0 under both 
$\posteriori\doubleprimed(0,\cdot)$ and $\posteriori\primed(0,\cdot)$.
Thus, $\ALG\doubleprimed$ has more information than 
$\ALG\primed$ and can continue mimicking $\ALG\primed$
in the future rounds.

\textbullet\ \textbf{[Step II]} Here we show that 
given any online policy $\ALG\doubleprimed$ 
which only uses signaling scheme 
in $\schemes\doubleprimed$ defined in [Step I],
we can construct a dynamic pricing mechanism $\ALG$ with 
regret $\Reg[\instance]{\ALG} \leq
\Reg[\instance\primed]{\ALG\doubleprimed}
+ 1$.

Suppose in round $t$, signaling scheme $\schemei\doubleprimed$ is used in $\ALG\doubleprimed$.
By the definition of $\ALG\doubleprimed$,
we know that $\schemei\doubleprimed(1) = 1$.
Thus, \Receiver\ $t$ takes action 1
if and only if the realized signal $\signal\doubleprimed = 1$
and his expected utility of taking action 1 is better than 
taking action 0 under his 
posterior belief, i.e.,
\begin{align*}
    \costState\primed(1)\,\schemei\doubleprimed(1)
    +
    \costState\primed(2)\,\schemei\doubleprimed(2) \geq 0
    \quad\Rightarrow
    \quad
    \schemei\doubleprimed(2) \leq \frac{\optval}{1-\eps}
\end{align*}
Hence, the expected regret induced by signaling scheme 
$\schemei\doubleprimed$ is 
\begin{align*}
    \Reg[\instance\primed]{\schemei\doubleprimed} 
    =&~
    \SExpUtility(\optoffPrimed) 
    -
    (\prior\primed(1)\schemei\doubleprimed(1) + 
    \prior\primed(2)\schemei\doubleprimed(2))
    \cdot
    \indicator{\text{\Receiver\ $t$ takes action 1}\condition\signal\doubleprimed=1}
    \\
    =&~
    \optval + \eps -
    \left(\eps + (1-\eps)\schemei\doubleprimed(2)\right)
    \cdot 
    \indicator{\schemei\doubleprimed(2) \leq 
    \frac{\optval}{1-\eps}}
\end{align*}
Therefore, 
suppose
online policy $\ALG\doubleprimed$ 
uses signaling scheme $\schemei\doubleprimed$
in round $t$ for the Bayesian recommendation instance $\instance\primed$.
We can construct the following dynamic pricing mechanism 
$\ALG$ which posts price $\pricei \triangleq
(1-\eps)\schemei\doubleprimed(2)$
in round $t$ for the dynamic pricing instance $\instance$.
The regret of posting price $\pricei$ is 
\begin{align*}
    \Reg[\instance]{\pricei} = \optval - 
    \pricei \cdot \indicator{\pricei \leq \optval}
    \leq 
    \Reg[\instance\primed]{\schemei\doubleprimed} + \eps
\end{align*}
Since dynamic pricing mechanism $\ALG$ 
has more information than 
online policy $\ALG\doubleprimed$,\footnote{In particular,
dynamic pricing mechanism deterministically learns
whether $\pricei \leq \optval$ (a.k.a., 
$ \indicator{\schemei\doubleprimed(2) \leq 
    \frac{\optval}{1-\eps}}$),
    while online policy $\ALG\doubleprimed$
    only learns this information when signal 1 is realized.}
$\ALG$ can simulate $\ALG\doubleprimed$ in the future rounds.
The total regret is 
\begin{align*}
    \Reg[\instance]{\ALG}
    -
    \Reg[\instance\primed]{\ALG\doubleprimed}
    =
    \sum_{t\in[\totaltime]}
    \left(
    \Reg[\instance]{\pricei}
    -
    \Reg[\instance\primed]{\schemei\doubleprimed}
    \right)
    \leq 
    \eps\cdot \totaltime = 1
\end{align*}
which concludes the proof.
\end{proof} 

\section{Missing Technical Details and Proofs in Section~\ref{sec_generalization}}
\label{apx:extension}

\newcommand{\cateSet}{\mathcal{N}}

In this section, we present the missing technical details and formal analysis for our extension models.

\xhdr{Platform with state-dependent utility function}
In the modified version of \Cref{alg:loglog T},
variables $\SExpUtilityUnderbar, L, R$ now denote 
the lower bound and upper bound of $\SExpUtility(\optoff) = 
\sum_{\realizedstatei\in[\totalstate]}\prior(\realizedstatei)
\SUtility(\realizedstatei,1)\optoff(\realizedstatei)$.
Each direct signaling scheme $\scheme^{(\permu,u)}$
used in the exploring phase 
is a signaling scheme such that 
there exists a threshold state $\realizedstatei\primed$
(a) for every state
$\realizedstatei\not=\realizedstatei\primed$,
$\scheme^{(\permu,u)}(\realizedstatei)
= 
\indicator{\permu(\realizedstatei) >\permu(\realizedstatei\primed)}$,
and (b) 
$\scheme^{(\permu,u)}(\realizedstatei\primed)
=
\frac{
u - \sum_{\realizedstatei\not=\realizedstatei\primed}
\prior(\realizedstatei)
\SUtility(\realizedstatei,1)
\scheme^{(\permu,u)}(\realizedstatei)
}{\prior(\realizedstatei\primed)\SUtility(\realizedstatei,1)}$.
All other parts of \Cref{alg:loglog T} remain the same.

The modification of \Cref{alg:log T} is similar.
Variable $\SExpUtilityUnderbar$
now denotes the lower bound of $\SExpUtility(\optoff)$.
In the construction of signaling scheme $\schemeanchorpair$,
$\schemeanchorpairk$, and $\schemestartingpoint$,
holding everything else the same as before,
we modify $\schemeanchorpair(\realizedstateanchorj) = 
\frac{\SExpUtilityUnderbar}{\prior(\realizedstateanchorj)\SUtility(\realizedstateanchorj,1)}$,
$\schemeanchorpairk(\realizedstateanchorj) = 
\frac{\SExpUtilityUnderbar}{2\prior(\realizedstateanchorj)\SUtility(\realizedstateanchorj,1)}$,
and 
$\schemestartingpoint(\realizedstateanchorj) = 
\frac{\SExpUtilityUnderbar}{8\prior(\realizedstateanchorj)\SUtility(\realizedstateanchorj,1)}$.
All other parts of \Cref{alg:log T} remain the same.

We note that the state-dependent utility extension also captures the scenario where the platform is a benevolent player that aims to maximize the social welfare. For example, we can define a net utility function $\netSUtility(i, a) = \SUtility(i, a) + \RUtility(i, a)$ where $\SUtility(\cdot, \cdot), \RUtility(\cdot, \cdot)$ are the utility function of the platform and the user, respectively.
Then given a direct and \IC\ signaling scheme $\scheme$, we can write the platform's expected net utility as $U(\scheme) = \sum_{i\in[\totalstate]} \prior(i)\cdot \left(\scheme(i, 1) \netSUtility(i, 1) + \scheme(i, 0)\netSUtility(i, 0)\right) =
\sum_{i\in[\totalstate]} \prior(i)\cdot \scheme(i, 1) \left(\netSUtility(i, 1) - \netSUtility(i, 0)\right) + \sum_{i\in[m]} \prior(i) \netSUtility(i, 0)$. 
When $\SUtility(i, a) \equiv a$ (i.e., our main setting),  we redefine $\newnetSUtility(i, 1) \triangleq \netSUtility(i, 1) - \netSUtility(i, 0)$, and it is easy to see $\newnetSUtility(i, 1) = 1 + \RUtility(i, 1) - \RUtility(i, 0) \ge 0$ provided that we normalize $\RUtility(i, 1) - \RUtility(i, 0)\in[-1, 1]$. 
Thus, we can exactly recover the state-dependent utility extension that we just mentioned.

\xhdr{Users with misspecified beliefs}
In this setting, the optimal signaling scheme in hindsight $\optoff$ can be solved by the following linear program,
\begin{align*}
    \begin{array}{llll}
      \optoff = \argmax\limits_{\scheme}   &  
      \displaystyle\sum_{\realizedstatei\in[\totalstate]}
      \prior(\realizedstatei)\scheme(\realizedstatei,1)
      &  \text{s.t.}
      &\\
& \displaystyle\sum_{\realizedstatei\in[\totalstate]}
         \left(\RUtility(\realizedstatei,1)
         -
         \RUtility(\realizedstatei,0)\right)
         \prior\primed(\realizedstatei)
         \scheme(\realizedstatei,1)
         \geq 0
         \quad
         & & \\
         &\scheme(\realizedstatei, 1)
         +
         \scheme(\realizedstatei, 0) = 1
         & \realizedstatei\in[\totalstate]
         &\\
         &\scheme(\realizedstatei, 1)
         \geq 0, 
         \scheme(\realizedstatei, 0) \geq 0
         &
         \realizedstatei\in[\totalstate]
         &
    \end{array}
\end{align*}
By rewriting $\left(\RUtility(\realizedstatei,1)
         -
         \RUtility(\realizedstatei,0)\right)
         \prior\primed(\realizedstatei)$
as
$\costState(\realizedstatei)$,
we can observe that this is equivalent to\footnote{Specifically, 
there exists a bijection between the problem instances
in both problems,
such that the optimal signaling schemes in hindsight
and the corresponding utilities of the platform
are identical.
} the original 
Bayesian recommendation problem
solved by \Cref{alg:loglog T} and \Cref{alg:log T},
where the platform and users share the same beliefs.
Moreover, it can be verified that the regret guarantees
for both online policies continue to hold in this extension.

\xhdr{Which video to display}
We formally describe the extension as follows. There are $\videoNum$ categories of videos, and the video in each category $k\in [\videoNum]$ has $\totalstate_k\in\N^+$ number of states.
For each category $k$, its video states are realized according to a prior distribution $\prior_k \in \stochastic([\totalstate_k])$, which is common knowledge to all players. 
For the videos in each category $k$, the user has a utility function $\RUtility_k:[\totalstate]\times \cA \rightarrow \reals$ mapping from the state of the video in the category $k$ and his action to his utility, and the platform has a state-independent utility function $\SUtility: \cA \rightarrow \reals$ that only depends on the user's action $\action\in\cA$.
At the beginning of each time round, the platform decides which category of videos to display to the users, and also decides how to recommend the video with different states in this category to the users (i.e., how to signal the video states to the user). 
The user knows which category of the video is displayed at each time round,\footnote{Typically, the thumbnail of the short video can easily tell which category of this video belongs to.} but is unsure about the exact video state. 
The platform's goal is to maximize cumulative payoff. This involves identifying the video category with the highest payoff and optimizing recommendations within that category to achieve the best results.

When the platform knows the user's ordinal preference for each video category, we can show that an extension of our algorithm \CRP\ can still guarantee an $O(\log\log T)$ expected regret. Similar analysis can be extended when the platform does not know the user's ordinal preference, which is omitted to avoid repetition.
\begin{restatable}{proposition}{propmulticategory}
\label{prop:multi category}
In the multi-video-category, there exists an algorithm (see \Cref{alg:loglog T category}) that can achieve expected regret at most $O(\videoNum\log\log T)$.
\end{restatable}
\begin{proof}
We first describe \Cref{alg:loglog T category} when the platform has access to choose which category that can be displayed to the users. 
The algorithm follows similar logic to \Cref{alg:loglog T w unknown order}. 

The \Cref{alg:loglog T category}
relies on a subclass of direct signaling schemes $\scheme^{(k,u)}$, but with an additional parameter 
$k$ to represent a possible category of the videos, which also defines user's preference differences $\{\RUtilityDiff_k(\realizedstatei)\}$ for this category.
In other words, for every category, one can identify a direct signaling scheme $\scheme^{(k,u)}$ such that 
the expected payoff to the platform when implementing this signaling scheme for this category will exactly equal $u$
if this signaling scheme $\scheme^{(k,u)}$ is {\em \IC}.
Meanwhile, whenever the algorithm identifies 
an \IC\ signaling scheme $\scheme^{(k,u)}$ (see Line $4$ and $15$ in \Cref{alg:loglog T category}), it then naturally gives a lower bound 
of the platform's expected payoff of the optimal
signaling scheme, namely, $\SExpUtility(\optoff) \ge u$ where here $\optoff$ denotes the optimal signaling scheme among all categories. 
Then one can use the payoff (i.e., $u$) of
such \IC\ signaling scheme to further prune out 
other signaling schemes that are constructed for other categories but are either non-\IC\ or have payoff less than $u$.
Similar to the argument of \Cref{alg:loglog T w unknown order}, we would like to note that without pruning,  the regret of \Cref{alg:loglog T category} might have linear dependence on the time horizon $\totaltime$. 

With \Cref{alg:loglog T category}, the proof of the regret bound then follows similar arguments to the one used in the proof of \Cref{alg:loglog T w unknown order} (See \Cref{apx:proofs unknwon order} for detailed arguments). 
\end{proof}

\begin{algorithm}[H]
\caption{Conservative Recommendation Policy for Multiple Categories}
\label{alg:loglog T category}
\linespread{0.8}\selectfont
\SetAlgoLined\DontPrintSemicolon
\KwIn{number of rounds $\totaltime$, number of video categories $\videoNum$, 
number of states $\totalstate_k$ and
prior distribution $\prior_k$ for each category $k\in[n]$.}
Initialize the set $\cateSet \triangleq [\videoNum]$ to contain all categories.\\
\tcc{exploring phase I
--
identify 
$\SExpUtilityUnderbar$
such that 
$\SExpUtilityUnderbar
\leq \SExpUtility(\optoff) 
\leq 2 \SExpUtilityUnderbar$}
Initialize $\SExpUtilityUnderbar \gets \frac{1}{2}$
\\
\While{\True
}{
    \If{there exists $k\in \cateSet$ such that 
    $\CheckPersu\left(\scheme^{(k,\SExpUtilityUnderbar)}\right) =
    \True$}{
        $\cateSet\gets\left\{k\in \cateSet:
        \CheckPersu\left(\scheme^{(k,\SExpUtilityUnderbar)}\right) = 
        \True
        \right\}$\\
        \cc{break}
    }
    \uElse{
    $\SExpUtilityUnderbar\gets \frac{\SExpUtilityUnderbar}{2}$
    }
}
\tcc{exploring phase II --
identify a signaling scheme $\scheme\primed$
such that $\SExpUtility(\scheme\primed) \geq \SExpUtility(\optoff) - \frac{1}{\totaltime}$}
Initialize
$R \gets 2\SExpUtilityUnderbar$,
 $L \gets \SExpUtilityUnderbar$,
$\stepperc \gets 1$
\\
\While{$R - L \geq \frac{1}{\totaltime}$
}{
    $\perc \gets \frac{\stepperc}{2}$,
    $S\gets \lfloor\frac{R - L}{\perc L}\rfloor$
    \\
    \For{$\ell = 1, 2, \dots, S$}{
        \If{there exists $k\in \cateSet$ such that 
        $\CheckPersu\left(\scheme^{(k,L + \ell\perc L)}\right) =
        \True$}{
            $\cateSet\gets \left\{k\in\cateSet:
            \CheckPersu\left(\scheme^{(k,L + \ell\perc L)}\right) =
            \True
        \right\}$
        }\uElse{
            $R\gets L + \ell\perc L$,
            $L \gets L + (\ell-1)\perc L$,
            $\stepperc \gets \perc^2$\\
            \cc{break}
        }
    }
}
Choose an arbitrary $k\in \cateSet$ and set $\scheme\primed \gets \scheme^{(k,L)}$.
\\
\tcc{exploiting phase}
Choose the category $k$ to display and 
use the signaling scheme $\scheme\primed$
for all remaining rounds.
\end{algorithm}

\section{Simulations}
\label{subsec:simulation}

\newcommand{\CRPNoBS}{\texttt{\CRP$^-$}}
\newcommand{\NoInfor}{\texttt{NoInfor}}
\newcommand{\FullInfor}{\texttt{FullInfor}}
In this section, we provide some insights from numerical experiments 
that test the empirical performance
of our proposed algorithm and 
highlight some of its salient features. 
The main goal of this section is to evaluate the 
performance of our main algorithm (i.e., \CRP). Thus, we will focus on the 
case where the user's ordinal preference is known 
to the platform, but the cardinal preference remains unknown to the platform.
\begin{figure}[h!]
    \centering
    \includegraphics[width=0.6\textwidth]{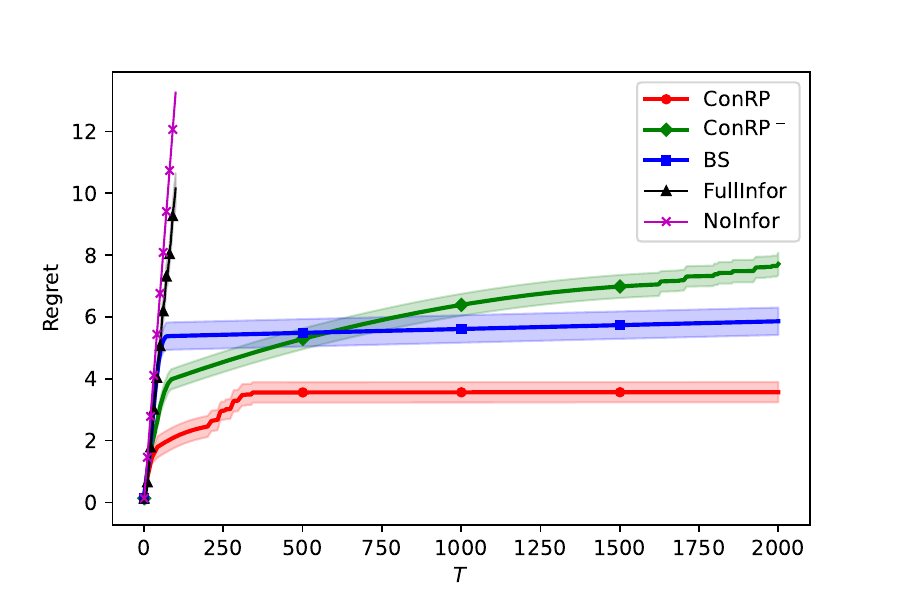}
    \caption{The above figure plots the regret growth with $T$ for various algorithms on a randomly generated instance with 
    $\totalstate = 50$, 
    and platform's optimal payoff is
    $\SExpUtility(\optoff) = 0.1326$. }
    \label{fig:algo comparison}
    \vspace{-10pt}
\end{figure}
We simulate an instance of the Bayesian recommendation problem with the number of states $\totalstate = 50$ and time horizon $T = 2000$,\footnote{\label{footnote:interaction number in practice}For context, real-world platforms often involve users making hundreds of sequential decisions per day: for example, TikTok users view roughly 350-600 videos daily \citep{Flixier,WG-22}, and Tinder users swipe through 140-200 profiles \citep{dating-app}. While these horizons are modest compared to the asymptotic regime in which our theoretical guarantees become tight, they motivate our choice to report cumulative regret over the first \(2{,}000\) rounds as a representative scale. Our results show that the proposed algorithm already performs well at this horizon, despite the idealized assumptions of our model.}
where the prior distribution $\prior$ is generated randomly in the space $\Delta([\totalstate])$, and the user's preference differences $\{\RUtilityDiff(i)\}$ are generated randomly from $\Unif[-2, 2]$.\footnote{Notice that, in our setting, 
the user's preference differences are the sufficient quantities to summarize user's behavior.}
And, we compute the average regret based on 
50 independent simulations from this randomly generated instance. 
In \Cref{fig:algo comparison}, we 
report the performance of the following algorithms:
\begin{enumerate}
    \item \CRP: This is our proposed algorithm 
    (\Cref{alg:loglog T} in \Cref{sec:knwon order}), which attains optimal $O(\log\log T)$ regret.
    \item \CRPNoBS: This is a modification of our \CRP\ with no exploring phase I (i.e., the binary-search steps) to lower bound the 
    platform's optimal payoff. 
    \item \texttt{BS}: This is a binary-search algorithm, and notice that since the user's ordinal preference is known to the platform, identifying the optimal signaling scheme is equivalent to identifying the platform's optimal payoff. Thus, this algorithm implements a 
    binary search to identify the value of platform's optimal payoff. 
    Particularly, this algorithm can utilize the user's response to determine whether the payoff of the current signaling scheme is lower/higher than the optimal payoff.
    \item \cc{FullInfor}: This is a naive algorithm that 
    keeps using a full-information signaling scheme.
    For the above generated random instance, the expected one-round payoff to the platform is $0.0304$.
    \item \NoInfor: This is a naive algorithm that 
    keeps using a no-information signaling scheme. 
    For the above generated random instance, the expected one-round payoff to the platform is $0$.   
\end{enumerate}
The last two baseline algorithms
$\FullInfor$ and $\NoInfor$ are used to 
demonstrate the usefulness of signaling in our Bayesian recommendation problem. 
It is expected that the performance of these two algorithms grows linearly with the time horizon $T$ (in \Cref{fig:algo comparison}, we only plot out the performance of these two algorithms for $T\le 100$ due to the limited margin of the presented figure).
As we can see in \Cref{fig:algo comparison},
the superior performance of our proposed \CRP\ demonstrates its effectiveness compared to other baseline algorithms. 
In particular, 
we observe that our \CRP\ 
performs significantly better than the binary-search algorithm $\texttt{BS}$, this is due to the asymmetric property of the payoff of the platform's signaling scheme (i.e.,  
the payoff $\SExpUtility(\scheme)$ of a particular signaling scheme $\scheme$ is 
$\SExpUtility(\scheme) = \sum_{i}\prior(i)\scheme(i) \cdot \indicator{\scheme \text{ is \IC}}$). 
Our \CRP\ also performs significantly better than the algorithm \CRPNoBS, and this significant improvement in performance is due to the additional exploring phase I to identify a 
lower bound of the platform's optimal payoff.

\end{document}